\documentclass[10pt]{article} 

\usepackage{amsfonts,amssymb,amsmath,stmaryrd,latexsym}
\usepackage{enumerate,array,xspace}

\usepackage[hidelinks]{hyperref}
\usepackage{graphicx}
\usepackage{multirow}

\usepackage[shortlabels]{enumitem}

\newenvironment{tbs}{%
   \small\tt
   \begin{itemize}}{\end{itemize}}
\newcommand{\btbs}{\begin{tbs}}                                                                      
\newcommand{\etbs}{\end{tbs}}

\newcommand{\hide}[1]{}

\newtheorem{theorem}{Theorem}[section]

\newtheorem{fact}[theorem]{Fact}
\newtheorem{proposition}[theorem]{Proposition}
\newtheorem{lemma}[theorem]{Lemma}
\newtheorem{corollary}[theorem]{Corollary}
\newtheorem{defi}[theorem]{Definition}

\newtheorem{conv}[theorem]{Convention}
\newtheorem{rema}[theorem]{Remark}
\newtheorem{exam}[theorem]{Example}

\newenvironment{definition}{\begin{defi}\rm}{\hfill $\lhd$\end{defi}}

\newenvironment{remark}{\begin{rema}\rm}{\hfill $\lhd$\end{rema}}
\newenvironment{example}{\begin{exam}\rm}{\hfill $\lhd$\end{exam}}

\newcommand{\qed}{\textsc{qed}}
\newenvironment{proof}{\begin{trivlist}\item[]{\bf
Proof.}}{\hfill {\sc qed}\end{trivlist}}

\newtheorem{claim2}{\sc Claim}

\newenvironment{proofof}[1]{\begin{trivlist}\item[\hskip\labelsep{\textsc{Proof
  of {#1}.\ }}]}{\hspace*{\fill} {\qed}\end{trivlist}}
\newenvironment{pfclaim}{\begin{trivlist}\item[]{\sc Proof of Claim.~}}
{\hfill {\mbox{$\blacktriangleleft$}}\end{trivlist}}

\newcommand{\Dom}{\mathsf{Dom}}
\newcommand{\Ran}{\mathsf{Ran}}
\newcommand{\nada}{\varnothing}
\newcommand{\rst}[1]{{\upharpoonright}_{#1}\,}
\newcommand{\isdef}{\mathrel{:=}}
\newcommand{\tup}[1]{\langle{#1}\rangle}
\newcommand{\pow}{\wp}
\newcommand{\sz}[1]{|#1|}

\newcommand{\LFP}{\mathit{LFP}}

\newcommand{\Prop}{\mathsf{Prop}}
  \newcommand{\pprop}{\mathsf{P}}
  \newcommand{\qprop}{\mathsf{Q}}
  \newcommand{\rprop}{\mathsf{R}}

\newcommand{\fovar}{\mathsf{iVar}}

\newcommand{\nxt}[1]{\raisebox{.3ex}{$\scriptstyle \bigcirc$}_{#1}}
\newcommand{\tmod}{\Diamond}
\newcommand{\qu}{\ensuremath{\exists^\infty}\xspace}
\newcommand{\dqu}{\ensuremath{\forall^\infty}\xspace}
\newcommand{\wqu}{\ensuremath{\mathbf{W}}\xspace}
\newcommand{\inc}{\sqsubseteq}
\newcommand{\foeq}{\approx}
\newcommand{\here}[1]{{\Downarrow}#1}
\newcommand{\arediff}[1]{\mathrm{diff}(#1)}

\newcommand{\dualop}{\delta}
\newcommand{\dual}[1]{#1^{\dualop}}
\newcommand{\ytr}{\mathtt{tr}}

\newcommand{\yvlang}{L}
\newcommand{\oslang}{L_{1}}

\newcommand{\muML}{\ensuremath{\mu\ML}\xspace}
\newcommand{\ML}{\ensuremath{\mathrm{ML}}\xspace}
\newcommand{\mucML}{\ensuremath{\mu_{C}\ML}\xspace}
\newcommand{\mudML}{\ensuremath{\mu_{D}\ML}\xspace}
   \newcommand{\AFMC}{\ensuremath{\mu_{D}\ML}\xspace}

\newcommand{\fo}{\ensuremath{\mathrm{FO}}\xspace}
\newcommand{\mso}{\ensuremath{\mathrm{MSO}}\xspace}
\newcommand{\wmso}{\ensuremath{\mathrm{WMSO}}\xspace}
\newcommand{\nmso}{\ensuremath{\mathrm{NMSO}}\xspace}
\newcommand{\smso}{\ensuremath{\mathrm{SMSO}}\xspace}

\newcommand{\ofo}{\ensuremath{\fo_1}\xspace}
\newcommand{\foe}{\ensuremath{\mathrm{FOE}}\xspace}
\newcommand{\foei}{\ensuremath{\mathrm{FOE}^{\infty}}\xspace}
\newcommand{\ofoe}{\ensuremath{\foe_1}\xspace}
\newcommand{\ofoei}{\ensuremath{\foei_1}\xspace}
\newcommand{\oml}{\ensuremath{\mathrm{ML}_1}\xspace}
\newcommand{\owmso}{\ensuremath{\mathrm{WMSO}_1}\xspace}

\newcommand{\cont}[2]{\mathtt{Con}_{#2}(#1)}
\newcommand{\cocont}[2]{\mathtt{CoCon}_{#2}(#1)}
\newcommand{\noe}[2]{\mathtt{Noe}_{#2}(#1)}
\newcommand{\conoe}[2]{\mathtt{CoNoe}_{#2}(#1)}

\newcommand{\dbnf}{\nabla}

\newcommand{\dbnfofoe}[2]{\dbnf_{\foe}(#1,#2)}

\newcommand{\mondbnfofoei}[4]{\dbnf^{#4}_{\foei}(#1,#2,#3)}

\newcommand{\posdgbnfofo}[2]{\dbnf^+_{\fo}(#1,#2)}
\newcommand{\posdbnfofoe}[2]{\dbnf^+_{\foe}(#1,#2)}
\newcommand{\posdbnfofoei}[3]{\dbnf^+_{\foei}(#1,#2,#3)}
\newcommand{\posdbnfinf}[1]{\dbnf^+_\infty(#1)}

\newcommand{\FV}{\mathit{FV}}
\newcommand{\BV}{\mathit{BV}}
\newcommand{\Sfor}{\mathit{Sfor}}
\newcommand{\subf}{\trianglelefteq}

\newcommand{\isbnf}{\mathrel{::=}}

\newcommand{\mathstr}[1]{\mathbb{#1}}
\newcommand{\bbA}{\mathstr{A}}

\newcommand{\bbS}{\mathstr{S}}
\newcommand{\bbT}{\mathstr{T}}
\newcommand{\nat}{\ensuremath{\mathbb{N}}\xspace}

\newcommand{\tscolors}{\kappa}
\newcommand{\tsval}{\tscolors^\natural}

\newcommand{\Aut}{\mathit{Aut}}
\newcommand{\AutW}{\mathit{Aut}_{w}}

\newcommand{\AutWC}{\mathit{Aut}_{wc}}

\newcommand{\ord}{\prec}
\newcommand{\tmap}{\Delta}
\newcommand{\tmapProj}{\tmap^\exists}
\newcommand{\pmap}{\Omega}

\newcommand{\egame}{\mathcal{E}}
\newcommand{\agame}{\mathcal{A}}
\newcommand{\UG}{\mathcal{U}}

\newcommand{\strat}{f}
   \newcommand{\ystrat}{\strat}

\newcommand{\eloise}{\ensuremath{\exists}\xspace}
\newcommand{\abelard}{\ensuremath{\forall}\xspace}
\newcommand{\pmatches}[2]{\mathsf{PM}^{#1}_{#2}}
\newcommand{\win}{\text{\sl Win}}
\newcommand{\last}{\textsf{last}}

\newcommand{\unravel}[1]{\hat{#1}}
\newcommand{\omegaunrav}[1]{{#1}^\omega}
\newcommand{\bis}{\mathrel{\underline{\leftrightarrow}}}

\newcommand{\ext}[1]{\llbracket#1\rrbracket}
\newcommand{\mmodels}{\Vdash}
\newcommand{\Mod}{\ensuremath{\mathsf{Mod}}}
\newcommand{\TMod}{\ensuremath{\mathsf{TMod}}}

\newcommand{\ol}[1]{\overline{#1}}
\newcommand{\sse}{\subseteq}
\newcommand{\vlist}[1]{\overline{{\mathbf{#1}}}}
\newcommand{\mc}{\mathcal}

\renewcommand{\phi}{\varphi}

\newcommand{\lift}[1]{#1^{\Uparrow}}

\newcommand{\shA}{\pow A}

\newcommand{\f}{F} 

\renewcommand{\sc}{\scshape}
\newenvironment{Iff-RL}{\textbf{($\Rightarrow$)} }{\bigskip}
\newenvironment{Iff-LR}{\textbf{($\Leftarrow$)} }{}
\newcommand{\mb}{\mathbb}
\def \: {\colon}
\def \p {\pow}

\newcommand{\shai}{a_{I}^{\sharp}}
\newcommand{\shDe}{\Delta^{\sharp}}
\newcommand{\bmDe}{\Delta^{\flat}}

\newcommand{\V}{\tscolors}
\newcommand{\R}[1]{R[#1]}
\newcommand{\MSO}{\mso}

\newcommand{\noet}{{\scriptscriptstyle N}}
\newcommand{\fin}{{\scriptscriptstyle F}}

\newcommand{\noetexists}{\exists_{\noet}}
\newcommand{\finexists}{\exists_{\fin}}
\newcommand{\df}{\mathrel{: =}}

\newcommand{\lque}{\ensuremath{{\foe}^\infty}\xspace}
\newcommand{\posdbnfolque}[3]{\dbnf^+_{\lque}(#1,#2,#3)}		

\newcommand{\om}{\omega}
\newcommand{\al}{\alpha}

\title{The Power of the Weak}

\author{
Facundo Carreiro
\thanks{Institute for Logic, Language and Computation, Universiteit van Amsterdam,
   P.O. Box 94242, 1090 GE Amsterdam. E-mail: \url{contact@facundo.io}.}
\and Alessandro Facchini\thanks{Dalle Molle Institute for Artificial Intelligence (IDSIA),
Galleria 2, 6928 Manno (Lugano), Switzerland. E-mail: \url{alessandro.facchini@idsia.ch}.}
\and Yde Venema
\thanks{Institute for Logic, Language and Computation, Universiteit van Amsterdam,
   P.O. Box 94242, 1090 GE Amsterdam. E-mail: \url{y.venema@uva.nl}.}
\and Fabio Zanasi
\thanks{University College London,
66-72 Gower Street, London, United Kingdom. E-mail: \url{f.zanasi@ucl.ac.uk}.}
}

\begin{document}

\maketitle

\begin{abstract}
A landmark result in the study of logics for formal verification is Janin \& Walukiewicz's theorem, stating that the modal $\mu$-calculus ($\muML$) is equivalent modulo bisimilarity to standard monadic second-order logic (here abbreviated as $\smso$), over the class of labelled transition systems (LTSs for short). Our work proves two results of the same kind, one for the alternation-free fragment of $\muML$ ($\mudML$) and one for weak $\mso$ ($\wmso$). Whereas it was known that $\mudML$ and $\wmso$ are equivalent modulo bisimilarity on binary trees, our analysis shows that the picture radically changes once we reason over arbitrary LTSs. The first theorem that we prove is that, over LTSs, $\mudML$ is equivalent modulo bisimilarity to \emph{noetherian} $\mso$ ($\nmso$), a newly introduced variant of $\smso$ where second-order quantification ranges over ``well-founded'' subsets only. Our second theorem starts from $\wmso$, and proves it equivalent modulo bisimilarity to a fragment of $\mudML$ defined by a notion of continuity. Analogously to Janin \& Walukiewicz's result, our proofs are automata-theoretic in nature: as another contribution, we introduce classes of parity automata characterising the expressiveness of $\wmso$ and $\nmso$ (on tree models) and of $\mucML$ and 
$\mudML$ (for all transition systems).
\end{abstract}


\newpage


\section{Introduction}
   \label{sec:intro}

\subsection{Expressiveness modulo bisimilarity}

A seminal result in the theory of modal logic is van Benthem's Characterisation
Theorem~\cite{vanBenthemPhD}, stating that, over  the class of all labelled transition systems (LTSs for short),  every bisimulation-invariant
first-order formula is equivalent to (the standard
translation of) a modal formula:
\begin{equation}
\label{eq:vB}
\ML \equiv \fo/{\bis}  \qquad  \text{ (over the class of all LTSs)}.
\end{equation}
Over the years, a wealth of variants of the Characterisation Theorem have been
obtained.
For instance, van Benthem's theorem is one of the few
preservation results that transfers to the setting of finite
models~\cite{rose:moda97}; for a recent, rich source of van Benthem-style
characterisation results, see~\cite{DawarO09}. 
The general pattern of these results takes the shape
\begin{equation}
\label{eq:vBGeneral}
M \equiv \yvlang/{\bis}  \qquad  
\text{ (over a class of models $\mathsf{C}$)}.
\end{equation}
Apart from their obvious relevance to model theory, the interest in these 
results increases if $\mathsf{C}$ consists of transition structures that 
represent certain computational processes, as in the theory of the formal 
specification and verification of properties of software.
In this context, one often takes the point of view that bisimilar models 
represent \emph{the same} process. 
For this reason, only bisimulation-invariant properties are relevant.
Seen in this light, \eqref{eq:vBGeneral} is an \emph{expressive completeness} 
result: all the relevant properties expressible in $\yvlang$ (which is generally
some rich yardstick formalism), can already be expressed in a (usually
computationally more feasible) fragment $M$.

Of special interest to us is the work~\cite{Jan96}, which extends van Benthem's
result to the setting of \emph{second-order} logic, by proving that the 
bisimulation-invariant fragment of standard monadic second-order logic 
($\smso$) is the \emph{modal $\mu$-calculus} ($\muML$), viz., the extension of
basic modal logic with least- and greatest fixpoint operators:
\begin{equation}
\label{eq:JW}
\muML \equiv \smso/{\bis}  \qquad  \text{ (over the class of all LTSs)}.
\end{equation}
The aim of this paper is to study the fine structure of such connections between
second-order logics and modal $\mu$-calculi, obtaining variations of the
expressiveness completeness results \eqref{eq:vB} and \eqref{eq:JW}.
Our departure point is the following result, from \cite{ArnNiw92}.
\begin{equation}
\label{eq:weakfinite}
\AFMC \equiv \wmso/{\bis}  \qquad  \text{ (over the class of all binary trees)}.
\end{equation}
It relates two variants of $\muML$ and standard $\mso$ respectively. 
The first, $\AFMC$ is the \emph{alternation-free} fragment of $\muML$: the
constraint means that variables bound by least fixpoint operators cannot occur
in subformulas ``controlled'' by a greatest fixpoint operator, and viceversa. 
The fact that our best model-checking algorithms for $\muML$ are exponential in
the alternation depth of fixpoint operators 
\cite{EmersonL86,DBLP:conf/cav/LongBCJM94} makes $\AFMC$ of special interest for
applications. 
The second-order formalism appearing in \eqref{eq:weakfinite} is \emph{weak}
monadic second-order logic (see e.g. \cite[Ch. 3]{ALG02}), a variant of standard
$\mso$ where second-order quantification ranges over finite sets only. 
The interest in $\wmso$ is also justified by applications in software 
verification, in which it often makes sense to only consider finite portions of 
the domain of interest. 
\medskip

Equation \eqref{eq:weakfinite} offers only a very narrow comparison of the
expressiveness of these two logics. 
The equivalence is stated on binary trees, whereas \eqref{eq:vB} and
\eqref{eq:JW} work at the level of arbitrary LTSs. 
In fact, it turns out that the picture in the more general setting is far more
subtle. 
First of all, we know that $\AFMC \not\equiv \wmso/{\bis}$ already on (arbitrary) trees, because the class of well-founded trees, definable by the formula $\mu x.\square x$ of $\AFMC$, is not $\wmso$-definable. 
Moreover, whereas $\wmso$ is a fragment of $\smso$ on binary trees ---in fact, 
on all finitely branching trees--- as soon as we allow for infinite branching 
the two logics turn out to have \emph{incomparable} expressive power, 
see~\cite{CateF11,Zanasi:Thesis:2012}. 
For instance, the property ``each node has finitely many successors'' is 
expressible in $\wmso$ but not in $\smso$.

It is thus the main question of this work to clarify the status of 
$\wmso/{\bis}$ and $\AFMC$ on arbitrary LTSs. 
We shall prove that, in this more general setting, \eqref{eq:weakfinite}
``splits'' into the following two results, introducing in the picture a new
modal logic $\mucML$ and a new second-order logic $\nmso$.

\begin{theorem}~
   \label{t:11}
\begin{eqnarray}
   \mucML \equiv \wmso/{\bis} \qquad  \text{ (over the class of all LTSs)}. 
\label{eq:mucML=wmso} 
\\ \AFMC \equiv \nmso/{\bis}  \qquad \text{ (over the class of all LTSs)}. 
\label{eq:afmc=nmso}
\end{eqnarray}
\end{theorem}
For the first result \eqref{eq:mucML=wmso}, our strategy is to start from 
$\wmso$ and seek a suitable modal fixpoint logic characterising its 
bisimulation-invariant fragment. 
Second-order quantification $\exists p.\phi$ in $\wmso$ requires $p$ to be 
interpreted over a finite subset of an LTS. 
We identify a notion of \emph{continuity} as the modal counterpart of this 
constraint, and call the resulting logic $\mucML$, the \emph{continuous} 
$\mu$-calculus. 
This is a fragment of $\AFMC$, defined by the same grammar as the full $\muML$,
\begin{equation*}
\phi\ ::= q \mid \neg\phi \mid 
    \phi \lor \phi \mid  \Diamond \phi \mid
    \mu p.\phi' 
\end{equation*}
with the difference that $\phi'$ does not just need to be positive in $p$, but 
also continuous in $p$.
This terminology refers to the fact that $\phi'$ is interpreted by a function 
that is continuous with respect to the Scott topology; as we shall see, 
$p$-continuity can be given a \emph{syntactic} characterisation, as a certain
fragment $\cont{\muML}{p}$ of $\muML$--- see also~\cite{Fontaine08,FV12} for
more details. 

For our second result \eqref{eq:afmc=nmso}, we move in the opposite direction.
That is, we look for a natural second-order logic of which $\AFMC$ is the 
bisimulation-invariant fragment. 
Symmetrically to the case \eqref{eq:mucML=wmso} of $\wmso$ and continuity, a 
crucial aspect is to identify which constraint on second-order quantification
corresponds to the constraint on fixpoint alternation expressed by $\AFMC$. 
Our analysis stems from the observation that, when a formula $\mu p.\phi$ of
$\AFMC$ is satisfied in a tree model $\bbT$, the interpretation of $p$ must be
a subset of a \emph{well-founded} subtree of $\bbT$, because alternation-freedom
prevents $p$ from occurring in a $\nu$-subformula of $\phi$. 
We introduce the concept of \emph{noetherian} subset as a generalisation of 
this property from trees to arbitrary LTSs: intuitively, a subset of a LTS 
$\bbS$ is called noetherian if it is a subset of a bundle of paths that does
not contain any infinite ascending chain. 
The logic $\nmso$ appearing in \eqref{eq:afmc=nmso}, which we call 
\emph{noetherian} second-order logic, is the variant of $\mso$ restricting 
second-order quantification to noetherian subsets.

A unifying perspective over these results can be given through the lens of
K\"onig's lemma, saying that a subset of a tree $\bbT$ is finite precisely 
when it is included in a subtree of $\bbT$ which is both finitely branching 
and well-founded. In other words, finiteness on trees has two components, a
\emph{horizontal} (finite branching) and a \emph{vertical} (well-foundedness)
dimension. 
The bound imposed by $\nmso$-quantification acts only on the \emph{vertical}
dimension, whereas $\wmso$-quantification acts on both. 
It then comes at no surprise that \eqref{eq:mucML=wmso}-\eqref{eq:afmc=nmso}
collapse to \eqref{eq:weakfinite} on binary trees.
The restriction to finitely branching models nullifies the difference between
noetherian and finite, equating $\wmso$ and $\nmso$ (and thus also $\AFMC$ and
$\mucML$).


Another interesting observation concerns the relative expressive power of 
$\wmso$ with respect to standard $\mso$. 
As mentioned above, $\wmso$ is \emph{not} strictly weaker than $\smso$ on 
arbitrary LTSs. 
Nonetheless, putting together \eqref{eq:JW} and \eqref{eq:mucML=wmso} reveals 
that $\wmso$ collapses within the boundaries of $\smso$-expressiveness when it 
comes to bisimulation-invariant formulas, because $\mucML$ is strictly weaker
than $\muML$. 
In fact, modulo bisimilarity, $\wmso$ turns out to be even weaker than $\nmso$, 
as $\mucML$ is also a fragment of $\AFMC$. 
In a sense, this new landscape of results tells us that the feature 
distinguishing $\wmso$ from $\smso$/$\nmso$, \emph{viz.} the ability of 
expressing cardinality properties of the horizontal dimension of models, 
disappears once we focus on the bisimulation-invariant part, and thus is not
computationally relevant.

\subsection{Automata-theoretic characterisations}

Janin \& Walukiewicz's proof of \eqref{eq:JW} passes through a characterisation 
of the two logics involved in terms of \emph{parity automata}.
In a nutshell, a parity automaton $\bbA = \tup{A,\tmap,\pmap,a_I}$ processes
LTSs as inputs, according to a transition function $\tmap$ defined in terms of
a so-called \emph{one-step logic} $\oslang(A)$, where the states $A$ of 
$\bbA$ may occur as unary predicates. 
The map $\pmap \colon A \to \mathbb{N}$ assigns to each state a \emph{priority};
if the least priority value occurring infinitely often during the computation is 
even, the input is accepted.
Both $\smso$ and $\muML$ are characterised by classes of parity automata: what 
changes is just the one-step logic, which is, respectively, first-order logic 
with ($\ofoe$) and without ($\ofo$) equality. 
\begin{eqnarray}
\smso & \equiv & \Aut(\ofoe)
 \qquad \text{ (over the class of all trees)}, \label{eq:AutCharMSO}
\\ \muML & \equiv & \Aut(\ofo)
  \qquad \text{ (over the class of all LTSs)}. \label{eq:AutCharMuML}
\end{eqnarray}
This kind of automata-theoretic characterisation, which we believe is of independent interest, also underpins our two correspondence results. As the second main contribution of this paper, we
introduce new classes of parity automata that exactly capture the expressive
power of the second-order languages $\wmso$ and $\nmso$ (over tree models), 
and of the modal languages $\mucML$ and $\mudML$ (over arbitrary models).

Let us start from the simpler case, that is $\nmso$ and $\mudML$.
As mentioned above, the leading intuition for these logics is that they are
constrained in what can be expressed about the \emph{vertical} dimension of 
models. 
In automata-theoretic terms, we translate this constraint into the requirement 
that runs of an automaton can see at most one parity infinitely often: this 
yields the class of so-called \emph{weak} parity 
automata \cite{MullerSaoudiSchupp92}, which we write $\AutW(\oslang)$ for a
given one-step logic $\oslang$. \footnote{%
    Interestingly, \cite{MullerSaoudiSchupp92} introduces the class
    $\AutW(\ofoe)$ in order to shows that it characterises $\wmso$ on binary 
    trees, whence the name of \emph{weak} automata. 
    As discussed above, this correspondence is an ``optical illusion'', due to
    the restricted class of models that are considered, on which $\nmso = 
    \wmso$.
    } 
We shall show:
\begin{theorem}
\begin{eqnarray}
\nmso & \equiv & \AutW(\ofoe)
  \qquad \text{ (over the class of all trees)}, \label{eq:AutCharNmso}
\\ \mudML & \equiv & \AutW(\ofo)\label{eq:AutCharMudML}
  \qquad \text{ (over the class of all LTSs)}.
\end{eqnarray}
\end{theorem}
It is worth to zoom in on our main point of departure from Janin \& Walukiewicz' 
proofs of \eqref{eq:AutCharMSO}-\eqref{eq:AutCharMuML}. 
In the characterisation \eqref{eq:AutCharMSO}, due to \cite{Walukiewicz96}, 
a key step is to show that each automaton in $\Aut(\ofoe)$ can be simulated by
an equivalent \emph{non-deterministic} automaton of the same class. 
This is instrumental in the projection construction, allowing to build an
automaton equivalent to $\exists p.\phi \in \MSO$ starting from an automaton 
for $\phi$. 
Our counterpart \eqref{eq:AutCharNmso} is also based on a simulation theorem. 
However, we cannot proceed in the same manner, as the class $\AutW(\ofoe)$,
unlike $\Aut(\ofoe)$, is \emph{not} closed under non-deterministic simulation.
Thus we devise a different construction, which starting from a weak automaton 
$\bbA$ creates an equivalent automaton $\bbA'$ which acts non-deterministically 
only on a \emph{well-founded} portion of each accepted tree.
It turns out that the class $\AutW(\ofoe)$ is closed under this variation of 
the simulation theorem; moreover, the property of $\bbA'$ is precisely what is
needed to make a projection construction that mirrors $\nmso$-quantification.
\medskip

We now consider the automata-theoretic characterisation of $\wmso$ and $\mucML$.
Whereas in \eqref{eq:AutCharNmso}-\eqref{eq:AutCharMudML} the focus was on the 
vertical dimension of a given model, the constraint that we now need to 
translate into automata-theoretic terms concerns both \emph{vertical} and 
\emph{horizontal} dimension. 
Our revision of \eqref{eq:AutCharMSO}-\eqref{eq:AutCharMuML} thus moves on two
different axes. 
The constraint on the vertical dimension is handled analogously to the cases 
\eqref{eq:AutCharNmso}-\eqref{eq:AutCharMudML}, by switching from standard to
\emph{weak} parity automata. 
The constraint on the horizontal dimension requires more work. 
The first problem lies in finding the right one-step logic, which should be able
to express cardinality properties as $\wmso$ is able to do. 
An obvious candidate would be weak monadic second-order logic itself, or more
precisely, its variant $\owmso$ over the signature of unary predicates 
(corresponding to the automata states).
A very helpful observation from~\cite{vaananen77} is that we can actually work
with an equivalent formalism which is better tailored to our aims.
Indeed, $\owmso \equiv \ofoei$, where $\ofoei$ is the extension of $\ofoe$ with
the generalised quantifier $\qu$, with $\qu x. \phi$ stating the existence of 
\emph{infinitely} many objects satisfying $\phi$. 

At this stage, our candidate automata class for $\wmso$ could be $\AutW(\ofoei)$. However, this fails because $\ofoei$ bears too much expressive power: since it extends $\ofoe$,
we would find that, over tree models, $\AutW(\ofoei)$ extends
$\AutW(\ofoe)$, whereas we already saw that $\AutW(\ofoe) \equiv
\nmso$ is incomparable to $\wmso$.
It is here that we crucially involve the notion of \emph{continuity}. 
For a class $\AutW(\oslang)$ of weak parity automata, we call 
\emph{continuous-weak} parity automata, forming a class $\AutWC(\oslang)$, those
satisfying the following additional constraint:
\begin{itemize}
\item 
for every state $a$ with even priority $\pmap(a)$, every one-step formula $\phi
\in \oslang(A)$ defining the transitions from $a$ has to be continuous in all
states $a'$ lying in a cycle with $a$;
dually, if $\pmap(a)$ is odd, every such $\phi$ has to be $a'$-cocontinuous.\footnote{%
   It is important to stress that, even though continuity is a semantic 
   condition, we have a \emph{syntactic} characterisation of $\ofoei$-formulas
   satisfying it (see \cite{carr:mode18}), meaning that 
   $\AutWC(\ofoei)$ is definable independently of the structures that takes as
   input.
   }
\end{itemize}
We can now formulate our characterisation result as follows.
\begin{theorem}
\begin{eqnarray}
\wmso & \equiv & \AutWC(\ofoei)
 \qquad  \text{ (over the class of trees)}, \label{eq:AutCharWmso}
\\ \mucML & \equiv & \AutWC(\ofo)\label{eq:AutCharMucML}
  \qquad \text{ (over the class of all LTSs)}.
\end{eqnarray}
\end{theorem}
Thus automata for $\wmso$ deviate from $\smso$-automata $\Aut(\ofoe)$ on two
different levels: at the global level of the automaton run, because of the 
weakness and continuity constraint, and at the level of the one-step logic 
defining a single transition step. 
Another interesting point stems from pairing
\eqref{eq:AutCharWmso}-\eqref{eq:AutCharMucML} with the expressive completeness
result \eqref{eq:mucML=wmso}: although automata for $\wmso$ are based on a more
powerful one-step logic ($\ofoei$) than those for $\mucML$ ($\ofo$), modulo 
bisimilarity they characterise the same expressiveness.
This connects back to our previous observation, that the ability of $\wmso$ to 
express cardinality properties on the horizontal dimension vanishes in a 
bisimulation-invariant context.

\subsection{Outline}

It is useful to conclude this introduction with a roadmap of how the various results are achieved. In a nutshell, the two expressive completeness theorems \eqref{eq:mucML=wmso} and \eqref{eq:afmc=nmso} will be based respectively on the following two chains of equivalences:
\begin{eqnarray}
\AFMC \equiv \mu_{D}\ofo \equiv \AutW(\ofo) \equiv \AutW(\ofoe)/{\bis} \equiv \nmso/{\bis}  \ \text{(over  LTSs)}.\label{eq:chain-afmc=nmso}
\\
	\mucML \equiv \mu_{C}\ofo \equiv \AutWC(\ofo) \equiv \AutWC(\ofoei)/{\bis} \equiv \wmso/{\bis}  \ \text{  (over LTSs)}. \label{eq:chain-mucML=wmso}
\end{eqnarray}

After giving a precise definition of the necessary preliminaries in Section \ref{sec:prel}, we proceed as follows. 
First, Section \ref{sec:parityaut} introduces parity automata parametrised over 
a one-step language $\oslang$, both in the standard ($\Aut(\oslang)$), weak 
($\AutW(\oslang)$) and continuous-weak ($\AutWC(\oslang)$) form. 
With Theorems \ref{t:autofor} and \ref{t:fortoaut}, we show that
\begin{equation}\label{sec:outlineFix=Aut}
\mu_{D}\oslang\equiv \AutW(\oslang)  \qquad \qquad \mu_{C}\oslang \equiv \AutWC(\oslang)    \qquad\text{ (over LTSs)}
\end{equation}
where $\mu_{D}\oslang$ and $\mu_{C}\oslang$ are extensions of $\oslang$ with fixpoint operators subject to a ``noetherianess'' and a ``continuity'' constraint respectively. Instantiating \eqref{sec:outlineFix=Aut} yield the second equivalence both in \eqref{eq:chain-afmc=nmso} and \eqref{eq:chain-mucML=wmso}:
\begin{equation*}
	 \mu_{D}\ofo \equiv \AutW(\ofo) \qquad \qquad \mu_{C}\ofo \equiv \AutWC(\ofo)  \qquad \text{ (over LTSs)}.
\end{equation*}
Next, in Section \ref{sec:autwmso}, Theorem~\ref{t:wmsoauto}, we show how to construct from a $\wmso$-formula an equivalent automaton of the class $\AutWC(\ofoei)$. In Section \ref{sec:autnmso}, Theorem~\ref{t:nmsoauto}, we show the analogous characterisation for $\nmso$ and $\AutW(\ofoe)$. These two sections yield part of the last equivalence in \eqref{eq:chain-mucML=wmso} and in \eqref{eq:chain-afmc=nmso} respectively. 
\begin{equation} \label{eq:outlineForToAut}
 \AutW(\ofoe) \geq \nmso \qquad \qquad {\AutWC(\ofoei)} \geq \wmso \qquad \text{ (over trees)}.
 \end{equation}
Notice that, differently from all the other proof pieces, \eqref{eq:outlineForToAut} only holds on trees, because the projection construction for automata relies on the input LTSs being tree shaped. 

Section \ref{sec:fixpointToSO} yields the remaining bit of the automata characterisations. Theorem~\ref{t:mfl2mso} shows
\[ \mu_{D}\ofoe \leq \nmso \qquad \qquad \mu_{C}\ofoei \leq \wmso  \qquad \text{ (over LTSs)},\]
which, paired with \eqref{sec:outlineFix=Aut}, yields
\begin{equation*}
	 \AutW(\ofoe) \equiv \mu_{D}\ofoe \leq \nmso \quad \AutWC(\ofoei) \equiv \mu_{C}\ofoei \leq \wmso \ \text{ (over LTSs)}.
\end{equation*}
 Putting the last equation and \eqref{eq:outlineForToAut} together we have our automata characterisations
 \begin{equation*}
 \AutW(\ofoe) \equiv \nmso \qquad \qquad  {\AutWC(\ofoei)} \equiv \wmso \qquad \text{ (over trees)}.
 \end{equation*}
 which also yields the rightmost equivalence in \eqref{eq:chain-mucML=wmso} and in \eqref{eq:chain-afmc=nmso}, because any LTS is bisimilar to its tree unraveling.
  \begin{equation*}
 \AutW(\ofoe)/{\bis} \equiv \nmso/{\bis} \qquad \qquad {\AutWC(\ofoei)}/{\bis} \equiv \wmso/{\bis} \qquad \text{ (over LTSs)}.
 \end{equation*}
At last, Section \eqref{sec:expresso} is split into two parts. 
First, Theorem \ref{t:mlaut} extends the results in Section \ref{sec:parityaut} 
to complete the following chains of equivalences, yielding the first block 
in \eqref{eq:chain-afmc=nmso} and in \eqref{eq:chain-mucML=wmso}.
\[ 
\AFMC \equiv \mu_{D}\ofo \equiv \AutW(\ofo) \qquad \qquad 
\mucML \equiv \mu_{C}\ofo \equiv \AutWC(\ofo)  \qquad \text{ (over LTSs)}. 
\]
As a final step, Subsection \ref{ss:bisinv} fills the last gap 
in \eqref{eq:chain-afmc=nmso}-\eqref{eq:chain-mucML=wmso} by showing
\[ 
\AutW(\ofo) \equiv \AutW(\ofoe)/{\bis} \qquad \qquad 
\AutWC(\ofo) \equiv \AutWC(\ofoei)/{\bis} \qquad \text{ (over LTSs)}.
\]

\subsection{Conference versions and companion paper}
This journal article is based on two conference papers
\cite{DBLP:conf/lics/FacchiniVZ13,DBLP:conf/csl/CarreiroFVZ14},
which were based in their turn on a Master thesis \cite{Zanasi:Thesis:2012}
and a PhD dissertation \cite{carr:frag2015}.
Each of the two conference papers focussed on a single expressive completeness 
theorem between \eqref{eq:mucML=wmso} and \eqref{eq:afmc=nmso}: presenting both
results in a mostly uniform way has required an extensive overhaul, involving
the development of new pieces of theory, as in particular the entirety of 
Sections~\ref{sec:parity-to-mc}, \ref{sec:mc-to-parity} 
and \ref{sec:fixpointToSO}. 
All missing proofs of the conference papers are included and the simulation
theorem for $\nmso$- and $\wmso$-automata is simplified, as it is now based
on macro-states that are sets instead of relations. 
Moreover, we amended two technical issues with the characterisation 
$\AFMC \equiv \nmso /{\bis}$ presented in \cite{DBLP:conf/lics/FacchiniVZ13}.
First, the definition of noetherian subset in $\nmso$ has been made more 
precise, in order to prevent potential misunderstandings arising with the
formulation in \cite{DBLP:conf/lics/FacchiniVZ13}. 
Second, as stated in \cite{DBLP:conf/lics/FacchiniVZ13} the expressive 
completeness result was only valid on trees. 
In this version, we extend it to arbitrary LTSs, thanks to the new material 
in Section \ref{sec:fixpointToSO}. 

Finally, our approach depends on model-theoretic results on the three main
one-step logics featuring in this paper: $\ofo$, $\ofoe$ and $\ofoei$.
We believe these results to be of independent interest, and in order to save
some space here, we decided to restrict our discussion of the model theory of
these monadic predicate logics in this paper to a summary.
Full details can be found in the companion paper \cite{carr:mode18}.

\section{Preliminaries}
   \label{sec:prel}

We assume the reader to be familiar with the syntax and (game-theoretic)
semantics of the modal $\mu$-calculus and with the automata-theoretic 
perspective on this logic.
For background reading we refer to~\cite{ALG02,Ven08}; the purpose of this 
section is to fix some notation and terminology. 

\subsection{Transition systems and trees} 
\label{ssec:prelim_trees}

Throughout this article we fix a set $\Prop$ of elements that will be called
\emph{proposition letters} and denoted with small Latin letters $p, q, \ldots$ .
We will often focus on a finite subset $\pprop \subseteq_{\om} \Prop$, and
denote with $C$ the set $\pow (\pprop)$ of \emph{labels} on $\pprop$; it will
be convenient to think of $C$ as an \emph{alphabet}.
Given a binary relation $R \subseteq X \times Y$, for any element $x \in X$,
we indicate with $R[x]$ the set $\{ y \in Y \mid (x,y) \in R \}$ while 
$R^+$ and $R^{*}$ are defined respectively as the transitive closure of~$R$ and
the reflexive and transitive closure of~$R$. 
The set $\Ran(R)$ is defined as $\bigcup_{x\in X}R[x]$.

A \emph{$\pprop$-labeled transition system} (LTS) is a tuple $\bbS = 
\tup{T,R,\tscolors,s_I}$ where $T$ is the universe or domain of $\bbS$, 
$\tscolors:T\to\pow(\pprop)$ is a colouring (or marking),
$R\subseteq T^2$ is the accessibility relation and $s_I \in T$ is a 
distinguished node.
We call $\tscolors(s)$ the colour, or type, of node $s \in T$.
Observe that the colouring 
${\tscolors:T\to\pow(\pprop)}$ can be seen as a 
valuation $\tsval:\pprop\to\pow (T)$ given by $\tsval(p) \isdef \{s \in T \mid
p\in \tscolors(s)\}$. 
A \emph{$\pprop$-tree} is a $\pprop$-labeled LTS in which every node can
be reached from $s_I$, and every node except $s_I$ has a unique predecessor;
the distinguished node $s_I$ is called the \emph{root} of $\bbS$.
Each node $s \in T$ uniquely defines a subtree of $\bbS$ with carrier
$R^{*}[s]$ and root $s$. We denote this subtree by ${\bbS.s}$.

The \emph{tree unravelling} of an LTS $\bbS$ is given by $\unravel{\bbS} 
\isdef \tup{T_P,R_P,\tscolors',s_I}$ where $T_P$ is the set of finite paths in 
$\bbS$ stemming from $s_I$, $R_P(t,t')$ iff $t'$ is a one-step extension of $t$ 
and the colour of a path $t\in T_P$ is given by the colour of its last node in
$T$.
The \emph{$\omega$-unravelling} $\omegaunrav{\bbS}$ of $\bbS$ is an 
unravelling which has $\omega$-many copies of each node different from the root.

A \emph{$p$-variant} of a transition system $\bbS = \tup{T,R,\tscolors,s_I}$
is a $\pprop\cup\{p\}$-transition system $\tup{T,R,\tscolors',s_I}$
such that $\tscolors'(s)\setminus\{p\} = \tscolors(s) \setminus \{p \}$ for all 
$s \in T$.
Given a set $S \subseteq T$, we let $\bbS[p\mapsto S]$ denote the $p$-variant
where $p \in \tscolors'(s)$ iff $s \in S$.

Let $\phi \in \yvlang$ be a formula of some logic $\yvlang$,
we use $\Mod_{\yvlang}(\phi) = \{\bbS \mid \bbS \models \phi\}$ to denote 
the class of transition systems that make $\phi$ true.
The subscript $\yvlang$ will be omitted when $\yvlang$ is clear from context.
A class $\mathsf{C}$ of transition systems is said to be 
\emph{$\yvlang$-definable} if there is a formula $\phi \in \yvlang$ such that
$\Mod_{L}(\phi) = \mathsf{C}$.
We use the notation $\phi \equiv \psi$ to mean that $\Mod_{L}(\phi) = 
\Mod_{L}(\psi)$ and given two logics $\yvlang, \yvlang'$ we use $\yvlang \equiv 
\yvlang'$ when the $\yvlang$-definable and $\yvlang'$-definable classes of 
models coincide.


\subsection{Games}

We introduce some terminology and background on infinite games.
All the games that we consider involve two players called \emph{\'Eloise}
($\exists$) and \emph{Abelard} ($\forall$).
In some contexts we refer to a player $\Pi$ to specify a
a generic player in $\{\exists,\forall\}$.
Given a set $A$, by $A^*$ and $A^\omega$ we denote respectively the set of
words (finite sequences) and streams (or infinite words) over $A$.

A \emph{board game} $\mc{G}$ is a tuple $(G_{\exists},G_{\forall},E,\win)$,
where $G_{\exists}$ and $G_{\forall}$ are disjoint sets whose union
$G=G_{\exists}\cup G_{\forall}$ is called the \emph{board} of $\mc{G}$,
$E\subseteq G \times G$ is a binary relation encoding the \emph{admissible
moves}, and $\win \subseteq G^{\omega}$ is a \emph{winning condition}.
An \emph{initialized board game} $\mc{G}@u_I$ is a tuple
$(G_{\exists},G_{\forall},u_I, E,\win)$ where
$u_I \in G$ is the
\emph{initial position} of the game.
In a \emph{parity game}, the set $\win$ is  given by a \emph{parity function},
that is, a map $\pmap: G \to \omega$ of finite range, in the sense that a
sequence $(a_{i})i<\om$ belongs to $\win$ iff the maximal value $n$ that 
is reached as $n = \pmap(a_{i})$ for infinitely many $i$, is even.

Given a board game $\mc{G}$, a \emph{match} of $\mc{G}$ is simply a path
through the graph $(G,E)$; that is, a sequence $\pi = (u_i)_{i< \alpha}$ of
elements of $G$, where $\alpha$ is either $\omega$ or a natural number,
and $(u_i,u_{i+1}) \in E$ for all $i$ with $i+1 < \alpha$.
A match of $\mc{G}@u_{I}$ is supposed to start at $u_{I}$.
Given a finite match $\pi = (u_i)_{i< k}$ for some $k<\omega$, we call
$\mathit{last}(\pi) \isdef u_{k-1}$ the \emph{last position} of the match; the
player $\Pi$ such that $\mathit{last}(\pi) \in G_{\Pi}$ is supposed to move
at this position, and if $E[\mathit{last}(\pi)] = \emptyset$, we say that
$\Pi$ \emph{got stuck} in $\pi$.
A match $\pi$ is called \emph{total} if it is either finite, with one of the
two players getting stuck, or infinite. Matches that are not total are called
\emph{partial}.
Any total match $\pi$ is \emph{won} by one of the players:
If $\pi$ is finite, then it is won by the opponent of the player who gets stuck.
Otherwise, if $\pi$ is infinite, the winner is $\exists$ if $\pi \in
\win$, and $\forall$ if $\pi \not\in \win$.

Given a board game $\mc{G}$ and a player $\Pi$, let $\pmatches{G}{\Pi}$ denote
the set of partial matches of $\mc{G}$ whose last position belongs to player
$\Pi$.
A \emph{strategy for $\Pi$} is a function $f:\pmatches{G}{\Pi}\to G$.
A match $\pi  = (u_i)_{i< \alpha}$ of $\mc{G}$ is
\emph{$f$-guided} if for each $i < \alpha$ such that $u_i \in G_{\Pi}$ we
have that $u_{i+1} = f(u_0,\dots,u_i)$.
Let $u \in G$ and a $f$ be a strategy for $\Pi$.
We say that $f$ is a \emph{surviving strategy} for $\Pi$ in $\mc{G}@u$ if
for each $f$-guided partial match $\pi$ of $\mc{G}@u$, if $\mathit{last}(\pi)$
is in $G_{\Pi}$ then $f(\pi)$ is legitimate, that is, $(\mathit{last}(\pi),
f(\pi)) \in E$.
We say that $f$ is a \emph{winning strategy} for $\Pi$ in $\mc{G}@u$ if,
additionally, $\Pi$ wins each $f$-guided total match of $\mc{G}@u$.
If $\Pi$ has a winning strategy for $\mc{G}@u$ then $u$ is called a 
\emph{winning position} for $\Pi$ in $\mc{G}$.
The set of positions of $\mc{G}$ that are winning for $\Pi$ is denoted by
$\win_{\Pi}(\mc{G})$.

A strategy $f$ is called \emph{positional} if $f(\pi) = f(\pi')$ for
each $\pi,\pi'\in \Dom(f)$ with $\mathit{last}(\pi) = 
\mathit{last}(\pi')$.
A board game $\mc{G}$ with board $G$ is \emph{determined} if $G = \win_{\exists}(\mc{G}) \cup \win_{\forall}(\mc{G})$, that is, each $u \in G$ is a winning position for one of the two players.
The next result states that parity games are positionally determined.

\begin{fact}[\cite{EmersonJ91,Mostowski91Games}]
\label{THM_posDet_ParityGames}
For each parity game $\mc{G}$, there are positional strategies $f_{\exists}$
and $f_{\forall}$ respectively for player $\exists$ and $\forall$, such that
for every position $u \in G$ there is a player $\Pi$ such that $f_{\Pi}$ is a
winning strategy for $\Pi$ in $\mc{G}@u$.
\end{fact}
In the sequel we will often assume, without notification, that strategies in
parity games are positional. 
Moreover, we think of a positional strategy $f_\Pi$ for player $\Pi$ as a 
function $f_\Pi:G_\Pi\to G$.

\subsection{The Modal $\mu$-Calculus and some of its fragments.}
\label{subsec:mu}

The language of the modal $\mu$-calculus ($\muML$) is given by the following 
grammar:
\begin{equation*}
    \phi\ \isbnf  q \mid \neg q \mid \phi \land \phi \mid
    \phi \lor \phi \mid  \Diamond \phi \mid \Box \phi \mid
    \mu p.\phi \mid \nu p.\phi
\end{equation*}
%
where $p,q \in \Prop$ and $p$ is positive in $\phi$ (i.e., $p$ is not
negated).
We will freely use standard syntactic concepts and notations related to this
language, such as the sets $\FV(\phi)$ and $\BV(\phi)$ of \emph{free} and 
\emph{bound} variables of $\phi$, and the collection $\Sfor(\phi)$ of subformulas
of $\phi$.
We use the standard convention that no variable is both free and bound in a
formula and that every bound variable is fresh.
We let $\muML(\pprop)$ denote the collection of formulas $\phi$ with $\FV(\phi)
\sse \pprop$.
Sometimes we write $\psi \subf \phi$ to denote that $\psi$ is a subformula
of $\phi$.
For a bound variable $p$ occurring in some formula $\phi \in \muML$, we use
$\delta_p$ to denote the binding definition of $p$, that is, the unique formula
such that either $\mu p.\delta_p$ or $\nu p.\delta_p$ is a subformula of 
$\phi$.

We need some notation for the notion of \emph{substitution}.
Let $\phi$ and $\{ \psi_{z} \mid z \in Z \}$ be modal fixpoint formulas, where
$Z \cap \BV{(\phi)} = \nada$.
Then we let $\phi[\psi_{z}/z \mid z \in Z]$
denote the formula obtained from $\phi$ by simultaneously substituting each
formula $\psi_{z}$ for $z$ in $\phi$ (with the usual understanding that no 
free variable in any of the $\psi_{z}$ will get bound by doing so).
In case $Z$ is a singleton $z$, we will simply write $\phi[\psi_{z}/z]$, or 
$\phi[\psi]$ if $z$ is clear from context.
\smallskip

The semantics of this language is completely standard. 
Let $\bbS = \tup{T,R,\tscolors, s_I}$ be a transition system and $\phi \in 
\muML$. 
We inductively define the \emph{meaning} $\ext{\phi}^{\bbS}$ which includes
the following clauses for the least $(\mu)$ and greatest ($\nu$) fixpoint 
operators:
\begin{align*}
  \ext{\mu p.\psi}^{\bbS}  & \isdef   \bigcap \{S \subseteq T \mid S \supseteq \ext{\psi}^{\bbS[p\mapsto S]} \}  \\
  \ext{\nu p.\psi}^{\bbS}  & \phi   \bigcup \{S \subseteq T \mid S \subseteq \ext{\psi}^{\bbS[p\mapsto S]} \}
\end{align*}
We say that $\phi$ is \emph{true} in $\bbS$ (notation $\bbS \mmodels \phi$) iff $s_I \in \ext{\phi}^{\bbS}$.

We will now describe the semantics defined above in game-theoretic terms. 
That is, we will define the evaluation game $\egame(\phi,\bbS)$ associated with
a formula $\phi \in \muML$ and a transition system $\bbS$. 
This game is played by two players (\eloise and \abelard) moving through 
positions $(\xi,s)$ where $\xi \subf \phi$ and $s \in T$.

In an arbitrary position $(\xi,s)$ it is useful to think of \eloise trying to
show that $\xi$ is true at $s$, and of \abelard of trying to convince her that 
$\xi$ is false at $s$. 
The rules of the evaluation game are given in  the following table.
\begin{center}
\begin{tabular}{|l|c|l|c|}
\hline
Position & Player & Admissible moves
\\ \hline
   $(\psi_1 \vee \psi_2,s)$   & $\exists$ & $\{(\psi_1,s),(\psi_2,s) \}$ 
\\ $(\psi_1 \wedge \psi_2,s)$ & $\forall$ & $\{(\psi_1,s),(\psi_2,s) \}$ 
\\ $(\Diamond\phi,s)$         & $\exists$ & $\{(\phi,t)\ |\ t \in R[s] \}$ 
\\ $(\Box\phi,s)$             & $\forall$ & $\{(\phi,t)\ |\ t \in R[s] \}$ 
\\ $(\mu p.\phi,s)$           & $-$       & $\{(\phi,s) \}$ 
\\ $(\nu p.\phi,s)$           & $-$       & $\{(\phi,s) \}$ 
\\ $(p,s)$ with $p \in \BV(\phi)$ & $-$ & $\{(\delta_p,s) \}$ 
\\ $(\lnot q,s)$ with $q \in \FV(\phi)$ and $q \notin \tscolors(s)$ 
   & $\forall$ & $\emptyset$
\\ $(\lnot q,s)$ with $q \in \FV(\phi)$ and $q \in \tscolors(s)$ 
   & $\exists$ & $\emptyset$
\\ $(q,s)$ with $q \in \FV(\phi)$ and $q \in \tscolors(s)$ 
   & $\forall$ & $\emptyset$
\\ $(q,s)$ with $q \in \FV(\phi)$ and $q \notin \tscolors(s)$ 
   & $\exists$ & $\emptyset$
\\ \hline
\end{tabular}
\end{center}
Every finite match of this game is lost by the player that got stuck. 
To give a winning condition for an infinite match let $p$ be, of the bound 
variables of $\phi$ that get unravelled infinitely often, the one such that 
$\delta_{p}$ the highest subformula in the syntactic tree of $\phi$. 
The winner of the match is \abelard if $p$ is a $\mu$-variable and \eloise if 
$p$ is a $\nu$-variable.
We say that $\phi$ is true in $\bbS$ iff \eloise has a winning strategy in 
$\egame(\phi,\bbS)$.

\begin{proposition}[Adequacy Theorem]\label{p:unfold=evalgame}
Let $\phi = \phi(p)$ be a formula of $\muML$ in which all occurrences of $p$ are positive, $\bbS$ be a LTS and $s \in T$. Then:
\begin{equation}
\label{eq:adeq3}
s \in \ext{\mu p.\phi}^{\bbS} 
\iff (\mu p.\phi,s) \in \win_{\eloise}(\egame(\mu p.\phi,\bbS)).
\end{equation}
\end{proposition}

\bigskip
Formulas of the modal $\mu$-calculus may be classified according to their
\emph{alternation depth}, which roughly is given as the maximal length of
a chain of nested alternating least and greatest fixpoint operators~\cite{Niwinski86}.
The \emph{alternation-free fragment} of the modal $\mu$-calculus~($\AFMC$) is 
usually defined as the collection of $\muML$-formulas without nesting of least
and greatest fixpoint operators. 
It can also be also given a more standard grammatical definition as follows.

\begin{definition}
Given a set $\qprop$ of propositional variables, we define the fragment 
$\noe{\mu\ML}{\qprop}$ of \muML-formulas that are (syntactically) 
\emph{noetherian} in $\qprop$, by the following grammar:
\begin{equation*}
   \phi \isbnf  q
   \mid \psi
   \mid \phi \lor \phi
   \mid \phi \land \phi
     \mid \Diamond \phi
       \mid \Box \phi
   \mid \mu p.\phi'
\end{equation*}
where $q \in \qprop$, $\psi$ is a $\qprop$-free $\muML$-formula, and 
$\phi' \in \noe{\mu\ML}{\qprop\cup\{p\}}$. 
The \emph{co-noetherian} fragment $\conoe{\mu\ML}{Q}$ is defined dually.
\end{definition}

The alternation-free $\mu$-calculus can be defined as the fragment of the full
language where we restrict the application of the least fixpoint operator $\mu
p$ to formulas that are noetherian in $p$ (and apply a dual condition to the 
greatest fixpoint operator).

\begin{definition}
The formulas of the \emph{alternation-free} $\mu$-calculus $\AFMC$ 
are defined by the following grammar:
\begin{equation*}
   \phi \isbnf  
      q \mid \neg q  
   \mid \phi\lor\phi \mid \phi\land\phi 
      \mid \Diamond \phi
       \mid \Box \phi
   \mid \mu p. \phi'    
   \mid \nu p. \phi'',
\end{equation*} 
where $p,q \in \Prop$, $\phi' \in \AFMC \cap \noe{\mu\ML}{p}$
and dually $\phi'' \in \AFMC \cap \conoe{\mu\ML}{p}$.
\end{definition}

It is then immediate to verify that the above definition indeed captures exactly
all formulas without alternation of least and greatest fixpoints.
One may prove that 
a formula $\phi \in \muML$ belongs to the fragment $\AFMC$ iff for all 
subformulas $\mu p.\psi_1$ and $\nu q.\psi_2$ it holds that $p$ is not free in
$\psi_2$ and $q$ is not free in $\psi_1$.
%
Over arbitrary transition systems, this fragment is less expressive than the 
whole $\muML$~\cite{Park79}. 

In order to properly define the fragment $\mucML \subseteq \AFMC$ which is of
critical importance in this article, we are particularly interested in the 
\emph{continuous} fragment of the modal $\mu$-calculus. 
As observed in Section~\ref{sec:intro}, the abstract notion of continuity can 
be given a concrete interpretation in the context of $\mu$-calculus.
\begin{definition}
Let $\phi \in \muML$, and $q$ be a propositional variable. 
We say that \emph{$\phi$ is continuous in $q$} iff for every transition 
system $\bbS$ there exists some finite $S \subseteq_\omega \tsval(q)$ such that
$$
\bbS \mmodels \phi \quad\text{iff}\quad \bbS[q \mapsto S] \mmodels \phi.
$$
\end{definition}

We can give a syntactic characterisation of the fragment of $\muML$ that captures
this property. 

\begin{definition}
Given a set $\qprop$ of propositional variables, we define the fragment of \muML 
\emph{continuous} in $\qprop$, denoted by $\cont{\muML}{\qprop}$, by induction 
in the following way
\begin{equation*}
   \phi \isbnf  q
   \mid \psi
   \mid \phi \lor \phi
   \mid \phi \land \phi
   \mid \Diamond \phi
   \mid \mu p.\phi'
\end{equation*}
where $q,p \in \qprop$, $\psi$ is a $\qprop$-free $\muML$-formula and 
$\phi' \in \cont{\muML}{\qprop\cup\{p\}}$.

The  \emph{co-continuous} fragment $\cocont{\mu\ML}{Q}$ is defined dually. 
\end{definition}

\begin{proposition}[\cite{Fontaine08,FV12}]\label{prop:FVcont}
A $\muML$-formula is continuous in $q$ iff it is equivalent to a formula in the
fragment $\cont{\muML}{q}$.
\end{proposition}

Finally, we define $\mucML$ to be the fragment of $\muML$ where the use of the
least fixed point operator is restricted to the continuous fragment. 

\begin{definition}
Formulas of the fragment $\mucML$ are given by:
\begin{equation*}
   \phi \isbnf  q \mid \lnot q
    \mid \phi \lor \phi
        \mid \phi \land \phi
    \mid \Diamond \phi
     \mid \Box \phi \mid
    \mu p.\phi' 
    \mid \nu p.\phi''
    \end{equation*}
where $p,q \in \Prop$,  $\phi' \in \cont{\muML}{p} \cap \mucML$, and dually 
$\phi'' \in \cocont{\muML}{p} \cap \mucML$.
\end{definition}


It is easy to verify that $\mucML \sse \mudML$.
Characteristic about $\mucML$ is that in a formula $\mu p. \phi \in \mucML$,
all occurrences of $p$ are \emph{existential} in the sense that they may be 
in the scope of a diamond but not of a box.
Furthermore, as an immediate consequence of Proposition \ref{prop:FVcont} we 
may make the following observation.

\begin{corollary}\label{cor:cont}
For every $\mucML$-formula $\mu p. \phi$, $\phi$ is continuous in $p$.
\end{corollary}

\subsection{Bisimulation}

Bisimulation is a notion of behavioral equivalence between processes.
For the case of  transition systems, it is formally defined as follows.

\begin{definition}
Let $\bbS = \tup{T, R, \tscolors, s_I}$ and $\bbS' = \tup{T', R', \tscolors', 
s'_I}$ be $\pprop$-labeled transition systems.
A \emph{bisimulation} is a relation $Z \subseteq T \times T'$ such that for all 
$(t,t') \in Z$ the following holds:
\begin{description}
  \itemsep 0 pt
  \item[(atom)] 
  $\tscolors(t) = \tscolors'(t')$;
  \item[(forth)] 
  for all $s \in R[t]$ there is $s'\in R'[t']$ such
  that $(s,s') \in Z$;
  \item[(back)] 
  for all $s'\in R'[t']$ there is $s \in R[t]$ such
  that $(s,s') \in Z$.
\end{description}
Two pointed transition systems $\bbS$ and $\bbS'$ are
\emph{bisimilar} (denoted $\bbS \bis \bbS'$) if there is a
bisimulation $Z \subseteq T \times T'$ containing $(s_I,s'_I)$.
\end{definition}

The following observation about tree unravellings is the key to understand
the importance of tree models in the setting of invariance modulo bismilarity
results.

\begin{fact}
    \label{prop:tree_unrav}
$\bbS$, $\unravel{\bbS}$ and $\bbS^{\om}$ are bisimilar, for every transition
system $\bbS$.
\end{fact}

A class $\mathsf{C}$ of transition systems is \emph{bisimulation closed} if 
$\bbS \bis \bbS'$ implies that $\bbS \in \mathsf{C}$ iff $\bbS'
\in \mathsf{C}$, for all $\bbS$ and $\bbS'$.
A formula $\phi \in \yvlang$ is \emph{bisimulation-invariant} if $\bbS \bis
\bbS'$ implies that $\bbS \mmodels \phi$ iff $\bbS'
\mmodels \phi$, for all $\bbS$ and $\bbS'$.

\begin{fact}
Each $\muML$-definable class of transition systems is bisimulation closed.
\end{fact}

\subsection{Monadic second-order logics}
\label{sec:prel-so}

Three variants of monadic second-order logic feature in our work:
\emph{standard}, \emph{weak}, and \emph{noetherian} monadic second-order 
logic, and for each of these three variants, we consider a one-sorted and 
a two-sorted version.
As we will see later, the one-sorted version fits better in the 
automata-theoretic framework, whereas it is more convenient to use the 
two-sorted approach when translating $\mu$-calculi into second order languages.
In both the one-sorted and the two-sorted version, the syntax of the three 
languages is the same, the difference lying in the semantics, more specificaly,
in the type of subsets over which the second-order quantifiers range.
In the case of standard and weak monadic second-order logic, these quantifiers 
range over all, respectively, all finite subsets of the model.
In the case of \nmso we need the concept of a \emph{noetherian} subset of an LTS.

\begin{definition}
\label{d:bundle1}
Let $\bbS = \tup{T,R,\tscolors, s_I}$ be an LTS, and let $B$ be a non-empty set 
of finite paths that all share the same starting point $s$; we call $B$ a 
\emph{bundle rooted at} $s$, or simply an $s$-\emph{bundle}, if $B$ does not
contain an infinite ascending chain $\pi_{0} \sqsubset \pi_{1} \sqsubset 
\cdots$, where $\sqsubset$ denotes the (strict) initial-segment relation on 
paths.
A \emph{bundle} is simply an $s$-bundle for some $s \in T$.

A subset $X$ of $T$ is called \emph{noetherian} if there is a bundle $B$ such
that each $t \in X$ lies on some path in $B$.
\end{definition}

Notice that in a tree model, the noetherian subsets coincide with those that
are included in a well-founded subtree.

\subsubsection*{One-sorted monadic second-order logics}
\begin{definition}\label{def:mso}
The formulas of the \emph{(one-sorted) monadic second-order language} are
defined by the following grammar:
\begin{eqnarray*}\label{EQ_mso}
  \phi \isbnf  \here{p} \mid p \inc q \mid R(p,q) \mid \lnot\phi 
     \mid \phi\lor\phi \mid \exists p.\phi,
\end{eqnarray*}
where $p$ and $q$ are letters from $\Prop$.
We  adopt the standard convention that no proposition letter is both free and
bound in $\phi$.
\end{definition}

As mentioned, the three logics $\smso$, $\wmso$ and $\nmso$ are distinguished by
their semantics. 
Let  $\bbS = \tup{T,R,\tscolors, s_I}$ be an LTS.
The interpretation of the atomic formulas is fixed:
\begin{align*}
\bbS \models \here{p} & \quad\text{ iff }\quad  \tsval(p) = \{s_I\} \\
\bbS \models p \inc q & \quad\text{ iff }\quad  \tsval(p) \subseteq \tsval(q) \\
\bbS \models R(p,q) & \quad\text{ iff }\quad  \text{for every $s\in \tsval(p)$ there exists $t\in \tsval(q)$ such that $sRt$} 
\end{align*}
Furthermore, the interpretation of the boolean connectives is standard.
The interpretation of the existential quantifier is where the logics diverge:

\begin{align*}
\bbS \models\ \exists p. \phi  & \quad\text{ iff }\quad  \bbS[p \mapsto X] \models \phi \,
\left.\begin{cases}
   \text{for some }                   & (\smso)
\\ \text{for some \emph{finite} }     & (\wmso) 
\\ \text{for some \emph{noetherian} } & (\nmso)
\end{cases}\right\}\,
 X \subseteq T.
\end{align*}

Observe that for a given monadic second-order formula $\phi$, the classes 
$\Mod_{\smso}(\phi)$, $\Mod_{\wmso}(\phi)$ and $\Mod_{\nmso}(\phi)$ will 
generally be different.

\subsubsection*{Two-sorted monadic second-order logics}
The reader may have expected to see the following more standard language for
second-order logic.
\begin{definition}
\label{def:2mso}
Given a set $\fovar$ of individual (first-order) variables, we define the 
formulas of the \emph{two-sorted monadic second-order language} by the following
grammar:
\[
\phi \isbnf  p(x)
\mid R(x,y)
\mid x \foeq y
\mid \neg \phi
\mid \phi \lor \phi
\mid \exists x.\phi
\mid \exists p.\phi
\]
where $p \in \Prop$, $x,y \in \fovar$ and $\foeq$ is the symbol for equality.   
\end{definition}

Formulas are interpreted over an LTS $\bbS = \tup{T,R,\tscolors, s_I}$ with a
variable assignment $g: \fovar \to T$, and the semantics of the language is
completely standard. 
Depending on whether second-order quantification ranges over all subsets, over 
finite subsets or over noetherian subsets, we obtain the three two-sorted 
variants denoted respectively as $2\smso$, $2\wmso$ and $2\nmso$.

\subsubsection*{Equivalence of the two versions}
In each variant, the one-sorted and the two-sorted versions can be proved to
be equivalent, but there is a sublety due to the fact that our models have a 
distinguished state.
In the one-sorted language, we use the downarrow $\here$ to access this
distinguished state; in the two-sorted approach, we will use a \emph{fixed}
variable $v$ to refer to the distinguished state, and given a formula 
$\phi(v)$ of which $v$ is the only free individual variable, we write 
$\bbS \models \phi[s_{I}]$ rather than $\bbS[v \mapsto s_{I}] \models \phi$.
As a consequence, the proper counterpart of the one-sorted language $\smso$ is
the set $2\smso(v)$ of those $2\smso$-formulas that have precisely $v$ 
as their unique free variable.

More in particular, with $L \in \{\smso, \wmso, \nmso\}$, we say that $\phi \in
L$ is \emph{equivalent to} $\psi(v) \in L(v)$ if
\[
\bbS \models \phi \text{ iff } \bbS \models \psi[s_{I}]
\]
for every model $\bbS = \tup{T,R,\tscolors, s_I}$.
We can now state the equivalence between the two approaches to 
monadic second-order logic as follows.

\begin{fact}
\label{fact:msovs2mso}
Let $L \in \{\smso, \wmso, \nmso\}$ be a monadic second-order logic.
\begin{enumerate}
\item
There is an effective construction transforming a formula $\phi \in L$ into
an equivalent formula $\phi^{t} \in 2L(v)$.
\item
There is an effective construction transforming a formula $\psi \in 2L(v)$ into
an equivalent formula $\psi^{o} \in L$.
\end{enumerate}
\end{fact}

Since it is completely straightforward to define a translation $(\cdot)^{t}$ as 
required for part (1) of Fact~\ref{fact:msovs2mso}, we only discuss the proof 
of part (2). 
The key observation here is that a single-sorted language can interpret the 
corresponding two-sorted language by encoding every individual variable $x \in 
\fovar$ as a set variable $p_x$ denoting a singleton, and that it is easy to 
write down a formula stating that a variable indeed is interpreted by a 
singleton.
As a consequence, where $2\yvlang(\pprop,\mathsf{X})$ denotes the set of
$2L$-formulas with free second-order variables in $\pprop$ and free first-order
variables in $\mathsf{X}$, it is not hard to formulate a translation 
$(\cdot)^{m} : 2\yvlang(\pprop,\mathsf{X}) \to \yvlang(\pprop \uplus 
\{ p_{x} \mid x \in \mathsf{X} \})$
such that, for every model $\bbS$, every variable assignment $g$ and every
formula $\psi \in 2\yvlang(\Prop,\mathsf{X})$:
\[
\bbS,g \models \psi \quad\text{iff}\quad 
\bbS[p_{x} \mapsto \{g(x)\} \mid x \in \mathsf{X}] \models \psi^m.
\]
From this it is immediate that any $\psi \in 2L(v)$ satisfies
\[
\bbS \models \psi[s_{I}]
\quad\text{iff}\quad 
\bbS \models \exists p_{v} (\here{p_{v}} \land \psi^{m}),
\]
so that we may take $\psi^{o} \isdef \exists p_{v} (\here{p_{v}} \land 
\psi^{m})$.


\section{One-step logics, parity automata and $\mu$-calculi}
   \label{sec:parityaut}

This section introduces and studies the type of parity automata that will be
used in the characterisation of $\wmso$ and $\nmso$ on tree models. 
In order to define these automata in a uniform way, we introduce, at a slightly
higher level of abstraction, the notion of a \emph{one-step logic}, a concept 
from coalgebraic modal logic~\cite{cirs:modu04} which provides a nice framework 
for a general approach towards the theory of automata operating on infinite
objects.
As salient specimens of such one-step logics we will discuss monadic 
first-order logic with equality ($\ofoe$) and its extension with the infinity 
quantifier ($\ofoei$).
We then define, parametric in the language $\oslang$ of such a one-step logic, 
the notions of an $\oslang$-automaton and of a mu-calculus $\mu\oslang$, 
and we show how various classes of $\oslang$-automata effectively correspond 
to fragments of $\mu\oslang$.


\subsection{One-step logics and normal forms}
\label{sec:onestep-short}
\label{ssec:onestep}

\begin{definition}\label{def:one-step}
Given a finite set $A$ of \emph{monadic predicates}, a \emph{one-step model} is
a pair $(D, V)$ consisting of a \emph{domain} set $D$ and a \emph{valuation}
or \emph{interpretation} $V : A \to \pow D$. 
Where $B \subseteq A$, we say that $V' : A \to \pow D$ is a 
\emph{$B$-extension of} $V : A \to \pow D$, notation $V \leq_{B} V'$,
if $V(b) \subseteq V'(b)$ for every $b \in B$ and $V(a) = V'(a)$ 
for every $a \in A \setminus B$. 

A \emph{one-step language} is a map assigning any set $A$ to a collection
$\oslang(A)$ of objects called \emph{one-step formulas} over $A$. 
We assume that one-step languages come with a \emph{truth} relation $\models$
between one-step formulas and models, writing $(D, V) \models \phi$ to
denote that $(D,V)$ satisfies $\phi$.
\end{definition}

Note that we do allow the (unique) one-step model that is based on the empty
domain; we will simply denote this model as $(\nada,\nada)$.

Our chief examples of one-step languages will be variants of modal and 
first-order logic.

\begin{definition}
\label{d:oml}
A very simple example of a one-step logic is the following \emph{basic one-step
modal logic} $\oml$, of which the language is defined as follows, for a set 
$A$ of monadic predicates:
\[
\oml(A) \isdef \{ \Diamond a, \Box a \mid a \in A \}.
\]
The semantics of these formulas is given by
\begin{align*}
   (D, V) \models \Diamond a & \quad\text{ iff }\quad  V(a) \neq \nada
\\ (D, V) \models \Box a     & \quad\text{ iff }\quad  V(a) = D.
\end{align*}
\end{definition}

\begin{definition}
The one-step language $\ofoe(A)$ of \emph{first-order logic with equality} on 
a set of predicates $A$ and individual variables $\fovar$ is given by the 
sentences (formulas without free variables) generated by the following grammar,
where $a \in A$ and $x,y\in\fovar$.:
\begin{equation}\label{eq:grammarofoe}
\phi \isbnf  a(x) \mid \lnot a(x) \mid x \foeq y \mid x \not\foeq y \mid \exists x.\phi \mid \forall x.\phi \mid \phi \lor \phi \mid \phi \land \phi
\end{equation}
We use $\ofo$ for the equality-free fragment, where we omit the clauses 
$x \foeq y$ and $x \not\foeq y$.
\end{definition}

The interpretation of this language in a model $(D,V)$ with $D \neq \nada$
is completely standard.
Formulas of $\ofo$ and $\ofoe$ are interpreted inductively by augmenting 
the pair $(D,V)$ with a variable assignment $g: \fovar \to D$. 
The semantics then defines the desired truth relation $(D, V),g \models \phi$ 
between one-step models, assignments and one-step formulas.
As usual, the variable assignment $g$ can and will be omitted when we are
dealing with sentences --- and note that we only take sentences as one-step 
formulas.
For the interpretation in one-step models with empty domain we refer to 
Definition~\ref{d:ed}.


We now introduce an extension of first-order logic with two additional
quantifiers, which first appeared in the context of Mostowski's 
study~\cite{Mostowski1957} of generalised quantifiers. 
The first, written $\qu x. \phi$, expresses that there exist infinitely many
elements satisfying a formula $\phi$. 
Its dual, written $\dqu x. \phi$, expresses that there are \emph{at most 
finitely many} elements \emph{falsifying} the formula $\phi$. 
Formally:
\begin{equation}\label{eq:definfquant}
\begin{aligned}
 (D,V),g \models \qu x. \phi(x) & \quad\text{iff}\quad 
   |\{s\in D \mid (D, V),g[x\mapsto s] \models \phi(x) \}| \geq \om
\\ (D,V),g \models \dqu x. \phi(x) & \quad\text{iff}\quad 
   |\{s\in D \mid (D, V),g[x\mapsto s] \not\models \phi(x) \}| < \om
\end{aligned}
\end{equation}

\begin{definition}
\label{d:ofoei}
The one-step language $\ofoei(A)$ is defined by adding to the grammar 
\eqref{eq:grammarofoe} of $\ofoe(A)$ the cases $\qu x. \phi$ and $\dqu x. \phi$.
In the case of non-empty models, the truth relation $(D, V),g \models \phi$ is 
defined by extending the truth relation for $\ofoe(A)$ with the clauses
\eqref{eq:definfquant}.
\end{definition}

In the case of models with empty domain, we cannot give an inductive definition
of the truth relation using variable assignments.
Nevertheless, a definition of truth can be provided for formulas that are Boolean combinations of sentences of the form $Qx.\phi$, where $Q \in \{\exists, \qu,  \forall, \dqu\}$ is a quantifier.

\begin{definition}
\label{d:ed}
For the one-step model $(\nada,\nada)$ we define the truth relation as follows:
For every sentence $Qx.\phi$, where $Q \in \{\exists, \qu,  \forall, \dqu\}$,
we set
\[\begin{array}{lllllll}
     (\nada,\nada) & \not\models & Qx. \phi
   & \quad\text{ if } \quad
   & Q \in  \{\exists, \qu \}
\\ (\nada,\nada)   & \models & Qx. \phi
   & \quad\text{ if } \quad
   & Q \in  \{\forall, \dqu \},
\end{array}\]
and we extend this definition to arbitrary $\ofoei$-sentences via the standard
clauses for the boolean connectives.
\end{definition}

For various reasons it will be convenient to assume that our one-step languages
are closed under taking \emph{(boolean) duals}.
Here we say that the one-step formulas $\phi$ and $\psi$ are boolean duals if
for every one-step model we have $(D,V) \models \phi$ iff $(D,V^{c}) \not\models 
\psi$, where $V^{c}$ is the complement valuation given by $V^{c}(a) \isdef
D \setminus V(a)$, for all $a$.

As an example, it is easy to see that for the basic one-step modal logic $\oml$
the formulas $\Diamond a$ and $\Box a$ are each other's dual.
In the case of the monadic predicate logics $\ofo$, $\ofoe$ and $\ofoei$ we can 
define the boolean dual of a formula $\phi$ by a straightforward induction.

\begin{definition}
\label{def:concreteduals} 
For $\oslang \in \{ \ofo, \ofoe, \ofoei \}$, we define the following operation
on formulas:
\begin{align*}
 (a(x))^{\delta} & \isdef  a(x) 
 & (\lnot a(x))^{\delta} & \isdef  \lnot a(x) 
\\ (\top)^{\delta} & \isdef  \bot 
  & (\bot)^{\delta} & \isdef  \top 
\\  (x \approx y)^{\delta} & \isdef  x \not\approx y 
  & (x \not\approx y)^{\delta}& \isdef  x \approx y 
\\ (\phi \wedge \psi)^{\delta} &\isdef  \phi^{\delta} \vee \psi^{\delta} 
  &(\phi \vee \psi)^{\delta}& \isdef  \phi^{\delta} \wedge \psi^{\delta}
\\ (\exists x.\psi)^{\delta} &\isdef  \forall x.\psi^{\delta} 
  &(\forall x.\psi)^{\delta} &\isdef  \exists x.\psi^{\delta} 
\\ (\qu x.\psi)^{\delta} &\isdef \dqu x.\psi^{\delta} 
  &(\dqu x.\psi)^{\delta} &\isdef  \qu x.\psi^{\delta}
\end{align*}
\end{definition}
We leave it for the reader to verify that the operation $\dual{(\cdot)}$ indeed 
provides a boolean dual for every one-step sentence.
\medskip

The following semantic properties will be essential when studying the 
parity automata and $\mu$-calculi associated with one-step languages.

\begin{definition}\label{def:semnotions} 
Given a one-step language $\oslang(A)$, $\phi \in \oslang(A)$ and $B \sse A$,
\begin{itemize}
\item 
$\phi$ is \emph{monotone} in $B$ if for all pairs of one step models $(D,V)$ 
and $(D,V')$ with $V \leq_{B} V'$, $(D,V) \models \phi$ implies $(D,V'),g 
\models \phi$.
\item 
$\phi$ is \emph{$B$-continuous} if $\phi$ is monotone in $B$ and, whenever 
$(D,V) \models \phi$, then there exists $V' \: A \to \pow(D)$ such that 
$V' \leq_{B} V$, $(D,V') \models \phi$ and $V'(b)$ is finite for all $b \in B$.
\item 
$\phi$ is \emph{$B$-cocontinuous} if its dual $\phi^{\delta}$ is continuous in 
$B$.
\end{itemize}
\end{definition}

We recall from \cite{carr:mode18} syntactic characterisations of 
these semantic properties, relative to the monadic predicate logics $\ofo$, 
$\ofoe$ and $\ofoei$. 
We first discuss characterisations of monotonicity and (co)continuity given by
grammars. 

\begin{definition}
For $\oslang \in \{ \ofo, \ofoe, \ofoei \}$, we define the \emph{positive} 
fragment of $\oslang(A)$, written $\oslang^{+}(A)$, as the set of sentences 
generated by the grammar we obtain by leaving out the clause $\lnot a(x)$
from the grammar for $\oslang$.
 
For $B \subseteq A$, the \emph{$B$-continuous} fragment of $\ofoe^{+}(A)$, 
written $\cont{\ofoe(A)}{B}$, is the set of sentences generated by the following
grammar, for $b \in B$ and $\psi \in \ofoe^{+}(A \setminus B)$:
\[
\phi \isbnf  b(x) \mid \psi \mid \phi \land \phi \mid \phi \lor \phi 
   \mid \exists x.\phi.
\]
If $\psi \in \ofo^{+}(A \setminus B)$ in the condition above, we then obtain the
$B$-continuous fragment $\cont{\ofo(A)}{B}$ of $\ofo^{+}(A)$.
The  \emph{$B$-continuous} fragment of ${\ofoei}^{+}(A)$, written 
$\cont{\ofoei(A)}{B}$, is defined by adding to the above grammar the clause 
$\wqu x.(\phi,\psi)$, which is a shorthand for $\forall x.(\phi(x) \lor \psi(x)) 
\land \dqu x.\psi(x)$.\footnote{In words, 
   $\wqu x.(\phi,\psi)$ says: ``every element of the domain validates $\phi(x)$ 
   or $\psi(x)$, but only finitely many need to validate $\phi(x)$ (where $b \in 
   B$ may occur). Thus $\dqu$ makes a certain use of $\forall$ compatible with 
   continuity.}
For $\oslang \in \{ \ofo, \ofoe,\ofoei \}$ and $B \subseteq A$, the 
\emph{$B$-cocontinuous} fragment of $\oslang^{+}(A)$, written 
$\cocont{\oslang(A)}{B}$, is the set $\{\phi \mid \phi^\delta \in 
\cont{\oslang(A)}{B}\}$.
\end{definition}

Note that we do allow the clause $x \not\foeq y$ in the positive fragments of 
$\ofoe$ and $\ofoei$.

The following result provides syntactic characterizations for the mentioned 
semantics properties.

\begin{theorem}[\cite{carr:mode18}] 
    \label{th:onesteplogics-grammars}
For $\oslang \in \{\ofo, \ofoe,\ofoei \}$, we have 
let $\phi \in \oslang(A)$ be 
a one-step formula.
Then

\begin{enumerate}[(1)]
\item
$\phi \in \oslang(A)$ is $A$-monotone iff it is equivalent to some 
$\psi \in \oslang^{+}(A)$. 

\item
$\phi \in \oslang(A)$ is $B$-continuous iff it is equivalent to some $\psi \in
\cont{\oslang(A)}{B}$. 

\item
$\phi \in \oslang(A)$ is $B$-cocontinuous iff it is equivalent to some $\psi \in
\cocont{\oslang(A)}{B}$. 
\end{enumerate}
\end{theorem}

\begin{proof}
The first two statements are proved in \cite{carr:mode18}. 
The third one can be verified by a straightforward induction on $\phi$. 
\end{proof}

In some of our later proofs we need more precise information on the shape of
formulas belonging to certain syntactic fragments.
For this purpose we introduce normal forms for positive sentences in $\ofo$, 
$\ofoe$ and $\ofoei$. 

\begin{definition}
\label{def:basicform-ofoe}
\label{def:basicform-ofoei}
A \emph{type} $T$ is just a subset of $A$. It defines a $\ofoe$-formula 
\[
\tau^{+}_T(x) \df \bigwedge_{a \in T} a(x).
\]
Given a one-step model $(D,V)$, $s \in D$ \emph{witnesses} a type $T$ if 
$(D,V), g[x\mapsto s] \models \tau^{+}_T(x)$ for any~$g$. 
The predicate $\arediff{\vlist{y}}$, stating that the elements $\vlist{y}$ are 
distinct, is defined as $\arediff{y_1,\dots,y_n} \isdef 
\bigwedge_{1\leq m < m^{\prime} \leq n} (y_m \not\approx y_{m^{\prime}})$.

A formula $\phi \in \ofo(A)$ is said to be in \emph{basic form} if $\phi = 
\bigvee \posdgbnfofo{\Sigma}{\Sigma}$, where for sets $\Sigma,\Pi$ of types, 
the formula 
$\posdgbnfofo{\Sigma}{\Pi}$ is defined as 
\begin{equation*}
\posdgbnfofo{\Sigma}{\Pi} \isdef 
\bigwedge_{S \in \Sigma} \exists x\, \tau^{+}_{T_i}(x) 
\land 
\forall z. \bigvee_{S\in \Pi} \tau^{+}_S(z)
\end{equation*}

We say that $\phi \in \ofoe(A)$ is in \emph{basic form} if $\phi = \bigvee 
\posdbnfofoe{\vlist{T}}{\Pi}$ where each disjunct is of the form
\begin{equation*}
\posdbnfofoe{\vlist{T}}{\Pi} \isdef 
\exists \vlist{x}.\big(\arediff{\vlist{x}} \land \bigwedge_i \tau^{+}_{T_i}(x_i) 
\land 
\forall z.(\arediff{\vlist{x},z} \to \bigvee_{S\in \Pi} \tau^{+}_S(z))\big)
\end{equation*}
such that $\vlist{T} \in \pow(A)^k$ for some $k$ and $\Pi \subseteq \vlist{T}$. 

Finally, we say that $\phi \in \ofoei(A)$ is in \emph{basic form} if $\phi = 
\bigvee \posdbnfofoei{\vlist{T}}{\Pi}{\Sigma}$ where each disjunct is of the 
form
\begin{align*}
   \posdbnfofoei{\vlist{T}}{\Pi}{\Sigma} &\isdef
  \posdbnfofoe{\vlist{T}}{\Pi \cup \Sigma} \land \posdbnfinf{\Sigma}
\\ \posdbnfinf{\Sigma} &\isdef 
   \bigwedge_{S\in\Sigma} \qu y.\tau^{+}_S(y) \land 
      \dqu y.\bigvee_{S\in\Sigma} \tau^{+}_S(y)
\end{align*}
for some sets of types $\Pi,\Sigma \subseteq \pow A$ and $T_1, \dots, T_k 
\subseteq A$.
\end{definition}

Intuitively, the basic $\ofo$-formula $\posdgbnfofo{\Sigma}{\Sigma}$ simply 
states that $\Sigma$ covers a one-step model, in the sense that each element of
its domain witnesses some type $S$ of $\Sigma$ and each type $S$ of $\Sigma$ is 
witnessed by some element.
The formula $\posdbnfofoe{\vlist{T}}{\Pi}$ says that each one-step
model satisfying it admits a partition of its domain in two parts: distinct 
elements $t_1,\dots,t_n$ witnessing types $T_1,\dots,T_n$, and all the remaining
elements witnessing some type $S$ of $\Pi$.  
The formula $\posdbnfinf{\Sigma}$ extends the information given by
$\posdbnfofoe{\vlist{T}}{\Pi \cup \Sigma}$ by saying that (1) for every type 
$S\in\Sigma$, there are infinitely many elements witnessing each $S \in \Sigma$
and (2) only finitely many elements do not satisfy any type in $\Sigma$. 

The next theorem states that the basic formulas indeed provide normal forms.

\begin{theorem}[\cite{carr:mode18}]  
\label{t:osnf}
For each $\oslang \in \{\ofo, \ofoe,\ofoei \}$ there is an effective procedure 
transforming any sentence $\phi \in \oslang^{+}(A)$ into an equivalent
sentence $\phi^{\bullet}$ in basic $\oslang$-form.
\end{theorem}

One may use these normal forms to provide a tighter syntactic characterisation
for the notion of continuity, in the cases of $\ofo$ and $\ofoei$.

\begin{theorem}[\cite{carr:mode18}]  
\label{t:osnf-cont}
\begin{enumerate}

\item 
A formula $\phi \in \ofo(A)$ is continuous in $B \subseteq A$ iff it is
equivalent to a formula, effectively obtainable from $\phi$, in the basic form 
$\bigvee \posdgbnfofo{\Sigma'}{\Sigma}$ 
where we require that $B \cap \bigcup\Sigma = \nada$ for every $\Sigma$.

\item A formula $\phi \in \ofoei(A)$ is continuous in $B \subseteq A$ iff it is
equivalent to a formula, effectively obtainable from $\phi$, in the basic form 
$\bigvee \mondbnfofoei{\vlist{T}}{\Pi}{\Sigma}{+}$, 
where we require that $B \cap \bigcup\Sigma = \nada$ for every $\Sigma$.
\end{enumerate}
\end{theorem}

\begin{remark} 
We focussed on normal form results for monotone and (co)continuous sentences, 
as these are the ones relevant to our study of parity automata.
However, generic sentences both of $\ofo$, $\ofoe$ and $\ofoei$ also enjoy 
normal form results, with the syntactic formats given by variations of the 
``basic form'' above. 
The interested reader may find in \cite{carr:mode18} a detailed 
overview of these results.
\end{remark}

We finish this section with a disucssion of the notion of \emph{separation}.

\begin{definition}
    \label{d:sep} 
Fix a one-step language $\oslang$, and two sets $A$ and $B$ with $B \sse A$.
Given a one-step model $(D,V)$, we say that $V: A \to \pow D$ \emph{separates}
$B$ if $\sz{V^{-1}(d) \cap B} \leq 1$, for every $d \in D$.
A formula $\phi \in \oslang(A)$ is \emph{$B$-separating} if $\phi$ is monotone
in $B$ and, whenever $(D,V) \models \phi$, then there exists a $B$-separating 
valuation $V' \: A \to \pow(D)$ such that $V' \leq_{B} V$ and $(D,V') \models
\phi$. 
\end{definition}

Intuitively, a formula $\phi$ is $B$-separating if its truth in a monadic model
never requires an element of the domain to satisfy two distinct predicates in 
$B$ at the same time; any valuation violating this constraint can be reduced to
a valuation satisfying it, without sacrificing the truth of $\phi$.
We do not need a full syntactic characterisation of this notion, but the 
following sufficient condition is used later on.

\begin{proposition}  
\label{p:sep}
\begin{enumerate}[(1)]
\item 
Let $\phi \in \ofoe^{+}(A)$ be a formula in basic form, $\phi = 
\bigvee \posdbnfofoe{\vlist{T}}{\Pi}$. 
Then $\phi$ is $B$-separating if, for each disjunct,  $\sz{S \cap B} \leq 1$ for 
each $S \in \{T_1, \dots, T_k\} \cup \Pi$.
\item
Let $\phi \in {\ofoei}^{+}(A)$ be a formula in basic form, $\phi = 
\bigvee \mondbnfofoei{\vlist{T}}{\Pi}{\Sigma}{+}$. 
Then $\phi$ is $B$-separating if, for each disjunct,  $\sz{S \cap B} \leq 1$ for 
each $S \in \{T_1, \dots, T_k\} \cup \Pi \cup \Sigma$.
\end{enumerate}
\end{proposition}

\begin{proof}
We only discuss the case $\oslang = \ofoei$: a simplification of the same 
argument yields the case $\oslang = \ofoe$. 
Aassume that $(D,V) \models \phi$ for some model $(D,V)$. 
We want to construct a valuation $V' \leq_{B} V$ witnessing the $B$-separation
property. 
First, we fix one disjunct 
$\psi = \mondbnfofoei{\vlist{T}}{\Pi}{\Sigma}{+}$ of $\phi^{\bullet}$ such that
$(D,V) \models \psi$. 
The syntactic shape of $\psi$ implies that $(D,V)$ can be partitioned in three 
sets $D_1$, $D_2$ and $D_3$ as follows: $D_1$ contains elements $s_1, \dots,
s_k$ witnessing types $T_1,\dots, T_k,$ respectively; among the remaining
elements, there are infinitely many witnessing some $S\in \Sigma$ (these form
$D_2$), and finitely many not witnessing any $S \in \Sigma$ but each witnessing
some $R \in \Pi$ (these form $D_3$). 
In other words, we have assigned to each $d \in D$ a type $S_{d} \in
\{T_1, \dots, T_k\} \cup \Pi \cup \Sigma$ such that $d$ witnesses $S_{d}$.
Now consider the valuation $U$ that we obtain by pruning $V$ to the extent
that $U(a) \isdef V(a)$ for $a \in A \setminus B$, while $U(b) \isdef 
\{ d \in D \mid b \in S_{d}\}$.
It is then easy to see that we still have $(D,U) \models \psi$, while it is 
obvious that $U$ separates $B$ and that $U \leq_{B} A$.
Therefore $\psi$ is $B$-separating and so $\phi$ is too.
\end{proof}


\subsection{Parity automata}
\label{ssec:parityaut}

Throughout the rest of the section we fix, next to a set $\pprop$ of proposition 
letters, a one-step language $\oslang$, as defined in 
Subsection \ref{sec:onestep-short}.
In light of the results therein, we assume that we have isolated fragments 
$\oslang^+(A)$, $\cont{\oslang(A)}{B}$ and $\cocont{\oslang(A)}{B}$ consisting
of one-step formulas in $\oslang(A)$ that are respectively monotone, 
$B$-continuous and $B$-co-continuous, for $B \subseteq A$.

We first recall the definition of a general parity automaton, adapted to this
setting. 

\begin{definition}[Parity Automata] \label{def:partityaut}
A \emph{parity automaton} based on the one-step language $\oslang$ and the set
$\pprop$ of proposition letters, or briefly: an \emph{$\oslang$-automaton}, is a 
tuple $\bbA = \tup{A,\tmap,\pmap,a_I}$ such that $A$ is a finite set of states,
$a_I \in A$ is the initial state, $\tmap: A\times \pow(\pprop) \to \oslang^+(A)$
is the transition map, and $\pmap: A \to \nat$ is the priority map.
The class of such automata will be denoted by $\Aut(\oslang)$.

Acceptance of a $\pprop$-transition system $\bbS = \tup{T,R,\tscolors,s_I}$ by
$\bbA$ is determined by the \emph{acceptance game} $\agame(\bbA,\bbS)$ of $\bbA$
on $\bbS$. 
This is the parity game defined according to the rules of the following table.
\begin{center}
\small
\begin{tabular}{|l|c|l|c|} \hline
Position & Player & Admissible moves & Priority \\
\hline
    $(a,s) \in A \times T$
  & $\eloise$
  & $\{V : A \to \pow(R[s]) \mid (R[s],V) \models \tmap (a, \tscolors(s)) \}$
  & $\pmap(a)$ 
\\
    $V : A \rightarrow \pow(T)$
  & $\abelard$
  & $\{(b,t) \mid t \in V(b)\}$
  & $0$
\\ \hline
 \end{tabular}
\end{center}
$\bbA$ \emph{accepts} $\bbS$ if $\eloise$ has a winning strategy in 
$\agame(\bbA,\bbS)@(a_I,s_I)$, and \emph{rejects} $\bbS$ if $(a_I,s_I)$ is 
a winning position for $\abelard$. 
We write $\Mod(\bbA)$ for the class of transition systems that are accepted
by $\bbA$ and $\TMod(\bbA)$ for the class of tree models in $\Mod(\bbA)$.
\end{definition}

Explained in words, the acceptance game $\agame(\bbA,\bbS)$ proceeds in rounds, 
each round moving from one basic position $(a,s) \in A \times T$ to the next.
At such a basic position, it is $\eloise$'s task to turn the set $R(s)$ of 
successors of $s$ into the domain of a one-step model for the formula 
$\tmap(a, \tscolors(s)) \in \oslang(A)$.
That is, she needs to come up with a valuation $V: A \to \pow(R[s])$ such that
$(R[s],V) \models \tmap (a, \tscolors(s))$ (and if she cannot find such a 
valuation, she looses immediately).
One may think of the set $\{(b,t) \mid t \in V(b)\}$ as a collection of 
witnesses to her claim that, indeed, $(R[s],V) \models \tmap (a, \tscolors(s))$.
The round ends with $\abelard$ picking one of these witnesses, which then
becomes the basic position at the start of the next round.
(Unless, of course, $\eloise$ managed to satisfy the formula $\tmap(a, 
\tscolors(s))$ with an empty set of witnesses, in which case $\abelard$ gets 
stuck and looses immediately.)

Many properties of parity automata can already be determined at the one-step
level.
An important example concerns the notion of complementation, which will be used
later in this section. Recall the notion of \emph{dual} of a one-step formula (Definition \ref{def:one-step}). Following ideas from~\cite{Muller1987,DBLP:conf/calco/KissigV09}, we can use duals, together with a \emph{role switch} between $\abelard$ and
$\eloise$, in order to define a negation or complementation operation on 
automata.


\begin{definition}
\label{d:caut}
Assume that, for some one-step language $\oslang$, the map $\dual{(\cdot)}$
provides, for each set $A$, a dual $\dual{\varphi} \in \oslang(A)$ for each
$\varphi \in \oslang(A)$.
We define the \emph{complement} of a given $\oslang$-automaton 
$\bbA = \tup{A,\tmap,\pmap,a_I}$ as the automaton $\dual{\bbA} \isdef 
\tup{A,\dual{\tmap},\dual{\pmap},a_I}$ where $\dual{\tmap}(a,c) \isdef
\dual{(\tmap(a,c))}$, and $\dual{\pmap}(a) \isdef 1 + \pmap(a)$, for all 
$a \in A$ and $c \in \pow(\pprop)$.
\end{definition}

\begin{proposition}
\label{prop:autcomplementation}
Let $\oslang$ and $\dual{(\cdot)}$ be as in the previous definition.
For each $\bbA \in \Aut(\oslang)$ and $\bbS$ we have that $\dual{\bbA}$ accepts
$\bbS$ if and only if $\bbA$ rejects $\bbS$.
\end{proposition}

The proof of Proposition~\ref{prop:autcomplementation} is based on the fact
that the \emph{power} of $\eloise$ in $\agame(\dual{\bbA},\bbS)$ is the same
as that of $\abelard$ in $\agame(\bbA,\bbS)$, as defined 
in~\cite{DBLP:conf/calco/KissigV09}. 
As an immediate consequence, one may show that if the one-step language 
$\oslang$ is closed under duals, then the class $\Aut(\oslang)$ is closed 
under taking complementation.
Further on we will use Proposition~\ref{prop:autcomplementation} to show that
the same may apply to some subclasses of $\Aut(\oslang)$.

The automata-theoretic characterisation of $\wmso$ and $\nmso$ will use classes 
of parity automata constrained by two additional properties.
To formulate these we first introduce the notion of a \emph{cluster}.

\begin{definition}
Let $\oslang$ be a one-step language, and let $\bbA = \tup{A,\tmap,\pmap,a_I}$
be in $\Aut(\oslang)$. 
Write $\ord$ for the reachability relation in $\bbA$, i.e., the transitive 
closure of the ``occurrence relation'' $\{ (a,b) \mid b \text{ occurs in }
\tmap(a,c) \text{ for some } c \in \pow(\pprop) \}$;
in case $a \ord b$ we say that $b$ is \emph{active} in $a$.
A \emph{cluster} of $\bbA$ is a cell of the equivalence relation generated by 
the relation $\ord \cap \succ$ (i.e., the intersection of $\ord$ with its 
converse).
A cluster is called \emph{degenerate} if it consists of a  single element which
is not active in itself.
\end{definition}

Observe that any cluster of an automaton is either degenerate, or else each
of its states is active in itself and in any other state of the cluster.
Observe too that there is a natural order on clusters: we may say that one
cluster is \emph{higher} than another if each member of the second cluster
if active in each member of the first.
We may assume without loss of generality that the initial state belongs to the
highest cluster of the automaton.

We can now formulate the mentioned requirements on $\oslang$-automata as follows.

\begin{definition}
\label{d:wk}
\label{d:ctwk}
Let $\bbA = \tup{A,\tmap,\pmap,a_I}$ be some $\oslang$-automaton.
We say that $\bbA$ is \emph{weak} if $\pmap(a) = \pmap(b)$ whenever $a$ and $b$
belong to the same cluster.
For the property of \emph{continuity} we require that, for any cluster $M$, any
state $a \in M$ and any $c \in \pow\pprop$, we have that 
$\pmap(a) = 1$ implies $\tmap(a,c) \in \cont{\oslang(A)}{M}$
and 
$\pmap(a) = 0$ implies $\tmap(a,c) \in \cocont{\oslang(A)}{M}$.

We call a parity automaton $\bbA \in \Aut(\oslang)$ \emph{weak-continuous} if it 
satisfies both properties, weakness and continuity.
The classes of weak and weak-continuous automata are denoted as $\AutW(\oslang)$
and $\AutWC(\oslang)$, respectively.
\end{definition}


Intuitively, weakness forbids an automaton to register non-trivial properties 
concerning the vertical `dimension' of input trees, whereas continuity expresses
a constraint on how much of the horizontal `dimension' of an input tree the 
automaton is allowed to process. 
In terms of second-order logic, they correspond respectively to quantification 
over `vertically' finite (i.e. included in well-founded subtrees) and 
`horizontally' finite (i.e. included in finitely branching subtrees) sets. 
The conjunction of weakness and continuity thus corresponds to quantification 
over finite sets. 

\begin{remark}\label{rmk:weak01}
Any weak parity automaton $\bbA$ is equivalent to a special weak automaton
$\bbA'$ with $\pmap: A' \to \{0,1\}$. 
This is because \emph{(weakness)} prevents states of different parity to occur
infinitely often in acceptance games; so we may just replace any even priority 
with $0$, and any odd priority with $1$.
We shall assume such a restricted priority map for weak parity automata.
\end{remark}

\subsection{$\mu$-Calculi}
\label{sec:onestep-to-mc}

We now see how to associate, with each one-step language $\oslang$, the 
following variant $\mu\oslang$ of the modal $\mu$-calculus.
These logics are of a fairly artificial nature; their main use is to smoothen
the translations from automata to second-order formulas further on.

\begin{definition}
Given a one-step language $\oslang$, we define the language $\mu\oslang$ of the 
\emph{$\mu$-calculus over $\oslang$}  by the following grammar:
\[
\varphi \isbnf  
   q \mid \neg q 
   \mid \varphi\lor\varphi \mid \varphi\land\varphi 
   \mid \nxt{\al}(\varphi_{1},\ldots,\varphi_{n})
   \mid \mu p. \varphi'    \mid \nu p. \varphi',
\]
where $p,q \in\Prop$, $\al(a_{1},\ldots,a_{n}) \in \oslang^{+}$ and $\varphi'$ 
is monotone in $p$.

As in the case of the modal $\mu$-calculus $\muML$, we will freely use standard
syntactic concepts and notations related to this language.
\end{definition}

Observe that the language $\mu\oslang$ generally has a wealth of modalities:
one for each one-step formula in $\oslang$.

The semantics of this language is given as follows.

\begin{definition}
Let $\bbS$ be a transition system.
The satisfaction relation $\mmodels$ is defined in the standard way, with the 
following clause for the modality $\nxt{\alpha}$:
\begin{equation}\label{eq:mumod}
\bbS \mmodels \nxt{\al}(\varphi_{1},\ldots,\varphi_{n})
\quad\text{iff}\quad 
(R[s_{I}],V_{\overline{\varphi}}) \models \al(a_{1},\ldots,a_{n}),
\end{equation}
where $V_{\overline{\varphi}}$ is the one-step valuation given by 
\begin{equation}\label{eq:valmod}
V_{\overline{\varphi}}(a_{i}) \isdef 
  \{ t \in R[s_{I}] \mid \bbS.t \mmodels \varphi_{i}\}.
\end{equation}
\end{definition}

\begin{example}
\label{ex:mul1}
\begin{enumerate}[(1)]
\item
If we identify the modalities $\nxt{\Diamond a}$ and $\nxt{\Box a}$ of the basic
modal one-step language $\oml$ (cf.~Definition~\ref{d:oml}) with the standard
$\Diamond$ and $\Box$ operators, we may observe that $\mu(\oml)$ corresponds to
the standard modal $\mu$-calculus: $\mu(\oml) = \muML$.
\item
Consider the one-step formulas 
$\al = \exists x (a_{1}(x) \land \forall y\, a_{2}(y))$,
$\beta = \exists x y (x \not\foeq y \land a_{1}(x) \land a_{1}(y))$, and
$\gamma = \wqu x (a_{1}(x), a_{2}(x))$.
Then $\nxt{\al}(\phi_{1},\phi_{2})$ is equivalent to the modal formula
$\Diamond \phi_{1} \land \Box \phi_{2}$ and 
$\nxt{\beta}(\phi)$ expresses that the current state has at least two 
successors where $\phi$ holds.
The formula $\nxt{\gamma}(\phi_{1},\phi_{2})$ holds at a state $s$ if all 
its successors satisfy $\phi_{1}$ or $\phi_{2}$, while at most finitely
many successors refute $\phi_{2}$.
Neither $\nxt{\beta}$ nor $\nxt{\gamma}$ can be expressed in standard modal 
logic.
\item
If the one-step language $\oslang$ is closed under taking disjunctions  
(conjunctions, respectively), it is easy to see that 
$\nxt{\al\lor\beta}(\ol{\phi}) \equiv \nxt{\al}(\ol{\phi}) \lor 
\nxt{\beta}(\ol{\phi})$ 
($\nxt{\al\land\beta}(\ol{\phi}) \equiv \nxt{\al}(\ol{\phi}) \land 
\nxt{\beta}(\ol{\phi})$, respectively).
\end{enumerate}
\end{example}

Alternatively but equivalently, one may interpret the language
game-theoretically.

\begin{definition}
Given a $\mu\oslang$-formula $\phi$ and a model $\bbS$ we define the 
\emph{evaluation game} $\egame(\varphi,\bbS)$ as the two-player infinite
game 
whose rules are given in the next table.
\begin{center}
\begin{tabular}{|l|c|l|c|}
\hline
Position & Player & Admissible moves
\\\hline
    $(q,s)$, with $q \in \FV(\phi) \cap \tscolors(s)$ 
  & $\abelard$ 
  & $\emptyset$
\\  $(q,s)$, with $q \in \FV(\phi) \setminus \tscolors(s)$ 
  & $\eloise$ & $\emptyset$
\\  $(\lnot q,s)$, with $q \in \FV(\phi) \cap \tscolors(s)$ 
  & $\eloise$ 
  & $\emptyset$
\\  $(\lnot q,s)$, with $q \in \FV(\phi) \setminus \tscolors(s)$ 
  & $\abelard$ 
  & $\emptyset$
\\ $(\psi_1 \lor \psi_2,s)$ 
  & $\eloise$ 
  & $\{(\psi_1,s),(\psi_2,s) \}$ 
\\  $(\psi_1 \land \psi_2,s)$ 
  & $\abelard$ 
  & $\{(\psi_1,s),(\psi_2,s) \}$ 
\\  $(\nxt{\al}(\varphi_{1},\ldots,\varphi_{n}),s)$ 
  & $\eloise$ 
  & $\{ Z \sse \{ \varphi_{1},\ldots,\varphi_{n} \} \times R[s]
     \mid (R[s],V^{*}_{Z}) \models \al(\ol{a}) \}$ 
\\  $Z \sse  \Sfor(\phi) \times S$
  & $\abelard$
  & $\{ (\psi, s) \mid (\psi,s) \in Z \}$
\\  $(\mu p.\varphi,s)$ & $-$ & $\{(\varphi,s) \}$ 
\\  $(\nu p.\varphi,s)$ & $-$ & $\{(\varphi,s) \}$ 
\\  $(p,s)$, with $p \in \BV(\phi)$ & $-$ & $\{(\delta_p,s) \}$ \\
  \hline
\end{tabular}
\end{center}
For the admissible moves at a position of the form 
$(\nxt{\al}(\varphi_{1},\ldots,\varphi_{n}),s)$, we consider the valuation 
$V^{*}_{Z}: \{ a_{1}, \ldots, a_{n} \} \to \pow(R[s])$, given by
$V^{*}_{Z}(a_{i}) \isdef \{ t \in R[s] \mid (\phi_{i},t) \in Z \}$.
The winning conditions of $\egame(\varphi,\bbS)$ are standard: $\eloise$ wins
those infinite matches of which the highest variable that is unfolded infinitely
often during the match is a $\mu$-variable.
\end{definition}

The following proposition, 
stating the adequacy of the evaluation game for the semantics of $\mu\oslang$,
is formulated explicitly for future reference.
We omit the proof, which is completely routine.

\begin{fact}[Adequacy]
\label{f:adeqmu}
For any formula $\phi \in \mu\oslang$ and any model $\bbS$ the following 
equivalence holds:
\[
\bbS \mmodels \phi
\quad\text{iff}\quad 
(\phi,s_{I}) \text{ is a winning position for $\eloise$ in } 
\egame(\varphi,\bbS).
\]
\end{fact}

We will be specifically interested in two fragments of $\mu\oslang$, associated 
with the properties of being noetherian and continuous, respectively, and with 
the associated variants of the $\mu$-calculus $\mu\oslang$ where the use of the 
fixpoint operator $\mu$ is restricted to formulas belonging to these two
respective fragments.

\begin{definition}
Let $\qprop$ be a set  of proposition letters.
We first define the fragment $\noe{\mu\oslang}{\qprop}$ of $\mu\oslang$ of 
formulas that are syntactically \emph{noetherian} in $\qprop$ by the following 
grammar:
\begin{equation*}
   \varphi \isbnf  q
   \mid \psi
   \mid \varphi \lor \varphi
   \mid \varphi \land \varphi
   \mid \nxt{\al}(\varphi_{1},\ldots,\varphi_{n})
   \mid \mu p.\phi'
\end{equation*}
where $q \in \qprop$, $\psi$ is a $\qprop$-free $\muML$-formula,
$\al(a_{1},\ldots,a_{n}) \in \oslang^{+}$ and 
$\phi' \in \noe{\mu\oslang}{\qprop\cup\{p\}}$. 
The \emph{co-noetherian} fragment $\conoe{\mu\oslang}{Q}$ is defined dually.

Similarly, we define the fragment $\cont{\mu\oslang}{\qprop}$ of 
$\mu\oslang$-formulas that are syntactically \emph{continuous} in $\qprop$ as
follows:
\begin{equation*}
   \varphi \isbnf  q
   \mid \psi
   \mid \varphi \lor \varphi
   \mid \varphi \land \varphi
   \mid 
   \nxt{\al}(\varphi_{1},\ldots,\varphi_{k},\psi_{1},\ldots,\psi_{m})
   \mid \mu p.\phi'
\end{equation*}
where $p\in\Prop$, $q \in \qprop$, $\psi$, $\psi_{i}$ are $\qprop$-free 
$\mu\oslang$-formula, $\al(a_{1},\ldots,a_{k},b_{1},\ldots,b_{m}) \in 
\cont{\oslang}{\ol{a}}(\ol{a},\ol{b})$,
and $\phi' \in \cont{\mu\oslang}{\qprop\cup\{p\}}$. 
The \emph{co-continuous} fragment $\cocont{\mu\oslang}{Q}$ is defined dually.
\end{definition}

Based on this we can now define the mentioned variants 
of the $\mu$-calculus $\mu\oslang$ where the use of the least (greatest) 
fixpoint operator can only be applied to formulas that belong to, 
respectively, the noetherian (co-noetherian) and continuous (co-continuous)
fragment of the language that we are defining.

\begin{definition}
The formulas of the \emph{alternation-free} $\mu$-calculus $\mu_{D}\oslang$ 
are defined by the following grammar:
\begin{equation*}
   \varphi \isbnf  
      q \mid \neg q 
   \mid \varphi\lor\varphi \mid \varphi\land\varphi 
   \mid \nxt{\al}(\varphi_{1},\ldots,\varphi_{n})
   \mid \mu p. \varphi'    
   \mid \nu p. \varphi'',
\end{equation*} 
where $\al(a_{1},\ldots,a_{n}) \in \oslang^{+}$,
$\phi' \in \mu_{D}\oslang \cap \noe{\mu\oslang}{p}$
and dually $\phi'' \in \mu_{D}\oslang \cap \conoe{\mu\oslang}{p}$.

Similarly, the formulas of the \emph{continuous} $\mu$-calculus $\mu_{C}\oslang$
are given by the grammar
\begin{equation*}
   \varphi \isbnf  
      q \mid \neg q 
   \mid \varphi\lor\varphi \mid \varphi\land\varphi 
   \mid \nxt{\al}(\varphi_{1},\ldots,\varphi_{n})
   \mid \mu p. \varphi'    
   \mid \nu p. \varphi'',
\end{equation*} 
where $\al(a_{1},\ldots,a_{n}) \in \oslang^{+}$,
$\phi' \in \mu_{C}\oslang \cap \cont{\mu\oslang}{p}$
and dually $\phi'' \in \mu_{C}\oslang \cap \cocont{\mu\oslang}{p}$.
\end{definition}

\begin{example}
Following up on Example~\ref{ex:mul1}, it is easy to verify that 
$\mu_{D}\oml = \mudML$ and $\mu_{C}\oml = \mucML$.
\end{example}

\subsection{From automata to formulas}
   \label{sec:parity-to-mc}

It is well-known that there are effective translations from automata to formulas
and vice versa~\cite{ALG02}.
The first result on $\oslang$-automata that we need in this paper is the 
following.

\begin{theorem}\label{t:autofor}
There is an effective procedure that, given an automaton $\bbA$ in 
$\Aut(\oslang)$, returns a formula $\xi_{\bbA} \in \mu\oslang$ which satisfies
the following properties:
\begin{enumerate}[(1)]
    
\item $\xi_{\bbA}$ is equivalent to $\bbA$;

\item $\xi_{\bbA} \in \mu_{D}\oslang$ if $\bbA \in \AutW(\oslang)$;

\item $\xi_{\bbA} \in \mu_{C}\oslang$ if $\bbA \in \AutWC(\oslang)$.
\end{enumerate}
\end{theorem}

In the remainder of this subsection
we focus on the proof of this theorem, which is (a refinement of)
a variation of the standard proof showing that any modal automaton can be 
translated into an equivalent formula in the modal $\mu$-calculus (see for
instance~\cite[Section 6]{Ven08}). 
For this reason we will not go into the details of showing that $\bbA$ and 
$\xi_{\bbA}$ are equivalent, but we will provide a detailed definition of the 
translation, and pay special attention to showing that the translations of weak
and of weak-continuous $\oslang$-automata land in the right fragments of 
$\mu\oslang$. 

The definition of $\xi_{\bbA}$ is by induction on the number of clusters of
$\bbA$, with a subinduction based on the number of states in the top cluster 
of $\bbA$.
For this inner induction we need to widen the class of $\oslang$-automata, and
it will also be convenient to introduce the notion of a preautomaton (which is  
basically an automaton without initital state).

\begin{definition}
A \emph{preautomaton} based on $\oslang$ and $\pprop$, or briefly: a
\emph{preautomaton}, is a triple $\bbA = \tup{A,\tmap,\pmap}$ such that $A$ is
a (possibly empty) finite set of states, $\tmap: A\times \pow(\pprop) \to 
\oslang^+(A)$ and $\pmap: A \to \nat$.

Given a set $X$ of propositional variables, a \emph{generalized preautomaton} 
over $\pprop$ and $X$ is a triple $\bbA = \tup{A,\tmap,\pmap}$ such that $\pmap:
A \to \nat$ is a priority map on the finite state set $A$, while the transition 
map is of the form $\tmap: A\times \pow(\pprop) \to \oslang^+(A\cup X)$.
\end{definition}

Since we will not prove the semantic equivalence of $\bbA$ and $\xi_{\bbA}$,
we confine our attention to the semantics of generalised automata to the
following remark.

\begin{remark}
Generalised automata operate on $\pprop \cup X$-models; it will be convenient to
denote these structures as quintuples of the form $\bbS = \tup{S,R,\tscolors,U,
s_{I}}$, where $\tscolors: \pprop \to \pow S$ is a $\pprop$-colouring and $U: S \to 
\pow X$ is an $X$-valuation on $S$.
The acceptance game $\agame(\bbA,\bbS)$ associated with a generalised automaton
$\bbA = \tup{A,a_{I},\tmap,\pmap}$ and a $\pprop \cup X$-model $\bbS$ is a minor
variation of the one associated with a standard automaton.
At a basic position of the form $(a,s) \in A \times S$, as before $\eloise$ 
needs to come up with a valuation $V$ turning the set $R[s]$ into the domain of 
a one-step model of the formula $\tmap(a,\tscolors(s))$.
The difference with standard automata is that the formula $\tmap(a,\tscolors(s))$
may now involve variables from the set $X$, and that the interpretation of these
is already fixed by the valuation $U$ of $\bbS$, namely by the restriction 
$U_{s}: x \mapsto R[s] \cap U(x)$ to the collection of successors of $s$.
In table, we can present this game as follows:
\begin{center}
\small
\begin{tabular}{|l|c|l|c|} \hline
Position & Player & Admissible moves & Priority \\
\hline
    $(a,s) \in A \times S$
  & $\eloise$
  & $\{V : A \to \pow(R[s]) \mid (R[s],V\cup U_{s}) \models \tmap (a, \tscolors(s)) \}$
  & $\pmap(a)$ 
\\
    $V : A \rightarrow \pow(S)$
  & $\abelard$
  & $\{(b,t) \mid t \in V(b)\}$
  & $0$
\\ \hline
 \end{tabular}
\end{center}
where $V \cup U_{s}$ is the obviously defined $ A \cup X$-valuation on $R[s]$.
\end{remark}

We now turn to the definition of the translation.
We will use the same notation for substitution as for the standard 
$\mu$-calculus, cf.~Subsection~\ref{subsec:mu}.
In addition we use the following notation.

\begin{definition}
Consider, for some preautomaton $\bbA = \tup{A,\tmap,\pmap}$, some state $a \in
A$, and some colour $c \in \pow(\pprop)$, the one-step formula $\tmap(a,c)\in 
\oslang(A)$.
Suppose that for some subset $B \sse A$ we have a collection of 
$\mu\oslang$-formulas $\{ \phi_{b} \mid b \in B \}$.
Without loss of generality we may write $\tmap(a,c) = 
\al(a_{1},\ldots,a_{m},b_{1},\ldots,b_{n})$, where the $a_{i}$ and $b_{j}$ 
belong to $A\setminus B$ and $B$ respectively.
Then we will denote the $\mu\oslang$-formula 
$\nxt{\al}(a_{1},\ldots,a_{m},\phi_{1},\ldots,\phi_{n})$ as follows:
\[
\nxt{\tmap(a,c)}(\phi_{b}/b \mid b \in B)
\isdef \nxt{\al}(a_{1},\ldots,a_{m},\phi_{1},\ldots,\phi_{n}).
\]
\end{definition}

\noindent
We can now define the desired translation from $\oslang$-automata to 
$\mu\oslang$-formulas.

\begin{definition}
\label{d:tr}
By induction on the number of clusters of a preautomaton $\bbA = \tup{A,\tmap,
\pmap}$ we define a map 
\[
\ytr_{\bbA}: A \to \mu\oslang(P).
\]
Based on this definition, for an automaton $\bbA = \tup{A,\tmap,\pmap, a_{I}}$ 
we put
\[
\xi_{\bbA} \isdef \ytr_{\tup{A,\tmap,\pmap}}(a_{I}).
\]

In the base case of the definition of $\ytr$ the preautomaton $\bbA$ has no 
clusters at all, which means in particular that $A = \nada$.
In this case we let $\ytr_{\bbA}$ be the empty map.

In the inductive case we assume that $\bbA = \tup{A,\tmap,\pmap, a_{I}}$ does 
have clusters. 
Let $B \neq \nada$ be the highest cluster, and let $\bbA^{-}$ denote the 
preautomaton with carrier $A \setminus B$, obtained by restricting the maps 
$\tmap$ and $\pmap$ to the set $A \setminus B$.
Then inductively we may assume a translation $\ytr_{\bbA^{-}}: (A \setminus B)
\to \mu\oslang(P)$, and we will define
\[
\ytr_{\bbA}(a) \isdef \ytr_{\bbA^{-}}(a), \quad\text{ if } a \in A \setminus B.
\]

To extend this definition to the states in $B$, we make a case 
distinction.
If $B$ is a degenerate cluster, that is, $B = \{ b \}$ for some state $b$ 
which is not active in itself, then we define
\[
\ytr_{\bbA}(b) \isdef
   \bigvee_{c \in \pow{\pprop}}
   \nxt{\tmap(b,c)}(\ytr_{\bbA^{-}}(a)/a \mid a \in A \setminus B).
\]
The main case of the definition is where $B$ is not degenerate.
Fix an enumeration $b_{1},\ldots,b_{n}$ of $B$ such that $i \leq j$ implies 
$\pmap(b_{i}) \leq \pmap(b_{j})$.
Let $\bbA_{k}$ be the generalized preautomaton\footnote{%
   Here we see the reason to generalise the notion of an automaton:
   in the structure $\bbA_{k}$ ($0 \leq k \leq n$) the objects $b_{k+1},\ldots,
   b_{n}$ are no longer states, but in the formulas $\tmap_{k}(a,c)$ they still 
   occur at the position of states.
}
obtained from $\bbA$ by restricting the transition and priority map to the set
\[
A_{k} \isdef (A\setminus B) \cup \{ b_{1},\ldots,b_{k} \},
\]
so that $\bbA^{0} = \bbA^{-}$ and $\bbA^{n} = \bbA$.
Where $B_{k} \isdef \{ b_{1},\ldots,b_{k} \}$, we now define, by induction on 
$k$, a map 
\[
\ytr^{k}: B \to \mu\oslang(P \cup (B \setminus B_{k})).
\]
In the base case of this definition we set
\[
\ytr^{0}(b) \isdef 
   \bigvee_{c \in \pow{\pprop}} 
   \nxt{\tmap(b,c)}(\ytr_{\bbA^{-}}(a)/a \mid a \in A \setminus B),
\]
and in the inductive case we first define $\eta_{k+1} \isdef \mu$ if
$\pmap(b_{k+1})$ is odd, and $\eta_{k+1} \isdef \nu$ if $\pmap(b_{k+1})$ is 
even, and then set
\[\begin{array}{llll}
     \ytr^{k+1}(b_{k+1}) &\isdef &
   \eta_{k+1} b_{k+1}. \ytr^{0}(b_{k+1})[\ytr^{k}(b_{i})/b_{i} \mid 1 \leq i\leq k]
\\ \ytr^{k+1}(b_{i}) &\isdef &
   \ytr^{k}(b_{i})[\ytr^{k+1}(b_{k+1})/b_{k+1}]
   & \text{ for } i \neq k+1.
\end{array}\]
Finally, we complete the definition of $\ytr_{\bbA}$ by putting
\[
\ytr_{\bbA}(b) \isdef \ytr^{n}(b),
\]
for any $b \in B$.
\end{definition}

In the proof of Theorem~\ref{t:autofor} we will need the following closure 
property of the fragments $\noe{\mu\oslang}{\qprop}$ and 
$\cocont{\mu\oslang}{\qprop}$.

\begin{proposition}
\label{p:comp}
Let $\rprop \subseteq \qprop$ be sets of proposition letters, and let $\phi$ 
and $\phi_{q}$, for each $q \in \qprop$, be formulas in $\mu\oslang$.

(1) If $\phi$ and each $\phi_{q}$ belongs to 
    $\noe{\mu\oslang}{\qprop\setminus\rprop}$ 
    ($\conoe{\mu\oslang}{\qprop\setminus\rprop}$), 
   then so does $\phi[\phi_{q}/q \mid q \in \qprop]$.

(2) If $\phi$ and each $\phi_{q}$ belongs to
   $\cont{\mu\oslang}{\qprop\setminus\rprop}$
   ($\cocont{\mu\oslang}{\qprop\setminus\rprop}$), 
   then so does $\phi[\phi_{q}/q \mid q \in \qprop]$.
\end{proposition}

Both items of this proposition can be proved by a straightforward formula 
induction --- we omit the details.

\begin{proofof}{Theorem~\ref{t:autofor}}
As mentioned, the verification of the equivalence of $\xi_{\bbA}$ and $\bbA$
is a standard exercise in the theory of parity automata and mu-calculi, and so
we omit the details.
We also skip the proof of item (2), completely focussing on the (harder) third 
item.

To prove this item, it suffices to take an arbitrary continuous-weak 
$\oslang$-preautomaton $\bbA = \tup{A,\tmap,\pmap}$ for the set $\pprop$,
and to show that 
\begin{equation}
\label{eq:tr1}
\ytr_{\bbA}(a) \in \mu_{C}\oslang(\pprop)
\end{equation}
for all $a \in A$.
We will prove this by induction on the number of clusters of $\bbA$.

Since there is nothing to prove in the base case of the proof, we immediately
move to the inductive case.
Let $B$ be the highest cluster of $\bbA$.
By the induction hypothesis we have $\ytr_{\bbA}(a) = \ytr_{\bbA^{-}}(a) \in 
\mu_{C}\oslang(\pprop)$ for all $a \in A \setminus B$, where 
$\bbA^{-}$ is as in Definition~\ref{d:tr}.
To show that \eqref{eq:tr1} also holds for all $b \in B$, we distinguish cases.

If $B$ is a degenerate cluster, say, $B = \{ b \}$, then for every $c \in 
\pow(\pprop)$, the variables occurring in the formula $\tmap(b,c) \in \oslang(A)$
are all from $A \setminus \{ b \}$.
Given the definition of $\ytr_{\bbA}(b)$, it suffices to show that all formulas 
of the form $\nxt{\tmap(b,c)}(\ytr_{\bbA}(a)/a \mid a \in A \setminus \{ b \})$
belong to the set $\mu_{C}\oslang(\pprop)$, but this is immediate by the 
induction hypothesis and the definition of the language.

If, on the other hand, $B$ is nondegenerate, let $b_{1},\ldots,b_{n}$ enumerate
$B$, and let, for $0\leq k \leq n$, the map $\ytr^{k}: B \to \mu\oslang$ be as in
Definition~\ref{d:tr}.
We only consider the case where $B$ is an odd cluster, i.e., $\pmap(b)$ is odd
for all $b \in B$.
Our key claim here is that 
\begin{equation}
\label{eq:tr2}
\ytr^{k}(b_{i}) \in \mu_{C}\oslang(\pprop \cup \{ b_{k+1},\ldots,b_{n}\}) \cap 
\cont{\mu\oslang}{\{ b_{k+1},\ldots,b_{n}\}},
\end{equation}
for all $k$ and $i$ with $0 \leq k \leq n$ and $0 < i \leq n$.
We will prove this statement by induction on $k$ --- this is the `inner' 
induction that we announced earlier on.

In the base case of this inner induction we need to show that $\ytr^{0}(b_{i})$ 
belongs to both $\mu_{C}\oslang(\pprop \cup B)$ and $\cont{\mu\oslang}{B}$.
Showing the first membership relation is straightforward; for the second, the 
key observation is that by our assumption on $\bbA$, every one-step formula of
the form $\tmap(b_{i},c)$ is syntactically continuous in every variable $b \in 
B$.
Furthermore, by the outer inductive hypothesis we have $\ytr_{\bbA^{-}}(a) \in
\mu_{C}\oslang(\pprop) \subseteq \cont{\mu\oslang}{B}$,
for every $a \in A \setminus B$, and we trivially have that every variable $b 
\in B$ belongs to the set $\cont{\mu\oslang}{B}$.
But then it is immediate by the definition of $\cont{\mu\oslang}{B}$ that 
this fragment contains the formula $\nxt{\tmap(b_{i},c)}(\ytr_{\bbA^{-}}(b)
\mid b \in B)$, and since this fragment is closed under taking disjunctions, 
we find that $\ytr^{0}(b_{i}) \in \cont{\mu\oslang}{B}$ indeed.
This finishes the proof of the base case of the inner induction.

For the inner induction step we fix a $k$ and assume that
\eqref{eq:tr2} holds for this $k$ and for all $i$ with $0 < i \leq k$.
We will prove that
\begin{equation}
\label{eq:tr3}
\ytr^{k+1}(b_{i}) \in \mu_{C}\oslang(\pprop) \cap 
    \cont{\mu\oslang}{\{ b_{k+2},\ldots,b_{n}\}}.
\end{equation}
first for $i = k+1$, and then for an arbitrary $i \neq k+1$.
To prove \eqref{eq:tr3} for the case $i = k+1$, first note that
\[\begin{array}{lcl}
\ytr^{0}(b_{k+1}) &\in& \mu_{C}\oslang(\pprop\cup B) 
   \cap \cont{\mu\oslang}{B},
\end{array}\]
as we just saw in the base case of the inner induction.
But then it is immediate by Proposition~\ref{p:comp} and the induction 
hypothesis on the formulas $\ytr^{k}(b_{i})$ that 
\[
\ytr^{0}(b_{k+1})[\ytr^{k}(b_{i})/b_{i} \mid 1 \leq i\leq k]
\in \cont{\mu\oslang}{\{ b_{k+1},\ldots,b_{n}\}},
\]
and from this it easily follows by the definition of 
$\cont{\mu\oslang}{\{ b_{k+1},\ldots,b_{n}\}}$ that 
\[
\mu b_{k+1} b_{k+1}. \ytr^{0}(b_{k+1})[\ytr^{k}(b_{i})/b_{i} \mid 1 \leq i\leq k]
\in \cont{\mu\oslang}{\{ b_{k+2},\ldots,b_{n}\}}
\]
as well. 
That is, 
\[
\ytr^{k+1}(b_{k+1}) \in \cont{\mu\oslang}{\{ b_{k+2},\ldots,b_{n}\}}.
\]
This is the crucial step in proving \eqref{eq:tr3} for the case $i=k+1$,
the proof that $\ytr^{k+1}(b_{k+1}) \in \mu_{C}\oslang(\pprop)$ is easy.

Second, to prove \eqref{eq:tr3} for the case $i \neq k+1$, we first 
recall that by the induction hypothesis we have
\[
\ytr^{k}(b_{i}) \in \mu_{C}\oslang(\pprop) \cap 
    \cont{\mu\oslang}{\{ b_{k+1},\ldots,b_{n}\}},
\]
while we just saw that 
$\ytr^{k+1}(b_{k+1}) \in \cont{\mu\oslang}{\{ b_{k+2},\ldots,b_{n}\}}$.
But from the latter two statements it is immediate by Proposition~\ref{p:comp} 
that 
\[
\ytr^{k+1}(b_{i}) = \ytr^{k}(b_{i})[\ytr^{k+1}(b_{k+1})/b_{k+1}] 
\in \cont{\mu\oslang}{\{ b_{k+2},\ldots,b_{n}\}}
\]
so that we have indeed proved \eqref{eq:tr2} for the case $i \neq 
k+1$.
This finishes the proof of the inner induction.

Finally, it follows from \eqref{eq:tr2}, instantiated with $k = n$, 
that for all $b \in B$ we have
\[
\ytr_{\bbA}(b) = \ytr^{n}(b) \in \mu_{C}\oslang(\pprop),
\]
as required for proving the outer induction step.
In other words, we are finished with the proof of \eqref{eq:tr1}, and 
hence, finished with the proof of the theorem.
\end{proofof}

\subsection{From Formulas to Automata}
   \label{sec:mc-to-parity}

In this subsection we focus on the meaning preserving translation in the
opposite direction, viz., from formula to automata.
In our set-up we need the one-step language to be closed under conjunctions and
disjunctions.

\begin{theorem}\label{t:fortoaut}
Let $\oslang$ be a on-step language that is closed under taking conjunctions and 
disjunctions.
Then there is an effective procedure that, given a formula $\xi \in \mu\oslang$
returns an automaton $\bbA_{\xi}$ in $\Aut(\oslang)$, which satisfies the 
following properties:
\begin{enumerate}[(1)]
\item $\bbA_\xi$ is equivalent to $\xi$;
\item  $\bbA_\xi \in \AutW(\oslang)$ if $\xi \in \mu_{D}\oslang$;
\item  $\bbA_\xi \in \AutWC(\oslang)$ if $\xi \in \mu_{C}\oslang$.
\end{enumerate}
\end{theorem}

As in the case of the translation from automata to formulas, the proof of part
(1) of this theorem, is a straightforward variation of the standard proof showing
that any fixpoint modal formula can be translated into an equivalent modal 
automaton (see for instance~\cite[Section 6]{Ven08}). 
The point is to show that this standard construction transforms formulas from 
$\mu_{D}\oslang$ and $\mu_{C}\oslang$ into automata of the right kind.
In fact, we will confine our attention to proving the (second and) third part
of the theorem; the fact that the input formula belongs to the alternation-free
fragment of $\mu\oslang$ enables a slightly simplified presentation of the 
construction.

Our first observation is that without loss of generality we may confine 
attention to \emph{guarded} formulas.

\begin{definition}
An occurrence of a bound variable $p$ in $\xi  \in \mu\oslang$ is called 
\emph{guarded} if there is a modal operator between its binding definition and 
the variable itself. 
A formula $\xi  \in \mu\oslang$ is called \emph{guarded} if every occurrence of
every bound variable is guarded.
\end{definition}

There is a standard construction, going back to \cite{Koz83}, which transforms
any formula $\xi$ in $\mu\oslang$ into an equivalent guarded $\xi^\flat \in 
\mu\oslang$, and it is easily verified that the construction $(\cdot)^{\flat}$
restricts to the fragments $\mu_{D}\oslang$ and $\mu_{C}\oslang$.
It therefore suffices to show that any guarded formula in $\mu_{D}\oslang$
($\mu_{C}\oslang$, respectively) can be transformed into an equivalent
continuous-weak $\oslang$-automaton.

In the remainder of this section we will show that any guarded formula $\xi$ 
in $\mu_{C}\oslang$ can be transformed into an equivalent continuous-weak 
$\oslang$-automaton $\bbA_{\xi}$; the analogous result for $\mu_{D}\oslang$ will
be obvious from our construction.

We will prove the result by induction on the so-called weak alternation depth
of the formula $\xi$.
We only consider the $\mu$-case of the inductive proof step; that is, we 
assume that $\xi$ can be obtained as a formula in the following grammar, for 
some set $\qprop$ of variables:
\begin{equation}
\label{eq:gr}
\phi \isbnf 
   q \mid \psi \mid \phi\lor\phi \mid \phi\land\phi \mid
   \nxt{\al}(\ol{\phi},\ol{\psi}) \mid \mu q. \phi,
\end{equation}
where $q \in \qprop$, every $\psi,\psi_{j}$ is $\qprop$-free, and 
$\al(\ol{a},\ol{b}) \in \cont{\oslang}{\ol{a}}$.
We write $\phi \subf' \xi$ if $\phi$ is a subformula of $\xi$ according to 
the grammar \eqref{eq:gr} (that is, we consider $\qprop$-free formulas $\psi$
to be atomic).
Furthermore, we inductively assume that for every $\qprop$-free formula $\psi
\subf' \xi$ we have already constructed an equivalent continuous-weak 
$\oslang$-automaton $\bbA_{\psi} = \tup{A_{\psi},a_{\psi},\tmap_{\psi},
\pmap_{\psi}}$.
(Observe that every subformula of a guarded formula it itself guarded, so that 
we are justified to apply the induction hypothesis.)

Define
$\Psi \isdef \{ \psi \subf' \xi \mid \psi\text{ $\qprop$-free} \}$, let $\Phi$
consist of all formulas $\phi \not\in \Psi$ that occur as some $\phi = \phi_{i}$
in a formula $\nxt{\al}(\ol{\phi},\ol{\psi}) \subf' \xi$, and set $A_{\Psi} 
\isdef \{ a_{\psi} \mid \psi \in \Psi \}$ and $A_{\Phi} \isdef \{ a_{\phi} \mid 
\phi \in \Phi \}$.
It now follows by guardedness of $\xi$ that there is a unique map $(\cdot)^{o}$ 
assigning a one-step formula $\chi^{o} \in \oslang(A_{\Psi} \cup A_{\Phi})$
to every formula $\chi \subf' \xi$ such that 
\[\begin{array}{lll}
   q^{o}    & = & \delta_{q}^{o}
\\ \psi^{o} & = & a_{\psi}
\\ (\phi_{0}\lor\phi_{1})^{o}  & = & \phi_{0}^{o} \lor \phi_{1}^{o}
\\ (\phi_{0}\land\phi_{1})^{o} & = & \phi_{0}^{o} \land \phi_{1}^{o}
\\ \nxt{\al}(\phi_{1},\ldots,\phi_{n},\psi_{1},\ldots,\psi_{k})^{o} & = &
   \al(a_{\phi_{1}},\ldots,a_{\phi_{n}},a_{\psi_{1}},\ldots,a_{\psi_{k}})
\\ (\mu q. \delta_{q})^{o}     & = & \delta_{q}^{o},
\end{array}\]
where $\delta_{q}$ is the unique formula $\delta$ such that $\mu q.\delta 
\subf' \xi$.

It is easy to verify that for every $\chi \subf' \xi$, the formula $\chi^{o}$ 
is continuous in $A_{\Phi}$.
We may thus assume without loss of generality that every $\chi^{o}$ belongs to
the syntactic fragment $\cont{\oslang}{A_{\Phi}}$.
(Should this not be the case, then by Theorem~\ref{t:osnf-cont} we may replace
$\chi^{o}$ with an equivalent formula that does belong to this fragment.)
We are now ready to define the automaton $\bbA_{\xi} = \tup{A_{\xi},a_{I},
\tmap_{\xi},\pmap_{\xi}}$ by putting 
\[\begin{array}{lll}
A_{\xi} &\isdef& \{ a_{\xi} \} \cup A_{\Phi} \cup \bigcup_{\psi\subf'\xi}A_{\psi}
\\ a_{I}    &\isdef& a_{\xi}
\\ \tmap_{\xi}(a,c) &\isdef& 
   \left\{\begin{array}{ll}
    \chi^{o} & \text{if } a = a_{\chi} \text{ with } \chi \in \{ \xi \} \cup \Phi
   \\ \tmap_{\psi}(a,c) & \text{if } a \in A_{\psi}.
   \end{array}\right.
\\ \pmap_{\xi}(a) &\isdef& 
   \left\{\begin{array}{ll}
   1 & \text{if } a \in \{ a_{\xi} \} \cup A_{\Phi}
   \\ \pmap_{\psi}(a) & \text{if } a \in A_{\psi}.
   \end{array}\right.
\end{array}\]
It is then straightforward to check that $\bbA_{\xi}$ is a continuous-weak
$\oslang$-automaton.
Proving the equivalence of $\xi$ and $\bbA_{\xi}$ is a routine exercise.


\section{Automata for $\wmso$}
   \label{sec:autwmso}

In this section we start looking at the automata-theoretic characterisation of
$\wmso$.
That is, we introduce the following automata, corresponding to this version of
monadic second-order logic; these \emph{$\wmso$-automata} are the continuous-weak
automata for the one-step language $\ofoei$, cf.~Definition~\ref{d:ctwk}.

\begin{definition}
A \emph{$\wmso$-automaton} is a continuous-weak automaton for the one-step
language $\ofoei$.
\end{definition}

Recall that our definition of continuous-weak automata is syntactic in nature,
i.e., if $\bbA = \tup{A,\tmap,\pmap,a_I}$ is a $\wmso$-automaton, then for any
pair of states $a,b$ with $a \ord b$ and $b \ord a$, and any $c\in C$, we have
$\tmap(a,c) \in \cont{{\ofoei(A)}^+}{b}$ if $\pmap(a)$ is odd and $\tmap(a,c)
\in \cocont{{\ofoei(A)}^+}{b}$ if $\pmap(a)$ is even.

The main result of this section states one direction of the automata-theoretic
characterisation of $\wmso$.

\begin{theorem}
\label{t:wmsoauto}
There is an effective construction transforming a $\wmso$-formula $\phi$
into a $\wmso$-automaton $\bbA_{\phi}$ that is equivalent
to $\phi$ on the class of trees.
\end{theorem}

The proof proceeds by induction on the complexity of
$\phi$. For the inductive steps, we will need to verify that the class of
$\wmso$-automata is closed under the boolean operations and finite projection.
The latter closure property requires most of the work: we devote
Section \ref{sec:simulationwmso} to a simulation theorem that puts
$\wmso$-automata in a suitable shape for the projection construction.
To this aim, it is convenient to define a closure operation on classes of 
tree models corresponding to the semantics of $\wmso$ quantification.
The inductive step of the proof of Theorem \ref{t:wmsoauto} will show that 
the classes that are recognizable by $\wmso$-automata are closed under this 
operation.

\begin{definition}\label{def:tree_finproj-w}
Fix a set $\pprop$ of proposition letters, a proposition letter $p \not\in P$ 
and a language $\mathsf{C}$ of $\pprop\cup\{p\}$-labeled trees.
The \emph{finitary projection} of $\mathsf{C}$ over $p$ is the language of 
$\pprop$-labeled trees defined as 
\[
{\finexists} p.\mathsf{C} \df \{\bbT \mid
\text{ there is a finite $p$-variant } \bbT' \text{ of } \bbT \text{ with }
\bbT' \in \mathsf{C}\}.
\]
A collection of classes of tree models is \emph{closed under finitary 
projection over $p$} if it contains the class ${{\finexists} p}.\mathsf{C}$ 
whenever it contains the class $\mathsf{C}$ itself.
\end{definition}

\subsection{Simulation theorem for $\wmso$-automata}
\label{sec:simulationwmso}

\noindent
Our next goal is a \emph{projection construction} that, given a $\wmso$-automaton
$\bbA$, provides one recognizing ${\finexists p}.\TMod(\bbA)$.
For $\smso$-automata, the analogous construction crucially uses the following
\emph{simulation theorem}: every $\smso$-automaton $\bbA$ is equivalent to a
\emph{non-deterministic} automaton $\bbA'$ \cite{Walukiewicz96}.
Semantically, non-determinism yields the appealing property that every node of
the input model $\bbT$ is associated with at most one state of $\bbA'$ during
the acceptance game--- that means, we may assume $\eloise$'s strategy $f$ in
$\agame(\bbA',\bbT)$ to be \emph{functional} (\emph{cf.}
Definition \ref{def:StratfunctionalFinitary} below).
This is particularly helpful in case we want to define a $p$-variant of $\bbT$
that is accepted by the projection construct on $\bbA'$: our decision whether
to label a node $s$ with $p$ or not, will crucially depend on the value
$f(a,s)$, where $a$ is the unique state of $\bbA'$ that is associated with $s$.
Now, in the case of $\wmso$-automata we are interested in guessing
\emph{finitary} $p$-variants, which requires $f$ to be functional only on a
\emph{finite} set of nodes.
Thus the idea of our simulation theorem is to turn a $\wmso$-automaton $\bbA$
into an equivalent one $\bbA^{\f}$ that behaves non-deterministically on a
\emph{finite} portion of any accepted tree.

For $\smso$-automata, the simulation theorem is based on a powerset construction:
if the starting automaton has carrier $A$, the resulting non-deterministic
automaton is based on ``macro-states'' from the set $\shA$.
Analogously, for $\wmso$-automata we will associate the non-deterministic
behaviour with macro-states.
However, as explained above, the automaton $\bbA^{\f}$ that we construct has to
be non-deterministic just on finitely many nodes of the input and may behave as
$\bbA$ (i.e. in ``alternating mode'') on the others.
To this aim, $\bbA^{\f}$ will be ``two-sorted'', roughly consisting of a copy of
$\bbA$ (with carrier $A$) together with a variant of its powerset construction,
based both on $A$ and $\shA$.
For any accepted $\bbT$, the idea is to make any match $\pi$ of
$\mc{A}(\bbA^{\f},\bbT)$ consist of two parts:
\begin{description}
\item[(\textit{Non-deterministic mode})] For finitely many rounds $\pi$ is
  played on macro-states, i.e. positions belong to the set $\shA \times T$.
  In her strategy player $\exists$ assigns macro-states (from $\shA$) only to
  \emph{finitely many} nodes, and states (from $A$) to the rest.
  Also, her strategy is functional in $\shA$, i.e. it assigns \emph{at most one
  macro-state} to each node.
\item[(\textit{Alternating mode})] At a certain round, $\pi$ abandons
  macro-states and turns into a match of the game $\mc{A}(\bbA,\bbT)$, i.e. all
  subsequent positions are from $A \times T$ (and are played according to a
  not necessarily functional strategy).
\end{description}
Therefore successful runs of $\bbA^{\noet}$ will have the property of processing 
only a \emph{finite} amount of the input with $\bbA^{\noet}$ being in a 
macro-state and all the rest with $\bbA^{\noet}$ behaving exactly as $\bbA$.
We now proceed in steps towards the construction of $\bbA^{\noet}$. 
First, recall from Definition \ref{def:basicform-ofoe} that a \emph{$A$-type} is
just a subset of $A$. 
We now define a notion of liftings for sets of types, which is instrumental in
translating the transition function from states on macro-states.

\begin{definition}
The \emph{lifting} of a type $S \in \pow A$ is defined as the following
$\pow A$-type:
\[
\lift{S} \isdef 
\begin{cases} \{ S \} & \text{ if } S \neq \nada \\
\nada & \text{ if } S = \nada.
\end{cases}
\]
This definition is extended to sets of $A$-types by putting $\lift{\Sigma} 
\isdef \{ \lift{S} \mid S \in \Sigma \}$.
\end{definition}
The distinction between empty and non-empty elements of $\Sigma$ is to ensure 
that the empty type on $A$ is lifted to the empty type on $\pow A$. 
Notice that the resulting set $\lift{\Sigma}$ is either empty or 
contains exaclty one $\pow A$-type.
This property is important for functionality, see below.

Next we define a translation on the sentences associated with the transition
function of the original $\wmso$-automaton.
Following the intuition given above, we want to work with sentences that can be
made true by assigning macro-states (from $\shA$) to finitely many nodes in the
model, and ordinary states (from $A$) to all the other nodes.
Moreover, each node should be associated with \emph{at most one} macro-state,
because of functionality. 
These desiderata are expressed for one-step formulas as \emph{$\shA$-continuity} 
and \emph{$\shA$-separability}, see the Definitions~\ref{def:semnotions} 
and~\ref{d:sep}.
For the language $\ofoei$, Theorem \ref{t:osnf} and Proposition~\ref{p:sep}
guarantee these properties when formulas are in a certain syntactic shape.
The next definition will provide formulas that conform to this particular shape.

\begin{definition}\label{DEF_finitary_lifting}
Let $\phi \in {\ofoei}^+(A)$ be a formula of shape 
$\posdbnfofoei{\vlist{T}}{\Pi}{\Sigma}$ for some $\Pi,\Sigma \subseteq \shA$ 
and $\vlist{T} = \{T_1,\dots,T_k\} \subseteq \shA$. 
We define $\phi^{\fin} \in {\ofoei}^+(A \cup \shA)$ as the formula
$\posdbnfofoei{\lift{\vlist{T}}}{\lift{\Pi} \cup \lift{\Sigma}}{\Sigma}$, that
means,
\begin{equation}\label{eq:unfoldingNablaolque}
\begin{aligned}
\phi^{\fin} \df\ &
    \exists \vlist{x}.\Big(\arediff{\vlist{x}}
      \land \bigwedge_{0 \leq i \leq n} \tau^+_{\lift{T}_i}(x_i)
\land
    \forall z.(\arediff{\vlist{x},z} \to
    \bigvee_{S\in \lift{\Pi} \cup \lift{\Sigma} \cup \Sigma}
       \tau^+_S(z))\Big)
\\ &
    \land \bigwedge_{P\in\Sigma} \qu y.{\tau}^{+}_P(y)
 \land
    \dqu y.\bigvee_{P\in\Sigma} {\tau}^{+}_P(y)
    \end{aligned}
\end{equation}
\end{definition}

We combine the previous definitions to form the transition function for
macro-states.

\begin{definition}\label{PROP_DeltaPowerset}
Let $\bbA = \tup{A,\tmap,\pmap,a_I}$ be a $\wmso$-automaton. 
Fix $c \in C$ and $Q \in \shA$. 
By Theorem \ref{t:osnf}, for some $\Pi,\Sigma \subseteq
\shA$ and $T_i \subseteq A$, there is a sentence $\Psi_{Q,c} \in {\ofoei}^+(A)$
in the basic form $\bigvee \posdbnfofoei{\vlist{T}}{\Pi}{\Sigma}$ such that
$\bigwedge_{a \in Q} \tmap(a,c) \equiv \Psi_{Q,c}$. 
By definition $\Psi_{Q,c}$ is of the form $\bigvee_{i}\phi_i$, with each
$\phi_{i}$ of shape $\posdbnfofoei{\vlist{T}}{\Pi}{\Sigma}$. 
We put $\shDe(Q,c) \isdef \bigvee_{i}\phi_i^{\fin}$, where the translation
$(-)^{\fin}$ is given as in Definition~\ref{DEF_finitary_lifting}. 
Observe that $\shDe(Q,c)$ is of type ${\ofoei}^+(A \cup \shA)$.
\end{definition}

We have now all the ingredients to define our two-sorted automaton.

\begin{definition}\label{def:finitaryconstruct}
Let $\bbA = \tup{A,\tmap,\pmap,a_I}$ be a {\wmso-automaton}.
We define the \emph{finitary construct over $\bbA$} as the automaton
$\bbA^{\fin} = \tup{A^{\fin},\tmap^{\fin},\pmap^{\fin},a_I^{\fin}}$ given by
\[
\begin{array}{lll}
   A^{\fin}   &\df&  A \cup \shA
\\ a_I^{\fin} &\df&  \{a_I\}
\end{array}
\hspace*{5mm}
\begin{array}{lll}
   \pmap^{\fin}(a) &\df& \pmap(a)
\\ \pmap^{\fin}(R) &\df& 1
\end{array}
\hspace*{5mm}
\begin{array}{lll}
   \tmap^{\fin}(a,c) &\df& \tmap(a,c)
\\ \tmap^{\fin}(Q,c) &\df&
  \shDe(Q,c) \vee \bigwedge_{a \in Q} \! \! \tmap(a,c).
\end{array}
\]
\end{definition}

\begin{remark} 
In the standard powerset construction of non-deterministic parity automata
(\cite{Walukiewicz02}, see also \cite{Ven08,ArnoldN01})
macro-states are required to be \emph{relations} rather than sets in order to
determine whether a run through macro-states is accepting. 
This is not needed in our construction: macro-states will never be visited
infinitely often in accepting runs, thus they may simply be assigned the
priority $1$.
\end{remark}

The idea behind this definition is that $\bbA^{\fin}$ is enforced to process only a finite portion of any accepted tree while in the non-deterministic mode. This is encoded in game-theoretic terms through the notion of functional and finitary strategy.

\begin{definition}\label{def:StratfunctionalFinitary}
Given a $\wmso$-automaton $\bbA = \tup{A,\tmap,\pmap,a_I}$ and transition system $\bbT$, a strategy $f$ for \eloise in $\mathcal{A}(\bbA,\bbT)$ is \emph{functional in $B \subseteq A$} (or simply functional, if $B=A$) if for each node $s$ in $\bbT$ there is at most one $b \in B$ such that $(b,s)$ is a reachable position in an $f$-guided match. Also $f$ is \emph{finitary} in $B$ if there are only finitely many nodes $s$ in $\bbT$ for which a position $(b,s)$ with $b \in B$ is reachable in an $f$-guided match.
\end{definition}

The next proposition establishes the desired properties of the finitary
construct.

\begin{theorem}[Simulation Theorem for $\wmso$-automata]
\label{PROP_facts_finConstrwmso}
Let $\bbA$ be a $\wmso$-automaton and $\bbA^{\fin}$ its finitary construct.
\begin{enumerate}[(1)]
  \itemsep 0 pt
\item $\bbA^{\fin}$ is a $\wmso$-automaton.\label{point:finConstrAut-w}
\item 
For any tree model $\bbT$, if $(a_I^{\fin},s_I)$ is a winning position for 
$\eloise$ in $\agame(\bbA^{\fin},\bbT)$ , then she has a winning strategy that
is both functional and finitary in $\shA$.
  \label{point:finConstrStrategy}
\item $\bbA \equiv \bbA^{\fin}$. \label{point:finConstrEquiv}
\end{enumerate}
\end{theorem}

\begin{proof}
\begin{enumerate}[(1)]
\item
Observe that any cluster of $\bbA^{\fin}$ involves states of exactly one sort,
either $A$ or $\shA$.
For clusters on sort $A$, weakness and continuity of $\bbA^{\fin}$ follow by 
the same properties of $\bbA$.
For clusters on sort $\shA$, weakness follows by observing that all macro-states
in $\bbA^{\fin}$ have the same priority.
Concerning continuity, by definition of $\tmap^{\fin}$ any macro-state can only
appear inside a formula of the form $\phi^{\fin} =
\posdbnfofoei{\lift{\vlist{T}}}{\lift{\Pi} \cup \lift{\Sigma}}{\Sigma}$ as in
\eqref{eq:unfoldingNablaolque}.
Because $\shA \cap \bigcup\Sigma = \nada$, by Theorem \ref{t:osnf-cont}
$\phi^{\fin}$ is continuous in each $Q \in \shA$.

\item
Let $f$ be a (positional) winning strategy for $\eloise$ in $\mathcal{A}
(\bbA^{\fin},\bbT)@(a_I^{\fin},s_I)$.
We define a strategy $f'$ for $\eloise$ in the same game as follows:
\begin{enumerate}[label=(\alph*),ref=\alph*]
\item 
\label{point:stat2point1}
On basic positions of the form $(a,s) \in A\times T$, let $V: A \to \pow R[s]$
be the valuation suggested by $f$.
We let the valuation suggested by $f'$ be the restriction $V'$ of $V$ to $A$.
Observe that, as no predicate from $A^{\fin}\setminus A =\shA$ occurs in
$\tmap^{\fin}(a,\V(s)) = \tmap(a,\V(s))$, then $V'$ also makes that sentence
true in $\R{s}$.

\item
For winning positions of the form $(R,s) \in \shA \times T$, let $V_{R,s}: 
(\pow A \cup A) \to \pow R[s]$ be the valuation suggested by $f$.
As $f$ is winning, $\tmap^{\fin}(R,\V(s))$ is true in the model $V_{R,s}$.
If this is because the disjunct $\bigwedge_{a \in R} \tmap(a,\V(s))$ is made
true, then we can let $f'$ suggest the restriction to $A$ of $V_{R,s}$,
for the same reason as in \eqref{point:stat2point1}.
  
Otherwise, the disjunct $\shDe(R,\V(s)) = \bigvee_{i}\phi_i^{\fin}$ is made
true.
This means that, for some $i$, $(R[s], V_{R,s}) \models \phi_i^{\fin}$.
Now, by construction of $\phi_i^{\fin}$ as in \eqref{eq:unfoldingNablaolque},
we have $\shA \cap \bigcup\Sigma = \nada$.
By Theorem \ref{t:osnf-cont}, this implies that $\phi_i^{\fin}$ is continuous
in $\shA$.
Thus we have a restriction $V_{R,s}'$ of $V_{R,s}$ that verifies $\phi_i^{\fin}$ 
and assigns only finitely many nodes to predicates from $\shA$. 
Moreover, by construction of $\phi_i^{\fin}$, for each $S \in
\{\lift{T}_1,\dots,\lift{T}_k\}\cup \in \lift{\Pi} \cup \lift{\Sigma}$,
$S$ contains at most one element from $\shA)$. 
Thus, by Proposition~\ref{p:sep}, $\phi_i^{\fin}$ is $\shA$-separable. 
But then we may find a separating valuation $V_{R,s}''\leq_{\shA} V_{R,s}''$ 
such that $V_{R,s}''$ verifies $\phi_i^{\fin}$. 
Separation means that $V_{R,s}''$ associates with each node at most one
predicate from $\shA$, and the fact that $V_{R,s}''\leq_{\shA} V_{R,s}''$,
combined with the $\shA$-continuity of $V_{R,s}'$ ensures $\shA$-continuity of 
$V_{R,s}''$. 
In this case we let $f'$ suggest $V_{R,s}''$ at position $(R,s)$.
\end{enumerate}
The strategy $f'$ defined as above is immediately seen to be surviving for
$\eloise$. 
It is also winning, since at every basic winning position for $\eloise$, the set
of possible next basic positions offered by $f'$ is a subset of those 
offered by $f$.
By this observation it also follows that any $f'$-guided match visits basic
positions of the form $(R,s) \in \shA \times C$ only finitely many times, as
those have odd parity.
By definition, the valuation suggested by $f'$ only assigns finitely many nodes
to predicates in $\shA$ from positions of that shape, and no nodes from other
positions. It follows that $f'$ is finitary in $\shA$. Functionality in $\shA$ 
also follows immediately by definition of $f'$.
\item 
For the direction from left to right, it is immediate by definition of 
$\bbA^{\fin}$ that a winning strategy for $\eloise$ in $\mc{G} = 
\mathcal{A}(\bbA,\bbT)@(a_I,s_I)$ is also winning for $\eloise$ in 
$\mc{G}^{\fin} = \mathcal{A}(\bbA^{\fin},\bbT)@(a_I^{\fin},s_I)$.
\smallskip

For the direction from right to left, let $f$ be a winning strategy for 
$\eloise$ in $\mc{G}^{\fin}$.
The idea is to define a strategy $f'$ for $\eloise$ in stages, while playing a
match $\pi'$ in $\mc{G}$. 
In parallel to $\pi'$, a shadow match $\pi$ in $\mc{G}^{\fin}$ is maintained, 
where $\eloise$ plays according to the strategy $f$. 
For each round $z_i$, we want to keep the following relation between the two 
matches:
\smallskip
\begin{center}
\fbox{\parbox{12cm}{
Either
\begin{enumerate}[label=(\arabic*),ref=\arabic*]
  \item positions of the form $(Q,s) \in \shA \times T$ and $(a,s) \in A \times T$ occur respectively in $\pi$ and $\pi'$, with $a \in Q$,
\end{enumerate}
or
\begin{enumerate}[label=(\arabic*),ref=\arabic*]
  \item[(2)] the same position of the form $(a,s) \in A \times T$ occurs in both matches.
\end{enumerate}
}}\hspace*{0.3cm}($\ddag$)
\end{center}
\smallskip
The key observation is that, because $f$ is winning, a basic position of the 
form $(Q,s) \in \shA \times T$ can occur only for finitely many initial rounds
$z_0,\dots,z_n$ that are played in $\pi$, whereas for all successive rounds 
$z_n,z_{n+1},\dots$ only basic positions of the form $(a,s) \in A \times T$ are
encountered. 
Indeed, if this was not the case then either $\eloise$ would get stuck or the
highest priority occurring infinitely often would be odd, since states from 
$\shA$ all have priority $1$.

It follows that enforcing a relation between the two matches as in ($\ddag$) suffices to prove that the defined strategy $f'$ is winning for $\eloise$ in $\pi'$. For this purpose, first observe that $(\ddag).1$ holds at the initial round, where the positions visited in $\pi'$ and $\pi$ are respectively $(a_I,s_I) \in A \times T$ and $(\{a_I\},s_I) \in A^{\fin} \times T$. Inductively, consider any round $z_i$ that is played in $\pi'$ and $\pi$, respectively with basic positions $(a,s) \in A \times T$ and $(q,s) \in A^{\fin} \times T$. To define the suggestion of $f'$ in $\pi'$, we distinguish two cases.
\begin{itemize}

\item 
First suppose that $(q,s)$ is of the form $(Q,s) \in \shA\times T$. 
By ($\ddag$) we can assume that $a$ is in $Q$. 
Let $V_{Q,s} :A^{\fin} \rightarrow \pow(\R{s})$ be the valuation suggested by 
$f$, verifying the sentence $\tmap^{\fin}(Q,\V(s))$. 
We distinguish two further cases, depending on which disjunct of 
$\tmap^{\fin}(Q,\V(s))$ is made true by $V_{Q,s}$.
\begin{enumerate}[label=(\roman*), ref=\roman*]
\item 
\label{point:valuation1}
If $(\R{s},V_{Q,s})\models \bigwedge_{b \in Q} \tmap(b,\V(s))$, then we let 
$\eloise$ pick the restriction to $A$ of the valuation $V_{Q,s}$. 
\item 
\label{point:valuation2}
If $(\R{s},V_{Q,s})\models \shDe(Q,\V(s))$, we let $\eloise$ pick a valuation
$V_{a,s}:A \rightarrow \p (\R{s})$ defined by putting, for each $b \in A$:
\begin{align*}
V_{a,s}(b) \isdef \bigcup_{b \in Q'} \{t \in \R{s} \mid t \in V_{Q,s}(Q')\}
               \cup  \{t \in \R{s} \mid t \in V_{Q,s}(b)\} .
\end{align*} 
\end{enumerate}
It can be readily checked that the suggested move is legitimate for $\eloise$
in $\pi$, i.e. it makes $\tmap(a,\V(s))$ true in $\R{s}$.

For case \eqref{point:valuation2}, observe that the nodes assigned to $b$ by
$V_{Q,s}$ have to be assigned to $b$ also by $V_{a,s}$, as they may be necessary
to fulfill the condition, expressed with $\qu$ and $\dqu$ in $\shDe$, that 
infinitely many nodes witness (or that finitely many nodes do not witness) 
some type.

We now show that $(\ddag)$ holds at round $z_{i+1}$. 
If \eqref{point:valuation1} is the case, any next position $(b,t)\in A \times T$
picked by player $\forall$ in $\pi'$ is also available for $\forall$ in $\pi$, 
and we end up in case $(\ddag .2)$. 
Suppose instead that \eqref{point:valuation2} is the case. 
Given a move $(b,t) \in A \times T$ by $\forall$, by definition of $V_{a,s}$ 
there are two possibilities. 
First, $(b,t)$ is also an available choice for $\forall$ in $\pi$, and we end up
in case $(\ddag .2)$ as before. 
Otherwise, there is some $Q' \in \shA$ such that $b$ is in $Q'$ and $\forall$ 
can choose $(Q',t)$ in the shadow match $\pi$. By letting $\pi$ advance at round
$z_{i+1}$ with such a move, we are able to maintain $(\ddag .1)$ also in
$z_{i+1}$.
\item 
In the remaining case, inductively we are given the same basic position $(a,s)
\in A\times T$ both in $\pi$ and in $\pi'$. 
The valuation $V$ suggested by $f$ in $\pi$ verifies $\tmap^{\fin}(a,\V(s)) = 
\tmap(a,\V(s))$, thus we can let the restriction of $V$ to $A$ be the valuation
chosen by $\eloise$ in the match $\pi'$. 
It is immediate that any next move of $\forall$ in $\pi'$ can be mirrored by the
same move in $\pi$, meaning that we are able to maintain the same position 
--whence the relation $(\ddag.1)$-- also in the next round.
\end{itemize}
In both cases, the suggestion of strategy $f'$ was a legitimate move for 
$\eloise$ maintaining the relation $(\ddag)$ between the two matches for any
next round $z_{i+1}$.
It follows that $f'$ is a winning strategy for $\eloise$ in $\mc{G}$.
\end{enumerate}
\end{proof}

\subsection{From formulas to automata}

In this subsection we conclude the proof of Theorem~\ref{t:wmsoauto}.
We first focus on the case of projection with respect to finite sets, which
exploits our simulation result, Theorem~\ref{PROP_facts_finConstrwmso}.
The definition of the projection construction is formulated more generally for
parity automata, as it will be later applied to classes other than
$\AutWC(\ofoei)$.
It clearly preserves the weakness and continuity conditions.


\begin{definition}\label{DEF_fin_projection}
Let $\bbA = \tup{A, \tmap, \Omega, a_I}$ be a parity automaton on alphabet 
$\p(\pprop \cup \{p\})$.
We define the automaton ${{\exists} p}.\bbA = \tup{A, \tmapProj, \Omega, a_I}$ 
on alphabet $\p\pprop$ by putting
\begin{equation*}
  \tmapProj(a,c) \ \df \ \tmap(a,c) \qquad \qquad
  \tmapProj(Q,c) \ \df \ \tmap(Q,c) \vee \tmap(Q,c\cup\{p\}).
\end{equation*}
The automaton ${{\exists} p}.\bbA$ is called the \emph{finitary projection
construct of $\bbA$ over $p$}.
\end{definition}

\begin{lemma}\label{PROP_fin_projection}
Let $\bbA$ be a $\wmso$-automaton on alphabet $\p (\pprop \cup \{p\})$.
Then $\bbA^{\fin}$ is a $\wmso$-automaton on alphabet $\p\pprop$ which satisfies
$$\TMod({{\exists} p}.\bbA^{\fin}) \equiv
{{\finexists} p}.\TMod(\bbA).$$
\end{lemma}

\begin{proof}
Unraveling definitions, we need to show that for any tree $\bbT = \tup{T,R,\V \:
\pprop \to \pow T,s_I}$:
$${{\exists} p}.\bbA^{\fin} \text{ accepts } \bbT \text{ iff } \text{there is a
finite }p \text{ -variant }\bbT' \text{of } \bbT \text{  such that } \bbA 
\text{  accepts } \bbT'.
$$
For the direction from left to right, by the equivalence between $\bbA$ and 
$\bbA^{\fin}$ it suffices to show that if ${{\exists} p}.\bbA^{\fin}$ accepts
$\bbT$ then there is a finite $p$-variant $\bbT'$ of $\bbT$ such that 
$\bbA^{\fin}$ accepts $\bbT'$. 
First, we first observe that the properties stated by 
Theorem~\ref{PROP_facts_finConstrwmso}, which hold for $\bbA^{\fin}$ by 
assumption, by construction hold for ${{\exists} p}.\bbA^{\fin}$ as well. 
Thus we can assume that the given winning strategy $f$ for $\eloise$ in 
$\mc{G_{\exists}} = \mc{A}({\finexists p}.\bbA^{\fin},\bbT)@(a_I^{\fin},s_I)$ 
is functional and finitary in $\shA$. 
Functionality allows us to associate with each node $s$ either none or a unique 
state $Q_s \in \shA$ such that $(Q_s,s)$ is winning for $\eloise$. 
We now want to isolate the nodes that $f$ treats ``as if they were labeled with 
$p$''. 
For this purpose, let $V_{s}$ be the valuation suggested by $f$ from a position
$(Q_s,s) \in \shA \times T$. As $f$ is winning, $V_{s}$ makes 
$\tmapProj(Q,\tscolors(s))$ true in $\R{s}$. 
We define a $p$-variant $\bbT' = \tup{T,R,\V' \: \pprop\cup\{p\} \to \pow T,s_I}$
of $\bbT$ by defining $\tscolors' \isdef \tscolors[p \mapsto X_{p}]$, that is, 
by colouring with $p$ all nodes in the following set:
\begin{equation}\label{eq:X_p}
X_p \isdef \{s \in T\mid (\R{s},V_{s}) \models 
\tmap^{\f}(Q_s,\tscolors(s)\cup\{p\})\}.
\end{equation}
The fact that $f$ is finitary in $\shA$ guarantees that $X_p$ is finite, whence 
$\bbT'$ is a finite $p$-variant.
It remains to show that $\bbA^{\fin}$ accepts $\bbT'$: we claim that $f$ itself
is winning for $\eloise$ in $\mc{G} = (\bbA^{\fin},\bbT')@(a_I,s_I)$. 
In order to see that, let us construct in stages an $f$-guided match $\pi$ 
of $\mc{G}$ and an $f$-guided shadow match $\tilde{\pi}$ of $\mc{G_{\exists}}$.
The inductive hypothesis we want to bring from one round to the next is that the 
same basic position occurs in both matches, as this suffices to prove that $f$ 
is winning for $\eloise$ in $\mc{G}$.

First we consider the case of a basic position $(Q,s) \in A^{\fin} \times T$ 
where $Q \in \shA$. 
By assumption $f$ provides a valuation $V_s$ that makes $\tmapProj(Q,\V(s))$ true
in $\R{s}$. 
Thus $V_s$ verifies either $\tmap^{\fin}(Q,\V(s))$ or $\tmap^{\fin}(Q,\V(s)\cup 
\{p\})$. 
Now, the match $\pi^{\fin}$ is played on the $p$-variant $\bbT'$, where the 
labeling $\V'(s)$ is decided by the membership of $s$ to $X_p$. 
According to \eqref{eq:X_p}, if $V_s$ verifies 
$\tmap^{\fin}(Q,\V(s)\cup \{p\})$ then $s$ is in $X_p$, meaning that it is
labeled with $p$ in $\bbT'$, i.e. $\V'(s) = \V(s)\cup \{p\}$. 
Therefore $V_s$ also verifies $\tmap^{\fin}(Q,\V'(s))$ and it is a 
legitimate move for $\eloise$ in match $\pi^{\fin}$. 
In the remaining case, $V_s$ verifies $\tmap^{\fin}(Q,\V(s))$ but 
falsifies $\tmap^{\fin}(Q,\V(s)\cup \{p\})$, implying by definition that $s$ 
is not in $X_p$. 
This means that $s$ is not labeled with $p$ in $\bbT'$, i.e. $\V'(s) = \V(s)$.
Thus again $V_s$ verifies $\tmap^{\fin}(Q,\V'(s))$ and it is a 
legitimate move for $\eloise$ in match $\pi^{\fin}$.

It remains to consider the case of a basic position $(a,s) \in A^{\fin} \times T$
with $a \in A$ a state. 
By definition $\tmapProj(a,\V(s))$ is just $\tmap^{\fin}(a,\V(s))$. 
As $(a,s)$ is winning, we can assume that no position $(Q,s)$ with $Q$ a 
macro-state is winning according to the same $f$, as making $\tmapProj$-sentences
true never forces $\eloise$ to mark a node both with a state and a macro-state. 
Therefore, $s$ is not in $X_p$ either, meaning that it it is not labeled with 
$p$ in the $p$-variant $\bbT'$ and thus $\V'(s) = \V(s)$. 
This implies that $f$ makes $\tmap^{\fin}(a,\V'(s)) = \tmap^{\fin}(a,\V(s))$
true in $\R{s}$ and its suggestion is a legitimate move for $\eloise$ in match
$\pi^{\fin}$.
In order to conclude the proof, observe that for all positions that we consider
the same valuation is suggested to $\eloise$ in both games: this means that any 
next position that is picked by player $\abelard$ in $\pi^{\fin}$ is also 
available for $\abelard$ in the shadow match $\tilde{\pi}$.
\medskip

We now show the direction from right to left of the statement. Let $\bbT'$ be a finite $p$-variant of
$\bbT$, with labeling function $\tscolors'$, and $g$ a winning strategy for $\exists$ in $\mc{G} = \mathcal{A}(\bbA,\bbT')@(a_I,s_I)$. Our goal is to define a strategy $g'$ for $\exists$ in $\mc{G_{\exists}}$. 
As usual, $g'$ will be constructed in stages, while playing a match $\pi'$ in $\mc{G_{\exists}}$. In parallel to $\pi'$, a \emph{bundle} $\mc{B}$ of $g$-guided shadow matches in $\mc{G}$ is maintained, with the following condition enforced for each round $z_i$:
\smallskip
\begin{center}
\fbox{\parbox{13.5cm}{
\begin{enumerate}
  \item If the current basic position in $\pi'$ is of the form $(Q,s) \in \shA \times T$, then for each $a \in Q$ there is an $g$-guided (partial) shadow match $\pi_a$ at basic position $(a,s) \in A\times T$ in the current bundle $\mc{B}_i$. Also, either $\bbT'_s$ is not $p$-free (i.e., it does contain a node $s'$ with $p \in \tscolors'(s')$) or $s$ has some sibling $t$ such that $\bbT'_t$ is not $p$-free.
  \item Otherwise, the current basic position in $\pi'$ is of the form $(a,s) \in A \times T$ and $\bbT'_s$ is $p$-free. Also, the bundle $\mc{B}_i$ only consists of a single $g$-guided match $\pi_a$ whose current basic position is also $(a,s)$.
\end{enumerate}
}}\hspace*{0.3cm}($\ddag$)
\end{center}
\smallskip
We recall the idea behind ($\ddag$). Point ($\ddag.1$) describes the part of match $\pi'$ where it is still possible to encounter nodes which are labeled with $p$ in $\bbT'$. As $\tmapProj$ only takes the letter $p$ into account when defined on macro-states in $\shA$, we want $\pi'$ to visit only positions of the form $(Q,s) \in \shA \times T$ in that situation. Anytime we visit such a position $(Q,s)$ in $\pi'$, the role of the bundle is to provide one $g$-guided shadow match at position $(a,s)$ for each $a \in Q$.
Then $g'$ is defined in terms of what $g$ suggests from those positions.

 Point ($\ddag.2$) describes how we want the match $\pi'$ to be
 played on a $p$-free subtree: as any node that one might encounter has the same label in $\bbT$ and $\bbT'$,
it is safe to let ${\finexists p}.\bbA^{\fin}$ behave as $\bbA$ in such situation. Provided that the two matches visit the same basic positions, of the form $(a,s)\in A \times T$, we can let $g'$ just copy $g$.

The key observation is that, as $\bbT'$ is a \emph{finite} $p$-variant of $\bbT$, nodes labeled with $p$ are reachable only for finitely many rounds of $\pi'$. This means that, provided that ($\ddag$) hold at each round, ($\ddag.1$) will describe an initial segment of $\pi'$, whereas ($\ddag.2$) will describe the remaining part. Thus our proof that $g'$ is a winning strategy for $\exists$ in $\mc{G}_{\exists}$ is concluded by showing that ($\ddag$) holds for each stage of construction of $\pi'$ and $\mc{B}$.

For this purpose, we initialize $\pi'$ from position $(\shai,s) \in \shA\times T$ and the bundle $\mc{B}$ as $\mc{B}_0 = \{\pi_{a_I}\}$, with $\pi_{a_I}$ the partial $g$-guided match consisting only of the position $(a_I,s)\in A\times T$. The situation described by ($\ddag .1$) holds at the initial stage of the construction.
Inductively, suppose that at round $z_i$ we are given a position $(q,s) \in A^{\f} \times T$ in $\pi^{\f}$ and a bundle $\mc{B}_i$ as in ($\ddag$). To show that ($\ddag$) can be maintained at round $z_{i+1}$, we distinguish two cases, corresponding respectively to situation ($\ddag.1$) and ($\ddag.2$) holding at round $z_i$.
\begin{enumerate}[label = (\Alph*), ref = \Alph*]
  \item If $(q,s)$ is of the form $(Q,s) \in \shA \times T$, by inductive hypothesis we are given with $g$-guided shadow matches $\{\pi_a\}_{a \in Q}$ in $\mc{B}_i$. For each match $\pi_a$ in the bundle, we are provided with a valuation $V_{a,s}: A \rightarrow \p (\R{s})$ making $\tmap(a,\tscolors'(s))$ true. Then we further distinguish the following two cases.
\begin{enumerate}[label = (\roman*), ref = \roman*]
  \item \label{point:TsNotPFree} Suppose first that $\bbT'_s$ is not $p$-free. We let the suggestion $V' \: A^{\f} \to \p (\R{s})$ of $g'$ from position $(Q,s)$ be defined as follows:
       \begin{align*}
       V'(q') \isdef \begin{cases}
               \bigcap\limits_{\substack{(a,b) \in q',\\ a \in Q}}\{t\ \in \R{s} \mid t \in V_{a,s}(b)\}               & q' \in \shA \\[2em]
               \bigcup\limits_{a \in Q} \{t\ \in \R{s} \mid t \in V_{a,s}(q') \text{ and }\bbT'.t\text{ is $p$-free}\}              & q' \in A.
           \end{cases}
       \end{align*}
The definition of $V'$ on $q' \in \shA$ is standard 
(\emph{cf.}~\cite[Prop. 2.21]{Zanasi:Thesis:2012}) and guarantees a 
correspondence between the states assigned by the valuations 
$\{V_{a,s}\}_{a \in Q}$ and the macro-states assigned by $V'$. 
The definition of $V'$ on $q' \in A$ aims at fulfilling the conditions, 
expressed via $\qu$ and $\dqu$, on the number of nodes in $\R{s}$ witnessing 
(or not) some $A$-types. 
Those conditions are the ones that $\shDe(Q,\tscolors'(s))$ --and thus also
$\tmap^{\f}(Q,\tscolors'(s))$-- ``inherits'' by $\bigwedge_{a \in R} 
\tmap(a,\tscolors'(s))$, by definition of $\shDe$. 
Notice that we restrict $V'(q')$ to the nodes $t \in V_{a,s}(q')$ such that 
$\bbT'.t$ is $p$-free.
As $\bbT'$ is a \emph{finite} $p$-variant, only \emph{finitely many} nodes in
$V_{a,s}(q')$ will not have this property.
Therefore their exclusion, which is crucial for maintaining condition ($\ddag$) 
(\emph{cf.}~case \eqref{point:ddag2CardfromMacro} below), does not influence 
the fulfilling of the cardinality conditions expressed via $\qu$ and $\dqu$ in 
$\shDe(Q,\tscolors'(s))$.

       On the base of these observations, one can check that $V'$ makes $\shDe(Q,\tscolors'(s))$--and thus also $\tmap^{\f}(Q,\tscolors'(s))$--true in $\R{s}$. In fact, to be a legitimate move for $\exists$ in $\pi'$, $V'$ should make $\tmapProj(Q,\tscolors(s))$ true: this is the case, for $\tmap^{\f}(Q,\tscolors'(s))$ is either equal to $\tmap^{\f}(Q,\tscolors(s))$, if $p \not\in \tscolors'(s)$, or to $\tmap^{\f}(Q,\tscolors(s)\cup\{p\})$ otherwise. In order to check that we can maintain $(\ddag)$, let $(q',t) \in A^{\f} \times T$ be any next position picked by $\forall$ in $\pi'$ at round $z_{i+1}$. As before, we distinguish two cases:
       \begin{enumerate}[label = (\alph*), ref = \alph*]
         \item If $q'$ is in $A$, then, by definition of $V'$, $\forall$ can choose $(q',t)$ in some shadow match $\pi_a$ in the bundle $\mc{B}_i$. We dismiss the bundle --i.e. make it a singleton-- and bring only $\pi_a$ to the next round in the same position $(q',t)$. Observe that, by definition of $V'$, $\bbT'.t$ is $p$-free and thus ($\ddag.2$) holds at round $z_{i+1}$. \label{point:ddag2CardfromMacro}
         \item Otherwise, $q'$ is in $\shA$. The new bundle $\mc{B}_{i+1}$ is given in terms of the bundle $\mc{B}_i$: for each $\pi_a \in \mc{B}_i$ with $a\in Q$, we look if for some $b \in q'$ the position $(b,t)$ is a legitimate move for $\forall$ at round $z_{i+1}$; if so, then we bring $\pi_a$ to round $z_{i+1}$ at position $(b,t)$ and put the resulting (partial) shadow match $\pi_b$ in $\mc{B}_{i+1}$. Observe that, if $\forall$ is able to pick such position $(q',t)$ in $\pi'$, then by definition of $V'$ the new bundle $\mc{B}_{i+1}$ is non-empty and consists of an $g$-guided (partial) shadow match $\pi_b$ for each $b \in q'$. In this way we are able to keep condition ($\ddag.1$) at round $z_{i+1}$.
       \end{enumerate}
    \item Let us now consider the case in which $\bbT'_s$ is $p$-free. We let $g'$ suggest the valuation $V'$ that assigns to each node $t \in \R{s}$ all states in $\bigcup_{a \in Q}\{b \in A\ |\ t \in V_{a,s}(b)\}$. It can be checked that $V'$ makes $\bigwedge_{a \in Q} \tmap(a,\tscolors'(s))$ -- and then also $\tmap^{\f}(Q,\tscolors'(s))$ -- true in $\R{s}$. As $p \not\in \tscolors(s)=\tscolors'(s)$, it follows that $V'$ also makes $\tmapProj(Q,\tscolors(s))$ true, whence it is a legitimate choice for $\exists$ in $\pi'$. Any next basic position picked by $\forall$ in $\pi'$ is of the form $(b,t) \in A \times T$, and thus condition ($\ddag.2$) holds at round $z_{i+1}$ as shown in (i.a). 
  \end{enumerate}
  \item In the remaining case, $(q,s)$ is of the form $(a,s) \in A \times T$ and by inductive hypothesis we are given with a bundle $\mc{B}_i$ consisting of a single $f$-guided (partial) shadow match $\pi_a$ at the same position $(a,s)$. Let $V_{a,s}$ be the suggestion of $\exists$ from position $(a,s)$ in $\pi_a$. Since by assumption $s$ is $p$-free, we have that $\tscolors'(s) = \tscolors(s)$, meaning that $\tmapProj(a,\tscolors(s))$ is just $\tmap(a,\tscolors(s)) = \tmap(a,\tscolors'(s))$. Thus the restriction $V'$ of $V$ to $A$ makes $\tmap(a,\tscolors'(t))$ true and we let it be the choice for $\exists$ in $\tilde{\pi}$. It follows that any next move made by $\forall$ in $\tilde{\pi}$ can be mirrored by $\forall$ in the shadow match $\pi_a$.
\end{enumerate}
\end{proof}


\subsubsection{Closure under Boolean operations}

Here we show that the collection of $\Aut(\wmso)$-recognizable classes of tree
models is closed under the Boolean operations.
For union, we use the following result, leaving the straightforward proof as an 
exercise to the reader.

\begin{lemma}
\label{t:cl-dis}
Let $\bbA_{0}$ and $\bbA_{1}$ be $\wmso$-automata.
Then there is a $\wmso$-automaton $\bbA$ such that $\TMod(\bbA)$ is the
union of $\TMod(\bbA_{0})$ and $\TMod(\bbA_{1})$.
\end{lemma}

For closure under complementation we reuse the general results established in 
Section \ref{sec:parityaut} for parity automata.

\begin{lemma}
\label{t:cl-cmp}
Let $\bbA$ be an $\wmso$-automaton.
Then the automaton $\overline{\bbA}$ defined in Definition~\ref{d:caut} is a
$\wmso$-automaton recognizing the complement of $\TMod(\bbA)$.
\end{lemma}

\begin{proof} 
It suffices to check that Proposition \ref{prop:autcomplementation} restricts 
to the class $\AutWC(\ofoei)$ of $\wmso$-automata. 
First, the fact that $\ofoei$ is closed under Boolean duals 
(Definition~\ref{def:concreteduals}) implies that it holds for the class
$\Aut(\ofoei)$. 
It then remains to check that the dual automata construction 
$\overline{(\cdot)}$ preserves weakness and continuity. 
But this is straightforward, given the self-dual nature of these properties.
\end{proof}


We are now finally able to conclude the direction from formulas to automata of 
the characterisation theorem.

\begin{proof}[of Theorem \ref{t:wmsoauto}] The proof is by induction on $\phi$.
\begin{itemize}

\item 
For the base case, we consider the atomic formulas $\here{p}$, $p \inc q$ and 
$R(p,q)$.

\begin{itemize}
\item
The $\wmso$-automaton $\bbA_{\here{p}} = \tup{A,\tmap,\Omega,a_I}$ is given 
by putting
\begin{eqnarray*}
        A  \  \df \  \{a_0,a_1\}  \qquad \qquad  
	a_I  \   \df  \  a_0      \qquad \qquad   
	\Omega(a_0)  \  \df \  0  \qquad \qquad
        \Omega(a_1)  \  \df \  0 
\\ \tmap(a_0,c)  \  \df \  \left\{
	\begin{array}{ll}
           \forall x. a_1(x)  &  \mbox{if } p \in c
	\\ \bot               &  \mbox{otherwise.}
	\end{array}
\right. \qquad 
  \tmap(a_1,c)  \  \df \  \left\{
	\begin{array}{ll}
           \forall x. a_1(x)  &  \mbox{if } p \not\in c
	\\ \bot &  \mbox{otherwise.}
	\end{array}
\right.
\end{eqnarray*}
\item
The $\wmso$-automaton $\bbA_{p\inc q} = \tup{A,\tmap,\Omega,a_I}$ is given by
$A \isdef \{ a \}$, $a_{I} \isdef a$, $\pmap(a) \isdef 0$ and 
$\tmap(a,c) \isdef \forall x\, a(x)$ if $p \not\in c$ or $q\in c$, and
$\tmap(a,c) \isdef \bot$ otherwise.

\item
The $\wmso$-automaton $\bbA_{R(p,q)} = \tup{A,\tmap,\Omega,a_I}$ is given below:
\begin{eqnarray*}
        A  \  \df \  \{a_0,a_1\}  \qquad \qquad  
	a_I  \   \df  \  a_0      \qquad \qquad   
	\Omega(a_0)  \  \df \  0  \qquad \qquad
        \Omega(a_1)  \  \df \  1 
\\ \tmap(a_0,c)  \  \df \  \left\{
	\begin{array}{ll}
           \exists x. a_1(x) \wedge \forall y. a_0(y)  &  \mbox{if }p \in c
	\\ \forall x\ (a_0(x))  &  \mbox{otherwise.}
	\end{array}
\right. \qquad 
  \tmap(a_1,c)  \  \df \  \left\{
	\begin{array}{ll}
        \top  &  \mbox{if }q \in c \\
        \bot  &  \mbox{otherwise}
	\end{array}
\right.
\end{eqnarray*}
\end{itemize}

\item
For the Boolean cases, where $\phi = \psi_1 \vee \psi_2$ or $\phi = \neg\psi$
we refer to the Boolean closure properties that we just established in the
Lemmas~\ref{t:cl-dis} and~\ref{t:cl-cmp},
respectively.

\item
The case $\phi = \exists p. \psi$ follows by the following chain of
equivalences, where $\bbA_{\psi}$ is given by the inductive hypothesis and
${\finexists p}.\bbA_{\psi}$ is constructed according to
Definition~\ref{DEF_fin_projection}:
\begin{alignat*}{2}
{\finexists p}.\bbA_{\psi} \text{ accepts }\mb{T}
   & \text{ iff }
     \bbA_{\psi} \text{ accepts } \mb{T}[p \mapsto X],
     \text{ for some } X \sse_{\om} T
   & \quad\text{(Lemma~\ref{PROP_fin_projection})}
\\ & \text{ iff }
     \mb{T}[p \mapsto X] \models \psi,
     \text{ for some } X \sse_{\om} T
   & \quad\text{(induction hyp.)}
\\ & \text{ iff }
    \mb{T} \models \exists p. \psi
   & \quad\text{(semantics $\wmso$)}
\end{alignat*}
\end{itemize}
\end{proof}


\section{Automata for $\nmso$}
   \label{sec:autnmso}

In this section we introduce the automata that capture $\nmso$.

\begin{definition}
A \emph{$\nmso$-automaton} is a weak automaton for the one-step language $\ofoe$.
\end{definition}

Aanalogous to the previous section, our main goal here is to construct an
equivalent $\nmso$-automaton for every $\nmso$-formula.

\begin{theorem}
\label{t:nmsoauto}
There is an effective construction transforming a $\nmso$-formula $\phi$
into a $\nmso$-automaton $\bbA_{\phi}$ that is equivalent
to $\phi$ on the class of trees.
\end{theorem}

The proof for Theorem \ref{t:nmsoauto} will closely follow the steps for proving
the analogous result for $\wmso$ (Theorem \ref{t:wmsoauto}).
Again, the crux of the matter is to show that the collection of classes of tree
models that are recognisable by some $\nmso$-automaton, is closed under the 
relevant notion of projection.
Where this was finitary projection for $\wmso$ (Def. \ref{def:tree_finproj-w}),
the notion mimicking $\nmso$-quantification is \emph{noetherian} projection.

\begin{definition}\label{def:tree_finproj-n}
Given a set $\pprop$ of proposition letters, $p \not\in P$ and a class 
$\mathsf{C}$ of $\pprop\cup\{p\}$-labeled trees, we define the \emph{noetherian 
projection} of $\mathsf{C}$ over $p$ as the language of $\pprop)$-labeled trees 
given as 
$$
\noetexists p.\mathsf{C} \df \{\bbT \mid 
   \text{ there is a noetherian $p$-variant } \bbT' \text{ of } \bbT 
   \text{ with } \bbT' \in \mathsf{C}\}.
$$
A collection of classes of tree modelss is \emph{closed under noetherian 
projection over $p$} if it contains the class ${{\noetexists} p}.\mathsf{C}$
whenever it contains the class $\mathsf{C}$ itself.
\end{definition}

\subsection{Simulation theorem for $\nmso$-automata}\label{sec:simulation_nmso}

Just as for $\wmso$-automata, also for $\nmso$-automata the projection construction will rely on a simulation theorem, constructing a two-sorted automaton $\bbA^{\noet}$ consisting of a copy of the original automaton, based on states $A$, and a variation of its powerset construction, based on macro-states $\shA$. For any accepted $\bbT$, we want any match $\pi$ of $\mc{A}(\bbA^{\noet},\bbT)$ to split in two parts:
\begin{description}
\item[(\textit{Non-deterministic mode})] for finitely many rounds $\pi$ is
played on macro-states, i.e. positions are of the form $\shA \times T$. 
The strategy of player $\exists$ is functional in $\shA$, i.e. it assigns 
\emph{at most one macro-state} to each node.
\item[(\textit{Alternating mode})] 
At a certain round, $\pi$ abandons macro-states and turns into a match of the
game $\mc{A}(\bbA,\bbT)$, i.e. all next positions are from $A \times T$ (and 
are played according to a non-necessarily functional strategy). 
\end{description}
The only difference with the two-sorted construction for $\wmso$-automata is 
that, in the non-deterministic mode, the cardinality of nodes to which 
$\exists$'s strategy assigns macro-states is irrelevant. 
Indeed, $\nmso$'s finiteness is only on the vertical dimension: assigning an 
odd priority to macro-states will suffice to guarantee that the 
non-deterministic mode processes just a well-founded portion of any accepted
tree.

We now proceed in steps towards the construction of $\bbA^{\noet}$. First, the following lifting from states to macro-states parallels Definition \ref{DEF_finitary_lifting}, but for the one-step language $\ofoe$ proper of $\nmso$-automata. It is based on the basic form for $\ofoe$-formulas, see Definition \ref{def:basicform-ofoe}.

\begin{definition}\label{DEF_noetherian_lifting}
Let $\varphi \in {\ofoe}^+(A)$ be of shape $\posdbnfofoe{\vlist{T}}{\Pi}$ for some $\Pi \subseteq \shA$ and $\vlist{T} = \{T_1,\dots,T_k\} \subseteq \shA$. We define $\varphi^{\noet}$ as $\posdbnfofoe{\lift{\vlist{T}}}{\lift{\Pi}} \in {\ofoe}^+(\shA )$, that means,
\begin{equation}\label{eq:unfoldingNablaofoe}
\varphi^{\noet} \ \df \
    \exists \vlist{x}.\big(\arediff{\vlist{x}} \land \bigwedge_{0 \leq i \leq n} \tau^+_{\lift{T}_i}(x_i)
\land
    \forall z.(\arediff{\vlist{x},z} \to \bigvee_{S\in \lift{\Pi} } \tau^+_S(z))\big)
\end{equation}
\end{definition}

It is instructive to compare \eqref{eq:unfoldingNablaofoe} with its
$\wmso$-counterpart \eqref{eq:unfoldingNablaolque}: the difference is that,
because the quantifiers $\qu$ and $\dqu$ are missing, the sentence does not
impose any cardinality requirement, but only enforces $\shA$-separability --- 
\emph{cf.} Section \ref{sec:onestep-short}.

\begin{lemma}\label{lemma:automatafunctionalsentence}
Let $\varphi \in {\ofoe}^+(A)$ and $\varphi^{\noet}\in {\ofoe}^+(\shA )$ be as
in Definition~\ref{DEF_noetherian_lifting}. 
Then $\varphi^{\noet}$ is separating in $\shA$.
\end{lemma}

\begin{proof}
Each element of $\lift{\vlist{T}}$ and $\lift{\Pi}$ is by definition either the
empty set or a singleton $\{Q\}$ for some $Q \in \shA$. 
Then the statement follows from Proposition~\ref{p:sep} .
\end{proof}

We are now ready to define the transition function for macro-states. 
The following adapts Definition \ref{PROP_DeltaPowerset} to the one-step
language $\ofoe$ of $\nmso$-automata, and its normal form result, 
Theorem~\ref{t:osnf}. 
\begin{definition}\label{PROP_DeltaPowerset_noet}
Let $\bbA = \tup{A,\tmap,\pmap,a_I}$ be a $\nmso$-automaton. 
Fix any $c \in C$ and $Q \in \shA$. 
By Theorem~\ref{t:osnf} there is a sentence $\Psi_{Q,c} \in
{\ofoe}^+(A)$ in the basic form $\bigvee \dbnfofoe{\vlist{T}}{\Pi}$, for some
$\Pi \subseteq \shA$ and $T_i \subseteq A$, such that $\bigwedge_{a \in Q} 
\tmap(a,c) \equiv \Psi_{Q,c}$.
By definition, $\Psi_{Q,c} = \bigvee_{n}\varphi_n$, with each $\phi_{n}$ of 
shape $\dbnfofoe{\vlist{T}}{\Pi}$.
We put $\bmDe(Q,c) \isdef  \bigvee_{n}\varphi_n^{\noet}  \in {\ofoe}^+(\shA)$, where the translation $(\cdot)^{\noet}$ is as in Definition \ref{DEF_noetherian_lifting}.
\end{definition}

\noindent
We now have all the ingredients for the two-sorted construction over
$\nmso$-automata.

\begin{definition}\label{def:noetherianconstruct}
Let $\bbA = \tup{A,\tmap,\pmap,a_I}$ be a {\nmso-automaton}.
We define the \emph{noetherian construct over $\bbA$} as the automaton
$\bbA^{\noet} = \tup{A^{\noet},\tmap^{\noet},\pmap^{\noet},a_I^{\noet}}$ given
by
\[
\begin{array}{lll}
   A^{\noet}   &\df&  A \cup \shA
\\ a_I^{\noet} &\df&  \{a_I\}
\end{array}
\hspace*{5mm}
\begin{array}{lll}
   \pmap^{\noet}(a) &\df& \pmap(a)
\\ \pmap^{\noet}(R) &\df& 1
\end{array}
\hspace*{5mm}
\begin{array}{lll}
   \tmap^{\noet}(a,c) &\df& \tmap(a,c)
\\ \tmap^{\noet}(Q,c) &\df&
  \bmDe(Q,c) \vee \bigwedge_{a \in Q} \! \! \tmap(a,c).
\end{array}
\]
\end{definition}

The construction is the same as the one for $\wmso$-automata (Definition \ref{def:finitaryconstruct}) but for the definition of the transition function for macro-states, which is now free of any cardinality requirement. 

\begin{definition}\label{def:noetherianstrategy}
We say that a strategy $f$ in an acceptance game $\agame(\bbA,\bbT)$ is \emph{noetherian} in $B \subseteq A$ when in any $f$-guided match there can be only finitely many rounds played at a position of shape $(q,s)$ with $q \in B$.
\end{definition}


\begin{theorem}[Simulation Theorem for $\nmso$-automata]
\label{PROP_facts_noetConstr}
Let $\bbA$ be an $\nmso$-automaton and $\bbA^{\noet}$ its noetherian construct.
\begin{enumerate}[(1)]
  \itemsep 0 pt
  \item \label{point:finConstrAut-n}
$\bbA^{\noet}$ is an $\nmso$-automaton.
\item 
\label{point:finConstrStrategy-n}
For any $\bbT$, if $\eloise$ has a winning strategy in $\mathcal{A}(\bbA^{\noet},
\bbT)$ from position $(a_I^{\noet},s_I)$ then she has one that is functional in
$\shA$ and noetherian in $\shA$.
\item $\bbA \equiv \bbA^{\noet}$. \label{point:finConstrEquiv-n}
\end{enumerate}
\end{theorem}
\begin{proof}
The proof follows the same steps as the one of 
Proposition \ref{PROP_facts_finConstrwmso}, minus all the concerns about
continuity of the constructed automaton and any associated winning strategy
$f$ being finitary. 
One still has to show that $f$ is noetherian in $\shA$ (``vertically finitary''),
but this is enforced by macro-states having an odd parity: visiting one of them 
infinitely often would mean $\exists$'s loss.
\end{proof}

\begin{remark}
As mentioned, the class $\Aut(\ofoe)$ of automata characterising $\smso$
\cite{Jan96} also enjoys a simulation theorem \cite{Walukiewicz96}, turning any
automaton into an equivalent non-deterministic one.
Given that the class $\AutW(\ofoe)$ only differs for the weakness constraint,
one may wonder if the simulation result for $\Aut(\ofoe)$ could not actually be
restricted to $\AutW(\ofoe)$, making our two-sorted construction redundant.
This is actually not the case: not only does Walukiewicz's simulation theorem
\cite{Walukiewicz96} fail to preserve the weakness constraint, but even without
this failure our purposes would not be served:
A fully non-deterministic automaton is instrumental in guessing a $p$-variant
of any accepted tree, but it does not guarantee that the $p$-variant is also
noetherian, as the two-sorted construct does.
\end{remark}

\subsection{From formulas to automata}

We can now conclude one direction of the automata characterisation of $\nmso$.

\begin{lemma}\label{PROP_noet_projection}
For each $\nmso$-automaton $\bbA$ on alphabet $\p (\pprop \cup \{p\})$, let
$\bbA^{\noet}$ be its noetherian construct.
We have that
$$\TMod({{\exists} p}.\bbA^{\noet}) \ \equiv\
{{\noetexists} p}.\TMod(\bbA).
$$
\end{lemma}

\begin{proof} 
The argument is the same as for $\wmso$-automata (Lemma \ref{PROP_fin_projection}). As in that proof, the inclusion from left to right relies on the simulation result (Theorem \ref{PROP_facts_noetConstr}): ${{\exists} p}.\bbA^{\noet}$ is two-sorted and its non-deterministic mode can be used to guess a noetherian $p$-variant of any accepted tree. \end{proof}

\begin{proof}[of Theorem \ref{t:nmsoauto}] 
As for its $\wmso$-counterpart Theorem \ref{t:wmsoauto}, the proof is by
induction on $\varphi \in \nmso$. 
The boolean inductive cases are handled by the $\nmso$-versions of 
Lemma \ref{t:cl-dis} and \ref{t:cl-cmp}. 
The projection case follows from Lemma~\ref{PROP_noet_projection}.
\end{proof}

\section{Fixpoint operators and second-order quantifiers}
\label{sec:fixpointToSO}

In this section we will show how to translate some of the mu-calculi that we
encountered until now into the appropriate second-order logics.
Given the equivalence between automata and fixpoint logics that we established 
in Section~\ref{sec:parityaut}, and the embeddings of $\wmso$ and $\nmso$ into,
respectively, the automata classes $\AutWC(\ofoei)$ and $\AutW(\ofoe)$ that
we provided in the Sections~\ref{sec:autwmso} and~\ref{sec:autnmso} for the 
class of tree models, the results here provide the missing link in the 
automata-theoretic characterizations of the monadic second order logics
$\wmso$ and $\nmso$:
\begin{eqnarray*}
   \mu_{C}(\ofoei) \equiv \wmso 
   && \qquad  \text{ (over the class of all tree models)} 
\\ \mu_{D}(\ofoe)  \equiv \nmso 
   && \qquad  \text{ (over the class of all tree models)}. 
\end{eqnarray*}

\subsection{Translating $\mu$-calculi into second-order logics}

More specifically, our aim in this Section is to prove the following result.

\begin{theorem}
\label{t:mfl2mso}
\begin{enumerate}[(1)]
\item
There is an effective translation $(\cdot)^{*}: \mu_{D}\ofoe \to \nmso$
such that $\phi \equiv \phi^{*}$ for every $\phi \in \mu_{D}\ofoe$; that is:
\[
\mu_{D}\ofoe \leq \nmso.
\]

\item
There is an effective translation $(\cdot)^{*}: \mu_{C}\ofoei \to \wmso$
such that $\phi \equiv \phi^{*}$ for every $\phi \in \mu_{C}\ofoei$; that is:
\[
\mu_{C}\ofoei \leq \wmso.
\]
\end{enumerate}
\end{theorem}

Two immediate observations on this Theorem are in order.
First, note that we use the same notation $(\cdot)^{*}$ for both translations; 
this should not cause any confusion since the maps agree on formulas belonging 
to their common domain.
Consequently, in the remainder we will speak of a single translation 
$(\cdot)^{*}$.
Second, as the target language of the translation $(\cdot)^{*}$ we will take 
the \emph{two-sorted} version of second-order logic, as discussed in 
section~\ref{sec:prel-so}, and thus we will need Fact~\ref{fact:msovs2mso}
to obtain the result as formulated in Theorem~\ref{t:mfl2mso}, that is,
for the one-sorted versions of \mso.
We reserve a fixed individual variable $v$ for this target language, i.e., 
every formula of the form $\phi^{*}$ will have this $v$ as its unique free 
variable; the equivalence $\phi \equiv \phi^{*}$ is to be understood accordingly.

The translation $(\cdot)^{*}$ will be defined by a straightforward induction on
the complexity of fixpoint formulas.
The two clauses of this definition that deserve some special attention are the
ones related to the fixpoint operators and the modalities.

\paragraph{Fixpoint operators} 
It is important to realise that our clause for the fixpoint operators
\emph{differs} from the one used in the standard inductive translation 
$(\cdot)^{s}$ of $\muML$ into standard $\mso$, where we would inductively
translate $(\mu p. \phi)^{*}$ as
\begin{equation}
\label{eq:st}
\forall p\, \big( \forall w\, (\phi^{*}[w/v] \to p(w)) \to p(v) \big),
\end{equation}
which states that $v$ belongs to any prefixpoint of $\phi$ with respect to $p$.
To understand the problem with this translation in the current context, suppose,
for instance, that we want to translate some continuous $\mu$-calculus into
$\wmso$.
Then the formula in \eqref{eq:st} expresses that $v$ 
belongs to the intersection of all \emph{finite} prefixpoints of $\phi$, whereas
the least fixpoint is identical to the intersection of \emph{all} prefixpoints.
As a result, \eqref{eq:st} does not give the right translation for the formula 
$\mu p.\phi$ into \wmso.

To overcome this problem, we will prove that least fixpoints in restricted 
calculi like $\mu_{D}\ofoe$, $\mu_{C}\ofoei$ and many others, in fact satisfy a
rather special property, which enables an alternative translation.
We need the following definition to formulate this property.

\begin{definition}
\label{d:rst}
Let $F: \pow(S)\to \pow(S)$ be a functional; for a given $X \subseteq S$ we define
the \emph{restricted} map $F_{\rst{X}}: \pow(S)\to \pow(S)$ by putting 
$F_{\rst{X}}(Y) \isdef FY \cap X$.
\end{definition}

The observations formulated in the proposition below provide the crucial insight
underlying our embedding of various alternation-free and continuous 
$\mu$-calculi into, respectively, $\nmso$ and $\wmso$.

\begin{proposition}
\label{p:afmc-rstGen}
\label{p:keyfix}
Let $\bbS$ be an LTS, and let $r$ be a point in $\bbS$.
\begin{enumerate}[(1)]
\item
 For any formula $\phi$ with $\mu p. \phi \in \mu_{D}\ofoe$ we have
\begin{equation}
\label{eq:foe-d}
r \in \ext{\mu p.\phi}^{\bbS} \text{ iff there is a noetherian set $X$ such 
that } r \in \LFP. (\phi^{\bbS}_{p})_{\rst{X}}.
\end{equation}

\item
For any formula $\phi$ with $\mu p. \phi \in \mu_{C}\ofoei$ we have
\begin{equation}
\label{eq:foei-c}
r \in \ext{\mu p.\phi}^{\bbS} \text{ iff there is a finite set $X$ such 
that } r \in \LFP. (\phi^{\bbS}_{p})_{\rst{X}}.
\end{equation}
\end{enumerate}
\end{proposition}

\begin{remark}
In fact, the statements in Proposition~\ref{p:keyfix} can be generalised to the
setting of a fixpoint logic $\mu\oslang$ associated with an arbitrary one-step 
language $\oslang$.
\end{remark}

The right-to-left direction of both \eqref{eq:foe-d} and \eqref{eq:foei-c} follow
from the following, more general, statement, which can be proved by a routine 
transfinite induction argument.

\begin{proposition}
\label{p:rstfix}
Let $F:  \pow(S)\to \pow(S)$ be monotone.
Then for every subset $X \subseteq S$ it holds that $\LFP. F\rst{X}\subseteq 
\LFP.F$.
\end{proposition}

The left-to-right direction of \eqref{eq:foe-d} and \eqref{eq:foei-c} 
will be proved in the next two sections.
Note that in the continuous case we will in fact prove a slightly stronger
result, which applies to \emph{arbitrary} continuous functionals.

\newcommand{\PRE}{\mathit{PRE}}
The point of Proposition~\ref{p:keyfix} is that it naturally suggests the
following translation for the least fixpoint operator, as a subtle but important 
variation of \eqref{eq:st}:

\begin{equation}
\label{eq:trlmu}
(\mu p. \varphi)^{*} \df 
   \exists q\,\Big(\forall  p \sse q.\,
      \big(p \in \PRE((\varphi^{\bbS}_{p})_{\rst{q}}) \to p(v)\big)\Big),
\end{equation}
where $p \in \PRE((\varphi^{\bbS}_{p})_{\rst{q}})$ expresses that $p \sse q$ 
is a prefixpoint of the map $(\phi^{\bbS}_{p})_{\rst{q}}$, that is:
\[
p  \in \PRE((\varphi^{\bbS}_{p})_{\rst{q}}) \df
\forall w\, \Big(
( q(w) \land \varphi^{*}[w/v]) \to p(w)
\Big).
\]

\paragraph{Modalities}
Finally, before we can give the definition of the translation $(\cdot)^{*}$, we
briefly discuss the clause involving the modalities.
Here we need to understand the role of the \emph{one-step formulas} in the 
translation.
First an auxiliary definition.

\begin{definition}\label{def:exp}
Let $\bbS = \tup{T,R,\tscolors, s_I}$ be a $\pprop$-LTS, $A$ be a set of
new variables, and $V: A \to \pow(X)$ be a valuation on a subset $X\subseteq T$. 
The $\pprop\cup A$-LTS $\bbS^{V}\isdef \tup{T,R,\tscolors^V, s_I}$ given by 
defining the marking $\tscolors^V: T \to \pow{(\pprop \cup A)}$ where
\[\tscolors^V(s)\isdef  
\begin{cases} \tscolors(s) & \text{ if } s \notin X \\
\tscolors(s) \cup \{ a \in A \mid s \in V(a)\}& \text{ else,}
\end{cases}\]
is called the $V$-expansion of $\bbS$.
\end{definition}

The following proposition states that at the one-step level, the formulas that 
provide the semantics of the modalities of $\mu\ofoe$ and $\mu\ofoei$ can indeed 
be translated into, respectively $\nmso$ and $\wmso$.

\begin{proposition}
\label{p:1trl}
There is a translation $(\cdot)^{\dag}: \ofoei(A) \to \wmso$ such that for every 
model $\bbS$ and every valuation $V: A \to \pow(R[s_{I}])$:
\[
(R[s_{I}],V) \models \al \text{ iff } \bbS^{V} \models \al^{\dag}[s_{I}].
\]
Moreover, $(\cdot)^{\dag}$ restricts to first-order logic, i.e., $\al^{\dag}$ is
a first-order formula if $\al \in \ofoe$.
\end{proposition}


\begin{proof}
Basically, the translation $(\cdot)^{\dag}$ \emph{restricts} all quantifiers
to the collection of successors of $v$.
In other words, $(\cdot)^{\dag}$ is the identity on basic formulas, it commutes
with the propositional connectives, and for the quantifiers $\exists$ and $\qu$
we define:
\[\begin{array}{lll}
(\exists x\, \al)^{\dag} &\df& \exists x\, (Rvx \land \al^{\dag})
\\ (\qu x\, \al)^{\dag}  &\df& \forall p \exists x\, (Rvx \land \neg p(x) 
    \land \al^{\dag})
\end{array}\]
We leave it for the reader to verify the correctness of this definition ---
observe that the clause for the infinity quantifier $\qu$ is based on the 
equivalence between $\wmso$ and $\foei$, established by 
V\"a\"an\"anen~\cite{vaananen77}.
\end{proof}

\noindent
We are now ready to define the translation used in the main result of this 
section.

\begin{definition}
By an induction on the complexity of formulas we define the following 
translation $(\cdot)^{*}$ from $\mu\foei$-formulas to formulas of monadic 
second-order logic:
\[\begin{array}{lll}
   p^{*} &\df& p(v)
\\ (\neg\phi)^{*}        &\df& \neg \phi^{*}
\\ (\phi\lor\psi)^{*}    &\df& \phi^{*} \lor \psi^{*}
\\ (\nxt{\al}(\ol{\phi}))^{*} &\df& \al^{\dag}[\phi_{i}^{*}/a_{i} \mid i \in I],
\end{array}\]
where $\al^{\dag}$ is as in Proposition~\ref{p:1trl}, and $[\phi_{i}^{*}/a_{i}
\mid i \in I]$ is the substitution that replaces every occurrence of an atomic
formula of the form $a_{i}(x)$ with the formula $\phi_{i}^{*}(x)$ (i.e. the 
formula $\phi_{i}^{*}$, but with the free variable $v$ substituted by $x$).

Finally, the inductive clause for a formula of the form $\mu p.\phi$ is given
as in \eqref{eq:trlmu}.
\end{definition}

\begin{proofof}{Theorem~\ref{t:mfl2mso}}
First of all, it is clear that in both cases the translation $(\cdot)^{*}$ lands 
in the correct language.
For both parts of the theorem, we thence prove that $(\cdot)^{*}$ is truth
preserving by a straightforward formula induction.
E.g., for part (2) we need to show that, for an arbitrary formula $\phi\in
\mu_{C}\ofoei$ and an arbitrary model $\bbS$:
\begin{equation}
\label{eq:xxxx1}
\bbS \mmodels \phi \text{ iff } \bbS \models \phi^{*}[s_{I}].
\end{equation}

As discussed in the main text, the two critical cases concern the inductive 
steps for the modalities and the least fixpoint operators. 
Let $ \oslang^{+} \in \{\ofoe,\ofoei\}$.
We start verifying the case of modalities. 
Hence, consider the formula $\nxt{\al}(\varphi_{1},\ldots,\varphi_{n})$ with 
$\al(a_{1},\ldots,a_{n}) \in \oslang^{+}$. 
By induction hypothesis, $\phi_\ell \equiv \phi^{*}_\ell$, for $\ell=1,\dots, n$.
Now, let $\bbS$ be a transition system. We have that
\begin{align*}
\bbS \mmodels \nxt{\al}(\varphi_{1},\ldots,\varphi_{n}) \text{ iff } 
  & (R[s_{I}],V_{\overline{\varphi}}) \models \al(a_{1},\ldots,a_{n})  
  & \text{( by \eqref{eq:mumod})}
\\ \text{ iff } 
  & \bbS^{V_{\overline{\varphi}}} \models \al^{\dag}[s_{I}] 
  & \text{( by Prop. \ref{p:1trl})}
\\ \text{ iff } 
  & \bbS \models \al^{\dag}[\phi_{i}^{*}/a_{i} \mid i \in I][s_{I}] 
  & \text{( by \eqref{eq:valmod}, Def. \ref{def:exp} and IH)}
\end{align*}

The inductive step for the least fixpoint operator will be justified by 
Proposition~\ref{p:keyfix}.
In more detail, given a formula of the form $\mu x. \psi \in\mu_{Y} \oslang^{+}$,
with $Y=D$ for $\oslang^{+} =\ofoe$, and $Y=C$ for $\oslang^{+} =\ofoei$, 
consider the following chain of equivalences:
\begin{align*}
  & s_{I} \in \ext{\mu p.\psi}^{\bbS} 
\\ \text{ iff } 
  & s_{I} \in \LFP. (\psi^{\bbS}_{p})_{\rst{Q}} \text{ for some }
        \begin{cases} \text{ finite} & \\ \text{ noetherian} & 
	\end{cases} 
    \text{set } Q 
  & \text{( by \eqref{eq:foe-d}/\eqref{eq:foei-c})}
\\ \text{ iff } 
  &  s_{I} \in \bigcap \Big\{ P \subseteq Q \mid P \in 
         \PRE((\psi^{\bbS}_{p})_{\rst{Q}})\Big\} 
     \text{ for some }
        \begin{cases} \text{ finite} & \\ \text{ noetherian} & 
	\end{cases} 
     \text{set } Q 
\\ \text{ iff } & 
    \bbS \models \exists q.\, \Big(\forall p \subseteq q.\,
       \big(p \in \PRE((\psi^{\bbS}_{p})_{\rst{q}}) \to p(s_{I})\big)
       \Big)
\\ \text{ iff } & 
    \bbS \models (\mu p. \psi)^{*}[s_{I}].
   & (\text{IH})
\end{align*}
This concludes the proof of \eqref{eq:xxxx1}.
\end{proofof}
\subsection{Fixpoints of continuous maps}

It is well-known that continuous functionals are \emph{constructive}.
That is, if we construct the least fixpoint of a continuous functional $F: 
\pow(S) \to \pow(S)$ using the ordinal approximation $\nada, F\nada, F^2\nada, 
\ldots, F^{\al}\nada, \ldots$, then we reach convergence after at most $\omega$
many steps, implying that $\LFP. F = F^{\omega}\nada$.
We will see now that this fact can be strengthened to the following observation,
which is the crucial result needed in the proof of Proposition~\ref{p:keyfix}.

\begin{theorem}
\label{t:fixcont}
Let $F: \pow(S)\to \pow(S)$ be a continuous functional.
Then for any $s \in S$:
\begin{equation}
\label{eq:fixcont}
s \in \LFP. F \text{ iff }
s \in \LFP.F\rst{X}, \text{ for some finite } X \subseteq S.
\end{equation}
\end{theorem}

\begin{proof}
The direction from right to left of \eqref{eq:fixcont} is a special case of 
Proposition~\ref{p:rstfix}.
For the opposite direction of \eqref{eq:fixcont} a bit more work is needed.
Assume that $s \in \LFP. F$; we claim that there are sets $U_{1},\ldots,U_{n}$,
for some $n \in \omega$, such that $s \in U_{n}$, $U_{1} \sse_{\omega} F(\nada)$,
and $U_{i+1} \sse_{\omega} F(U_{i})$, for all $i$ with $1 \leq i < n$.

To see this, first observe that since $F$ is continuous, we have $\LFP. F = 
F^{\omega}(\nada) = \bigcup_{n\in\omega}F^{n}(\nada)$, and so we may take $n$ to
be the least natural number such that $s \in F^{n}(\nada)$.
By a downward induction we now define sets $U_{n},\ldots,U_{1}$, with $U_{i} 
\sse F^{i}(\nada)$ for each $i$. 
We set up the induction by putting $U_{n} \isdef \{ s \}$, then $U_{n}
\sse F^{n}(\nada)$ by our assumption on $n$.
For $i<n$, we define $U_{i}$ as follows.
Using the inductive fact that $U_{i+1} \sse_{\omega} F^{i+1}(\nada) = 
F(F^{i}(\nada))$, it follows by continuity of $F$ that for each $u \in U_{i+1}$
there is a set $V_{u} \sse_{\omega} F^{i}(\nada)$ such that $u \in F(V_{u})$.
We then define $U_{i} \isdef \bigcup \{ V_{u} \mid u \in U_{i+1} \}$,
so that clearly $U_{i+1} \sse_{\omega} F(U_{i})$ and $U_{i} \sse_{\omega}
F^{i}(\nada)$.
Continuing like this, ultimately we arrive at stage $i=1$ where we find
$U_{1} \sse F(\nada)$ as required.

Finally, given the sequence $U_{n},\ldots,U_{1}$, we define 
\[
X \isdef  \bigcup_{0<i\leq n} U_{i}.
\]
It is then straightforward to prove that $U_{i} \sse \LFP. F\rst{X}$, for each 
$i$ with $0<i\leq n$, and so in particular we find that $s \in U_{n} \sse \LFP.
F\rst{X}$.
This finishes the proof of the implication from left to right in 
\eqref{eq:fixcont}.
\end{proof}

\noindent
As an almost immediate corollary of this result we obtain the second part of 
Proposition~\ref{p:keyfix}.

\begin{proofof}{Proposition~\ref{p:keyfix}(2)}
Take an arbitrary formula $\mu p. \phi \in \mu_{C}\ofoei$, then by definition 
we have $\phi \in \mu_{C}\ofoei \cap \cont{\mu\ofoei}{p}$.
But it follows from a routine inductive proof that every formula $\psi \in 
\mu_{C}\ofoei \cap \cont{\mu\ofoei}{\qprop}$ is continuous in each variable 
in $\qprop$.
Thus $\phi$ is continuous in $p$, and so the result is immediate by 
Theorem~\ref{t:fixcont}.
\end{proofof}


\subsection{Fixpoints of noetherian maps}

We will now see how to prove Proposition~\ref{p:keyfix}(1), which is the key 
result that we need to embed alternation-free $\mu$-calculi such as 
$\mu_{D}\ofoe$ and $\mu_{D}\ML$ into noetherian second-order logic.
Perhaps suprisingly, this case is slightly more subtle than the characterisation of
fixpoints of continuous maps.

We start with stating some auxiliary definitions and results on monotone 
functionals, starting with a game-theoretic characterisation of their least
fixpoints~\cite{Ven08}.

\begin{definition}
\label{d:unfgame}
Given a monotone functional $F: \pow(S)\to \pow(S)$ we define the 
\emph{unfolding game} $\UG_{F}$ as follows:
\begin{itemize}
\item at any position $s \in S$, $\eloise$ needs to pick a set $X$ such that 
$s \in FX$;
\item at any position $X \in \pow(S)$, $\abelard$ needs to pick an element of 
$X$
\item all infinite matches are won by $\abelard$.
\end{itemize}
A positional strategy $\ystrat: S \to \pow(S)$ for $\eloise$ in $\UG_{F}$ is 
\emph{descending} if, for all ordinals $\alpha$,
\begin{equation}
\label{eq:dec}
s \in F^{\alpha+1}(\nada) \text{ implies } \ystrat(s) \sse F^{\alpha}(\nada).
\end{equation}
\end{definition}

It is not the case that \emph{all} positional winning strategies for $\eloise$ 
in $\UG_{F}$ are descending, but the next result shows that there always is one.

\begin{proposition}
\label{p:unfg}
Let $F: \pow(S)\to \pow(S)$ be a monotone functional. 
\begin{enumerate}[(1)]
\item
For all $s \in S$, $s \in \win_{\eloise}(\UG_{F})$ iff $s \in \LFP. F$;
\item 
If $s \in \LFP. F$, then \eloise has a descending winning strategy in 
$\UG_{F}@s$.
\end{enumerate}
\end{proposition}

\begin{proof}
Point (1) corresponds to \cite[Theorem 3.14(2)]{Ven08}.
For part (2) one can simply take the following strategy.
Given $s \in \LFP.F$, let $\alpha$ be the least ordinal such that $s \in 
F^{\alpha}(\nada)$; it is easy to see that $\alpha$ must be a successor ordinal,
say $\alpha = \beta + 1$. 
Now simply put $\ystrat(s) \isdef  F^{\beta}(\nada)$.
\end{proof}

\begin{definition}
\label{d:str-tree}
Let $F: \pow(S)\to \pow(S)$ be a monotone functional, let $\ystrat$ be a 
positional winning strategy for $\eloise$ in $\UG_{F}$, and let $r \in S$. 
Define $T_{\ystrat,r} \sse S$ to be the set of states in $S$ that are 
$\ystrat$-reachable in $\UG_{F}@r$.
This set has a tree structure induced by the map $\ystrat$ itself, where the 
children of $s \in T_{\ystrat,r}$ are given by the set $\ystrat(s)$; we will
refer to $T_{\ystrat,r}$ as the \emph{strategy tree} of
$\ystrat$.
\end{definition}
Note that a strategy tree $T_{\ystrat,r}$ will have no infinite paths, since we
define the notion only for a \emph{winning} strategy $\ystrat$.

\begin{proposition}
\label{p:afmc-2}
Let $F: \pow(S)\to \pow(S)$ be a monotone functional, let $r \in S$, and let 
$\ystrat$ be a descending winning strategy for $\eloise$ in $\UG_{F}$.
Then
\begin{equation}
\label{eq:afmc3}
r \in \LFP. F \text{ implies } r \in \LFP. F\rst{T_{\ystrat,r}}.
\end{equation}
\end{proposition}

\begin{proof}
Let $F,r$ and $\ystrat$ be as in the formulation of the proposition.
Assume that $r \in \LFP. F$, then clearly $r \in F^{\al}(\nada)$ for some 
ordinal $\al$; furthermore, $T_{\ystrat,r}$ is defined and clearly we have 
$r \in T_{\ystrat,r}$.
Abbreviate $T \isdef T_{\ystrat,r}$.
It then suffices to show that for all ordinals $\alpha$ we have
\begin{equation}
\label{eq:unf1}
F^{\al}(\nada) \cap T \sse (F\rst{T})^{\alpha}(\nada).
\end{equation}
We will prove \eqref{eq:unf1} by transfinite induction.
The base case, where $\alpha = 0$, and the inductive case where $\alpha$ is a 
limit ordinal are straightforward, so we focus on the case where $\alpha$ is a 
successor ordinal, say $\alpha = \beta +1$.
Take an arbitrary state $u \in F^{\beta+1}(\nada) \cap T$, then we 
find $\ystrat(u) \sse F^{\beta}(\nada)$ by our assumption \eqref{eq:dec}, and 
$\ystrat(u) \sse T$ by definition of $T$.
Then the induction hypothesis yields that 
$\ystrat(u) \sse (F\rst{T})^{\beta}(\nada)$, and so we have 
$\ystrat(u) \sse (F\rst{T})^{\beta}(\nada) \cap T$.
But since $\ystrat$ is a winning strategy, and $u$ is a winning position for 
$\eloise$ in $\UG_{F}$ by Claim~\ref{p:unfg}(i), $\ystrat(u)$ is a
legitimate move for $\eloise$, and so we have $u \in F(\ystrat(u))$.
Thus by monotonicity of $F$ we obtain $u \in 
F((F\rst{T})^{\beta}(\nada) \cap T)$, and since $u \in T$ 
by assumption, this means that $u \in (F\rst{T})^{\beta+1}(\nada)$ as 
required.
\end{proof}

We now turn to the specific case where we consider the least fixed point of a 
functional $F$ which is induced by some formula $\phi(p) \in 
\mu_{D}\oslang$ on some LTS $\bbS$.   
By Proposition \ref{p:unfg} and Fact~\ref{f:adeqmu}, $\eloise$ has a winning
strategy in  $\egame(\mu p.\phi(p),\bbS)@(\mu p.\phi(p),s)$ if and only if she
has a winning strategy in $\UG_{F}@s$ too, where $F \isdef  \phi_{p}^{\bbS}$ is the 
monotone functional defined by $\phi(p)$. 
The next Proposition  makes this correspondence explicit when $\oslang =
\foe$. 


First, we need to introduce some auxiliary concepts and notations.
Given  a winning strategy   $\ystrat$ for $\eloise$ in $\egame(\mu p. \phi,
\bbS)@(\mu p. \phi,s)$, we denote by $B(\ystrat)$ the set of all finite 
$\ystrat$-guided, possibly partial, matches in  $\egame(\psi,\bbS)@(\psi,s)$ in
which no position of the form $(\nu q. \psi, r)$ is visited. Let 
$\ystrat$ be a positional winning strategies for $\eloise$ in $\UG_F@s$ and 
$\ystrat'$ a winning strategy  for her in $\egame(\mu p.\phi,\bbS)@(\mu p.\phi,
s)$. We call $\ystrat$ and  $\ystrat'$  \emph{compatible} if
each point in $T_{\ystrat,s}$ occurs on some path belonging to $B(\ystrat')$. 

\begin{proposition}\label{p:unfold=evalgame2}
Let $\phi(p) \in \mu_{D}\foe{p}$ and $s \in \ext{\mu p.\phi}^{\bbS}$. 
Then there is a descending winning strategy for $\eloise$ in $\UG_F@s$ 
compatible with a  winning strategy for $\eloise$ in 
$\egame(\mu p.\phi,\bbS)@(\mu p.\phi,s)$
\end{proposition}

\begin{proof}
Let $F \isdef  \phi_{p}^{\bbS}$ be the monotone functional defined by $\phi(p)$.
From $s \in \ext{\mu p.\phi}^{\bbS}$, we get that $s \in \LFP. F$.  
Applying Proposition~\ref{p:unfg} to the fact that $s \in \LFP. F$ yields that 
$\eloise$ has a descending winning strategy $\ystrat: S \to \pow{(S)}$ in $\UG_{F}@s$. 
We define $\eloise$'s strategies $\ystrat'$  in
$\egame(\mu p.\phi,\bbS)@(\mu p.\phi,s)$, and $\ystrat^*$ in $\UG_{F}@s$ as follows:
\begin{enumerate}
\item 
In the evaluation games $\egame$, after the initial automatic move, the position 
of the match is $(\phi,s)$; there $\eloise$ first plays her positional winning 
strategy $\ystrat_s$ from $\egame(\phi(p),\bbS[p \mapsto \ystrat(s)])@(\phi(p),
s)$, and we define her move $\ystrat^*(s)$ in the unfolding game $\UG$ as the 
set of all nodes $t \in \ystrat(s)$ such that there is a $\ystrat_s$-guided 
match in $B(\ystrat_s)$ whose last position is $(p,t)$.
\item 
Each time a position $(p,t)$ is reached in the evaluation games $\egame$, 
distinguisg cases:
\begin{enumerate}
\item 
if $t \in \win_{\eloise}(\UG_{F})$, then $\eloise$ continues with the positional
winning strategy $\ystrat_t$ from $\egame(\phi(p),\bbS[p \mapsto
\ystrat(t)])@(\phi(p),t)$, and we define her move $\ystrat^*(t)$ in $\UG$ as the 
set of all nodes $w \in \ystrat(t)$ such that there is a $\ystrat_t$-guided
match in $B(\ystrat_t)$ whose last position $(p,w)$;
\item 
if $t \notin \win_{\eloise}(\UG_{F})$, then $\eloise$ continues with a random 
positional strategy and we define $\ystrat^*(t)\isdef \nada$.
\end{enumerate}
\item 
For any position $(p,t)$ that was not reached in the previous steps, $\eloise$ 
sets $\ystrat^*(t)\isdef \nada$.
\end{enumerate}
By construction, $\ystrat'$ and  $\ystrat^*$ are compatible. 
Moreover, $\ystrat^*(t) \subseteq \ystrat(t)$, for $t \in S$, meaning that 
$\ystrat^*$ is descending.
We verify that both $\ystrat'$ and  $\ystrat^*$ are actually winning strategies
for $\eloise$ in the respective games.

First of all, observe that every position of the form $(p,t)$ reached during
a $\ystrat'$-guided match, we have $t \in \win_{\eloise}(\UG_{F})$. 
This can be proved by induction on the number of position of the form $(p,t)$
visited during an $\ystrat'$-guided match. 
For the inductive step, assume $w \in \win_{\eloise}(\UG_{F})$. 
Hence $\ystrat_w$ is winning for $\eloise$ in  $\egame(\phi,\bbS[p \mapsto 
  \ystrat(w)])@(\phi,w)$. 
This means that if a position of the form $(p, t)$ is reached, the variable $p$ 
must be true at $t$ in the model $\bbS[p \mapsto \ystrat(w)]$, meaning that it
belongs to the set $\ystrat(w)$.
By assumption $\ystrat$ is a winning strategy for $\eloise$ in $\UG_F$, and 
therefore any element of $\ystrat(w)$ is again a member of the set 
$\win_{\eloise}(\UG_{F})$. 

Finally, let $\pi$ be an arbitrary $\ystrat'$-guided match of $\egame(\phi,
\bbS[p \mapsto \ystrat(w)])@(\phi,w)$. 
We verify that $\pi$ is winning for $\eloise$. 
First observe that since $\ystrat$ is winning for her in $\UG_F@s$, the fixpoint 
variable $p$ is unfolded only finitely many times during $\pi$. 
Let $(p,t)$ be the last basic position in $\pi$ where $p$ occurs. 
Then from now on $\ystrat'$ and $\ystrat_t$ coincide, yielding  that the match is 
winning for $\eloise$. 

We finally verify that $\ystrat^*$ is winning for $\eloise$ in the unfolding
game $\UG_F@s$. 
First of all, since $\ystrat'$ is winning, $B(\ystrat')$ does not contain an
infinite ascending chain of $\ystrat'$-guided matches, and thence any
$\ystrat^*$-guided match in $\UG_{F}@s$ is finite. 
It therefore remains to verify that for every $\ystrat^*$-guided match $\pi$ 
in $\UG_{F}@s$ such that $\last(\pi)$ is an $\eloise$ position, she can always
move.
We do it by induction on the length of a $\ystrat^*$-guided match. 
At each step, we use compatibility and thus keep track of the corresponding 
position in the evaluation game $\egame(\mu p.\phi,\bbS)@(\mu p.\phi,s)$. 
The initial position for her is $s \in S$. Notice that $\ystrat^*(s) = 
\ystrat(s) \cap B(\xi')$ and therefore $\ystrat'$ corresponds to $\ystrat_s$
on $\egame(\phi(p),\bbS[p \mapsto \ystrat^*(s)])@(\phi(p),s)$ and it is 
therefore winning for $\eloise$.
In particular, this means  that $s \in F(\ystrat^*(s))$.
Hence, as initial move, $\eloise$ is allowed to play $\ystrat^*(s)$. 
Moreover any subsequent choice for $\abelard$ is such that there is a winning
match $\pi \in B(\xi_s)$ for $\eloise$ such that $\last(\pi)=(p,w)$. 
For the induction step, assume $\abelard$ has chosen $ t \in \ystrat^*(w)$, 
where $\ystrat^*(w) = \ystrat(w) \cap B(\xi')$, $\ystrat'$ corresponds to the
winning strategy $\ystrat_w$ on $\egame(\phi(p),\bbS[p \mapsto 
\ystrat^*(w)])@(\phi(p),w)$, and there is a winning match $\pi \in B(\xi_w)$ 
for $\eloise$ such that $\last(\pi)=(p,w)$. 
By construction, $\ystrat'$ corresponds to the winning strategy $\ystrat_t$ 
for $\eloise$ on $\egame(\phi(p),\bbS[p \mapsto \ystrat(t)])@(\phi(p),t)$. 
Because $\ystrat^*(t)= \ystrat(s) \cap B(\xi')$, $\ystrat_t$ is also winning
for her in $\egame(\phi(p),\bbS[p \mapsto \ystrat^*(t)])@(\phi(p),t)$, meaning
that $s \in F(\ystrat^*(s))$. 
The move $\ystrat^*(t)$ is therefore admissible, and any subsequent choice for
$\abelard$ is such that there is a winning match $\pi \in B(\xi_t)$ for 
$\eloise$ with $\last(\pi)=(p,w)$.
\end{proof}


\begin{proofof}{Proposition~\ref{p:keyfix}(1)}
%
Let $\bbS$ be an LTS and $\phi(p) \in \mu_{D}\ofoe{p}$. 

The right-to-left direction of \eqref{eq:foe-d} being proved by 
Proposition \ref{p:rstfix}, we check the left-to-right direction.
We first verify that winning strategies in evaluation games for noetherian 
fixpoint formulas naturally induce bundles. 
More precisely:
\medskip

\textsc{Claim.}
Let $B^\bbS(\ystrat)$ be the projection of $B(\ystrat)$ on $S$, that is the set
of all paths in $\bbS$ that are a projection on $S$ of a $\ystrat$-guided 
(partial) match in $B(\ystrat)$. Then $B^\bbS(\ystrat)$ is a bundle.
\medskip

\begin{pfclaim}
Assume towards a contradiction that $B^\bbS(\ystrat)$ contains an infinite 
ascending chain $\pi_{0} \sqsubset \pi_{1} \sqsubset \cdots$. 
Let $\pi$ be the limit of this chain and consider the set of elements in 
$B(\ystrat)$ that, projected on $S$, are prefixes of $\pi$. 
By  K\"{o}nig's Lemma, this set contains an infinite ascending chain whose 
limit is an infinite $\ystrat$-guided match in $\egame(\mu p. \phi,\bbS)$
which starts at $(\mu p. \phi,s)$, and of which $\pi$ is the projection on $S$.
By definition of $B(\ystrat)$,  the highest bound variable of $\mu p. \phi$ 
that gets unravelled infinitely often in $\rho$ is a $\mu$-variable, meaning 
that the match is winning for $\abelard$, a contradiction.
\end{pfclaim}

Assume that $s \in \ext{\mu p.\phi}^{\bbS}$, and let $F \isdef  \phi_{p}^{\bbS}$ be the monotone functional defined by $\phi(p)$.
By Proposition \ref{p:unfold=evalgame2}, $\eloise$ has a winning strategy 
$\ystrat'$ in $\egame(\mu p.\phi,\bbS)@(\mu p.\phi,s)$ compatible with a descending winning strategy $\ystrat$ in $\UG_{F}@s$.
By Proposition~\ref{p:afmc-2}, we obtain 
that $s \in \LFP. F\rst{T_{\ystrat,s}}$. 
Because of compatibility, every node in $T_{\ystrat, s}$ occurs on some path of 
$B(\ystrat')$. 
From the Claim
we known that $B^\bbS(\ystrat')$ is a bundle, meaning that  $T_{\ystrat, s}$ is
noetherian as required. 
\end{proofof}


\section{Expressiveness modulo bisimilarity}
   \label{sec:expresso}

In this Section we use the tools developed in the previous parts to prove the
main results of the paper on expressiveness modulo bisimilarity, viz., 
Theorem~\ref{t:11} stating
\begin{eqnarray}
   \AFMC &\equiv& \nmso / {\bis} 
   \label{eq:z11}
\\[1mm] \mucML &\equiv& \wmso / {\bis} 
   \label{eq:z12}
\end{eqnarray}

\begin{proofof}{Theorem~\ref{t:11}}
The structure of the proof is the same for the statements \eqref{eq:z11} and
\eqref{eq:z12}.
In both cases, we will need three steps to establish a link between the modal 
language on the left hand side of the equation to the bisimulation-invariant 
fragment of the second-order logic on the right hand side.

The first step is to connect the fragments $\AFMC$ and $\mucML$ of the modal 
$\mu$-calculus to, respectively, the weak and the continuous-weak automata for
first-order logic without equality.
That is, in Theorem~\ref{t:mlaut} below we prove the following:
\begin{eqnarray}
   \AFMC &\equiv& \AutW(\ofo)
   \label{eq:z21}
\\[1mm] \mucML &\equiv&  \AutWC(\ofo)
   \label{eq:z22}
\end{eqnarray}

Second, the main observations that we shall make in this section is that 
\begin{eqnarray}
\AutW(\ofo)  &\equiv& \AutW(\ofoe)/{\bis} 
   \label{eq:z31}
\\[1mm] 
\AutWC(\ofo) &\equiv&  \AutWC(\ofoei)/ {\bis} 
   \label{eq:z32}
\end{eqnarray}
That is, for \eqref{eq:z31} we shall see in Theorem~\ref{t:bi-aut} below that a
weak $\ofoe$-automaton $\bbA$ is bisimulation invariant iff it is equivalent to
a weak $\ofo$-automaton $\bbA^{\tmod}$ (effectively obtained from $\bbA$);
and similarly for \eqref{eq:z32}.

Finally, we use the automata-theoretic characterisations of $\nmso$ and $\wmso$
that we obtained in earlier sections:
\begin{eqnarray}
\AutW(\ofoe)   &\equiv&  \nmso
   \label{eq:z41}
\\[1mm] 
\AutWC(\ofoei) &\equiv&  \wmso 
   \label{eq:z42}
\end{eqnarray}

Then it is obvious that the equation \eqref{eq:z11} follows from \eqref{eq:z21},
\eqref{eq:z31} and \eqref{eq:z41}, while similarly 
\eqref{eq:z12} follows from \eqref{eq:z22}, \eqref{eq:z32} and \eqref{eq:z42}.
\end{proofof}

It is left to prove the equations \eqref{eq:z21} and \eqref{eq:z22}, and 
\eqref{eq:z31} and  \eqref{eq:z32}; this we will take care of in the two 
subsections below.

\subsection{Automata for $\AFMC$ and $\mucML$}

In this subsection we consider the automata corresponding to the 
continuous and the alternation-free $\mu$-calculus. 
That is, we verify the equations \eqref{eq:z21} and \eqref{eq:z22}.

\begin{theorem}
\label{t:mlaut}
\begin{enumerate}
\item 
There is an effective construction transforming a formula $\phi \in \muML$ into
an equivalent automaton in $\Aut(\ofo)$, and vice versa.
\item 
There is an effective construction transforming a formula $\phi \in \AFMC$ into
an equivalent automaton in $\AutW(\ofo)$, and vice versa.
\item 
There is an effective construction transforming a formula $\phi \in \mucML$ into
an equivalent automaton in $\AutWC(\ofo)$, and vice versa.
\end{enumerate}
\end{theorem}

\begin{proof}
In each of these cases the direction from left to right is easy to verify,
so we omit details.
For the opposite direction, we focus on the hardest case, that is, we will 
only prove that $\AutWC(\ofo) \leq \mucML$.
By Theorem~\ref{t:autofor} it suffices to show that $\mu_{C}\ofo \leq \mucML$,
and we will in fact provide a direct, inductively defined, truth-preserving 
translation $(\cdot)^{t}$ from $\mu_{C}\ofo(\pprop)$ to $\mucML(\pprop)$.
Inductively we will ensure that, for every set $\qprop \sse \pprop$:
\begin{equation}
\label{eq:zz1}
\phi \in \cont{\mu\ofo}{\qprop} 
\text{ implies } \phi^{t} \in 
\cont{\muML}{\qprop} 
\end{equation}
and that the dual property holds for cocontinuity.

Most of the clauses of the definition of the translation $(\cdot)^{t}$ are 
completely standard: for the atomic clause we take $p^{t} \isdef p$ and
$(\neg p)^{t} \isdef \neg p)$, for the boolean connectives we define 
$(\phi_{0}\lor\phi_{1})^{t} \isdef \phi_{0}^{t} \lor \phi_{1}^{t}$ and 
$(\phi_{0}\land\phi_{1})^{t} \isdef \phi_{0}^{t} \land \phi_{1}^{t}$, and 
for the fixpoint operators we take $(\mu p. \phi)^{t} \isdef \mu p. \phi^{t}$ 
and $(\nu p. \phi)^{t} \isdef \nu p. \phi^{t}$ ---
to see that the latter clauses indeed provide formulas in $\mucML$ we use
\eqref{eq:zz1} and its dual.
In all of these cases it is easy to show that \eqref{eq:zz1} holds (or remains
true, in the inductive cases).

The only interesting case is where $\phi$ is of the form 
$\nxt{\al}(\phi_{1},\ldots,\phi_{n})$.
By definition of the language $\mu_{C}\ofo$ we may assume that 
$\al(a_{1},\ldots,a_{n}) \in \cont{\ofo(A)}{B}$, where 
$A = \{ a_{1},\ldots,a_{n} \}$ and $B = \{ a_{1}, \ldots, a_{k} \}$,
that for each $1 \leq i \leq k$ the formula $\phi_{i}$ belongs to the set 
$\cont{\mu_{C}\ofo}{\qprop}$ and that for each $k+1\leq j \leq n$ the formula 
$\phi_{j}$ is $\qprop$-free.
It follows by the induction hypothesis 
that $\phi_{l} \equiv \phi_{l}^{t} \in \mucML$ for each $l$, 
that $\phi_{i}^{t} \in \cont{\muML}{\qprop}$ for each $1 \leq i \leq k$,
and that  the formula $\phi_{j}^{t}$ is $\qprop$-free for each $k+1\leq j \leq n$.
The key observation is now that by Theorem~\ref{t:osnf-cont} we may without 
loss of generality assume that $\al$ is in \emph{normal form}; that is, a 
disjunction of formulas of the form $\al_{\Sigma} = \posdgbnfofo{\Sigma}{\Pi}$,
where every $\Sigma$ and $\Pi$ is a subset of $\pow (A)$, $B \cap \bigcup\Pi =
\nada$ for every $\Pi$, and 
\[
\posdgbnfofo{\Sigma}{\Pi} \isdef 
\bigwedge_{S\in\Sigma} \exists x \bigwedge_{a \in S} a(x) 
\;\land\; \forall x \bigvee_{S\in\Pi} \bigwedge_{a \in S} a(x) 
\]
We now define
\[\begin{array}{rll}
\bigvee \big(\nxt{\al_{\Sigma}}(\ol{\phi})\big)^{t} & \isdef &
    \bigwedge_{S\in\Sigma} \Diamond \bigwedge_{a_{l} \in S} \phi_{l}^{t}
     \;\land\; 
     \Box \bigvee_{S\in\Pi} \bigwedge_{a_{j} \in S} \phi_{j}^{t} 
\\ \phi^{t} & \isdef & {\bigvee}_{\Sigma} \big(\nxt{\al_{\Sigma}}(\ol{\phi})\big)^{t}
\end{array}\]
It is then obvious that $\phi$ and $\phi^{t}$ are equivalent, so it remains to 
verify \eqref{eq:zz1}.
But this is immediate by the observation that all formulas $\phi_{j}^{t}$
in the scope of the $\Box$ are associated with an $a_{j}$ belonging to a set 
$S \sse A$ that has an empty intersection with the set $B$; that is, each 
$a_{j}$ belongs to the set $\{ a_{k+1}, \ldots, a_{n}\}$ and so $\phi_{j}^{t}$
is $\qprop$-free.
\end{proof}

\subsection{Bisimulation invariance, one step at a time}
\label{ss:bisinv}

In this subsection we will show how the bisimulation invariance results in this
paper can be proved by automata-theoretic means.
Following Janin \& Walukiewicz~\cite{Jan96}, 
we will define a construction that, for $\oslang \in \{{\ofoe},{\ofoei}\}$, 
transforms an arbitrary $\oslang$-automaton $\bbA$ into an $\ofo$-automaton 
$\bbA^{\tmod}$ such that $\bbA$ is bisimulation invariant iff it is equivalent
to $\bbA^{\tmod}$.
In addition, we will make sure that this transformation preserves both the
weakness and the continuity condition.
The operation $(\cdot)^{\tmod}$ is completely determined by the following 
translation at the one-step level.

\begin{definition}
Recall from Theorem~\ref{t:osnf} that any formula in ${\ofoe}^+(A)$ is 
equivalent to a disjunction of formulas of the form 
$\posdbnfofoe{\vlist{T}}{\Sigma}$, whereas any formula in ${\ofoei}^+(A)$ is 
equivalent to a disjunction of formulas of the form 
$\posdbnfolque{\vlist{T}}{\Pi}{\Sigma}$. 
Based on these normal forms, for both one-step languages $\oslang={\ofoe}$ and 
$\oslang={\ofoei}$, we define the translation 
$(\cdot)^{\tmod} : {\oslang}^+(A) \to \ofo^+(A)$ by setting
\[
\left.\begin{array}{l}
   \Big( \posdbnfofoe{\vlist{T}}{\Sigma} \Big)^{\tmod} 
\\ \Big( \posdbnfolque{\vlist{T}}{\Pi}{\Sigma} \Big)^{\tmod} 
\end{array}\right\}
\df \bigwedge_{i} \exists x_i. \tau^+_{T_i}(x_i) \land 
\forall x. \bigvee_{S\in\Sigma} \tau^+_S(x),
\]
and for $\al = \bigvee_{i} \al_{i}$ we define $\al^{\tmod} \df \bigvee 
\al_{i}^{\tmod}$.
\end{definition}

\noindent
This definition propagates to the level of automata in the obvious way.

\begin{definition}
Let $\oslang\in \{{\ofoe},{\ofoei}\}$ be a one-step language.
Given an automaton $\bbA = \tup{A,\tmap,\pmap,a_{I}}$ in $\Aut(\oslang)$, define 
the automaton $\bbA^{\tmod} \df \tup{A,\tmap^{\tmod},\pmap,a_{I}}$ in 
$\Aut(\ofo)$ by putting, for each $(a,c) \in A \times C$:
\[
\tmap^{\tmod}(a,c) \df (\tmap(a,c))^{\tmod}.
\]
\end{definition}

The main result of this section is the theorem below.
For its formulation, recall that $\bbS^{\om}$ is the $\om$-unravelling of 
the model $\bbS$ (as defined in the preliminaries).
As an immediate corollary of this result, we see that \eqref{eq:z31} and
\eqref{eq:z32} hold indeed.

\begin{theorem}
\label{t:bi-aut}
Let $\oslang\in \{{\ofoe},{\ofoei}\}$ be a one-step language and let $\bbA$ be an
$\oslang$-automaton.

\begin{enumerate}
\item
The automata $\bbA$ and $\bbA^{\tmod}$ are related as follows, for every model $\bbS$:
\begin{equation}
\label{eq:crux}
\bbA^{\tmod} \text{ accepts } \bbS \text{ iff } \bbA \text{ accepts
} \bbS^{\om}.
\end{equation}
\item
The automaton $\bbA$ is bisimulation invariant iff $\bbA \equiv \bbA^{\tmod}$.
\item
If $\bbA\in \AutW(\oslang)$ then $\bbA^{\tmod}\in \AutW(\ofo)$, and 
if $\bbA\in \AutWC(\ofoei)$ then $\bbA^{\tmod}\in \AutWC(\ofo)$.
\end{enumerate}
\end{theorem}

The remainder of this section is devoted to the proof of Theorem~\ref{t:bi-aut}.
The key proposition is the following observation on the one-step translation,
that we take from the companion paper~\cite{carr:mode18}.

\begin{proposition}
\label{p-1P}
Let $\oslang\in \{{\ofoe},{\ofoei}\}$.
For every one-step model $(D,V)$ and every $\al \in \oslang^+(A)$ we have
\begin{equation}
\label{eq-1cr}
(D,V) \models \alpha^{\tmod} \text{ iff } (D\times \om,V_\pi) \models \alpha,
\end{equation}
where $V_{\pi}$ 
 is the induced valuation given by 
$V_{\pi}(a) \df \{ (d,k) \mid d \in V(a), k\in\omega\}$.
\end{proposition}

\begin{proofof}{Theorem~\ref{t:bi-aut}}
The proof of the first part is based on a fairly routine comparison, based on
Proposition~\ref{p-1P}, of the acceptance games $\mathcal{A}(\bbA^{\tmod},\bbS)$
and $\mathcal{A}(\bbA,\bbS^{\om})$.
(In a slightly more general setting, the details of this proof can be found 
in~\cite{Venxx}.)

For part~2, the direction from right to left is immediate by Theorem \ref{t:mlaut}.
The opposite direction follows from the following equivalences, where we use
the bisimilarity of $\bbS$ and $\bbS^{\om}$ (Fact~\ref{prop:tree_unrav}):
\begin{align*}
\bbA \text{ accepts } \bbS
  & \text{ iff } \bbA \text{ accepts } \bbS^{\om}
  & \tag{$\bbA$ bisimulation invariant}
\\ & \text{ iff } \bbA^{\tmod} \text{ accepts } \bbS
  & \tag{equivalence~\eqref{eq:crux}}
\end{align*}

It remains to be checked that the construction $(\cdot )^{\tmod}$, which has
been defined for arbitrary automata in $\Aut(\oslang)$, transforms 
both $\wmso$-automata and $\nmso$-automata into automata of the right kind.
This can be verified by a straightforward inspection at the one-step level.
\end{proofof}

\begin{remark}{\rm
In fact, we are dealing here with an instantiation of a more general phenomenon 
that is essentially coalgebraic in nature.
In~\cite{Venxx} it is proved that if $\oslang$ and $\oslang'$ are two one-step
languages that are connected by a translation $(\cdot )^{\tmod}: \oslang' \to 
\oslang$ satisfying a condition similar to \eqref{eq-1cr}, then we find that 
$\Aut(\oslang)$ corresponds to the bisimulation-invariant fragment of 
$\Aut(\oslang')$: $\Aut(\oslang) \equiv \Aut(\oslang')/{\bis}$.
This subsection can be generalized to prove similar results relating
$\AutW(\oslang)$ to $\AutW(\oslang')$, and $\AutWC(\oslang)$ to 
$\AutWC(\oslang')$.
}\end{remark}

{\small 
\bibliographystyle{plain}
\bibliography{references/logic,references/extra}
}

\end{document}